\documentclass[manuscript, screen, authorversion, nonacm, preprint]{jair}

\setcopyright{cc}
\setcopyright{none}
\copyrightyear{2025}
\acmDOI{10.1613/jair.1.xxxxx}

\JAIRAE{Insert JAIR AE Name}
\JAIRTrack{} %
\acmVolume{4}
\acmArticle{6}
\acmMonth{8}
\acmYear{2025}

\usepackage{amsmath}
\usepackage{amsthm}
\usepackage{xspace}
\usepackage{cleveref}
\usepackage[inline]{enumitem}
\usepackage{mathtools}
\usepackage{natbib}

\usepackage{caption}
\usepackage{subcaption}

\usepackage{pifont}
\renewcommand{\checkmark}{\ding{51}}

\usepackage{comment}
\usepackage{doi}

\usepackage{tabularx}
\usepackage{booktabs}
\usepackage{multirow}
\newcolumntype{Y}{>{\centering\arraybackslash}X}
\newcolumntype{M}[1]{>{\centering\arraybackslash}m{#1}}

\usepackage{placeins}

\usepackage{tikz}
\usetikzlibrary{decorations.pathreplacing}

\newcommand{\N}{\ensuremath{\mathbb{N}}}

\newcommand{\probName}[1]{\textsc{#1}\xspace}

\newcommand{\CCDCs}[2][]{\probName{#2-CC#1DC}}

\newcommand{\DCDCs}[2][]{\probName{#2-DC#1DC}}
\newcommand{\CCAC}[2][]{\probName{#2-Constructive Control by #1Adding Projects}}
\newcommand{\CCACs}[2][]{\probName{#2-CC#1AC}}

\newcommand{\DCACs}[2][]{\probName{#2-DC#1AC}}

\newcommand{\RXthreeC}{\probName{Restricted Exact Cover by~$3$-Sets}}
\newcommand{\RXthreeCs}{\probName{RX$3$C}}

\newcommand{\voters}{\ensuremath{V}} %
\newcommand{\numVoters}{\ensuremath{n}} %

\newcommand{\projects}{\ensuremath{P}} %
\newcommand{\numProjects}{\ensuremath{m}} %

\newcommand{\budget}{\ensuremath{B}} %
\newcommand{\cost}{\ensuremath{\operatorname{cost}}} %

\newcommand{\wFn}{\ensuremath{\omega}} %
\newcommand{\numDeleted}{\ensuremath{r}} %

\newcommand{\score}[1]{\ensuremath{\operatorname{score}_{\text{#1}}}}

\newcommand{\ruleName}[1]{\textsc{#1}\xspace}
\newcommand{\greedyAV}{\ruleName{GreedyAV}}
\newcommand{\AV}{\ruleName{AV}}
\newcommand{\greedyCost}{\ruleName{GreedyCost}}
\newcommand{\AVcost}{\ruleName{AV/c}}
\newcommand{\phragmen}{\ruleName{Phragmén}}

\newcommand{\equalShares}{\ruleName{Equal-Shares}}

\renewcommand{\P}{\textsf{P}\xspace}
\newcommand{\NP}{\textsf{NP}\xspace}
\newcommand{\NPh}{\NP-hard\xspace}
\newcommand{\NPhness}{\NP-hardness\xspace}
\newcommand{\NPc}{\NP-comp\-lete\xspace}
\newcommand{\W}[1][1]{\textsf{W[#1]}}
\newcommand{\Wh}[1][1]{\W[#1]-hard\xspace}
\newcommand{\Whness}[1][1]{\Wh{}ness\xspace}

\newcommand{\Oh}[1]{{\mathcal{O}(#1)}}

\newtheorem{example}{Example}
\newtheorem{theorem}{Theorem}

\newtheorem{definition}{Definition}
\newtheorem{claim}{Claim}
\newtheorem{observation}{Observation}
\newtheorem{remark}{Remark}
\newtheorem{proposition}{Proposition}

\Crefname{observation}{Observation}{Observations}
\Crefname{claim}{Claim}{Claims}
\Crefname{remark}{Remark}{Remarks}
\newenvironment{claimproof}[1]{\textsc{Proof.}\hspace{0.15cm}#1}{~\hfill~$\blacktriangleleft$\smallskip}

\newcommand{\Yes}{\textsc{Yes}\xspace}
\newcommand{\No}{\textsc{No}\xspace}

\begin{document}

\fancypagestyle{firstpagestyle}{
    \fancyhf{}%
    \fancyfoot[RO,LE]{}%
  }%

\fancyfoot{}

\title{Participatory Budgeting Project Strength via Candidate Control}
\titlenote{An extended abstract of this work has been published in the Proceedings of the 34th International Joint Conference on Artificial Intelligence~\cite{FaliszewskiJKPSSS2025}.}

\author{Piotr Faliszewski}
\affiliation{
  \institution{AGH University of Krakow}
  \city{Krakow}
  \country{Poland}}
\email{faliszew@agh.edu.pl}
\orcid{0000-0002-0332-4364}

\author{Łukasz Janeczko}
\affiliation{
  \institution{AGH University of Krakow}
  \city{Krakow}
  \country{Poland}}
\email{ljaneczk@agh.edu.pl}
\orcid{0000-0002-4596-8109}

\author{Dušan Knop}
\affiliation{
  \institution{Czech Technical University in Prague}
  \city{Prague}
  \country{Czech Republic}}
\email{dusan.knop@fit.cvut.cz}
\orcid{0000-0003-2588-5709}

\author{Jan Pokorný}
\affiliation{
  \institution{Czech Technical University in Prague}
  \city{Prague}
  \country{Czech Republic}}
\email{jan.pokorny@fit.cvut.cz}
\orcid{0000-0003-3164-0791}

\author{Šimon Schierreich}
\affiliation{
  \institution{AGH University of Krakow}
  \city{Krakow}
  \country{Poland}
}
\affiliation{
  \institution{Czech Technical University in Prague}
  \city{Prague}
  \country{Czech Republic}
}
\email{schiesim@fit.cvut.cz}
\orcid{0000-0001-8901-1942}

\author{Mateusz Słuszniak}
\affiliation{
  \institution{AGH University of Krakow}
  \city{Krakow}
  \country{Poland}}
\email{msluszniak1@gmail.com}
\orcid{0009-0006-9639-7197}

\author{Krzysztof Sornat}
\affiliation{
  \institution{AGH University of Krakow}
  \city{Krakow}
  \country{Poland}}
\email{sornat@agh.edu.pl}
\orcid{0000-0001-7450-4269}

\renewcommand{\shortauthors}{Faliszewski et al.}

\begin{abstract}
    We study the complexity of candidate control in participatory budgeting elections. The goal of constructive candidate control is to ensure that a given candidate wins by either adding or deleting candidates from the election (in the destructive setting, the goal is to prevent a given candidate from winning). We show that such control problems are \NPh to solve for many participatory budgeting voting rules, including \phragmen{} and \equalShares, but there are natural cases with polynomial-time algorithms (e.g., for the \greedyAV rule and projects with costs encoded in unary).
    We also argue that control by deleting candidates is a useful tool for assessing the performance (or, strength) of initially losing projects, and we support this view with experiments.
\end{abstract}

\maketitle

\section{Introduction}

Participatory budgeting is a recent democratic innovation where cities
allow their inhabitants to decide about a certain fraction of their
budgets~\citep{cab:j:participatory-budgeting,goe-kri-sak-ait:c:knapsack-voting,rey-mal:t:pb-survey}.
Specifically, some of the community members propose possible projects
to be implemented and, then, all the citizens get a chance to vote as
to which of them should be funded. Most commonly, such elections use
approval ballots, where people indicate which projects they would like
to see funded, and the \greedyAV{} rule, which selects the most
approved projects (subject to not exceeding the budget).  However,
there also are more advanced rules, such as
\phragmen{}~\citep{bri-fre-jan-lac:c:phragmen,los-chr-gro:c:phragmen-pb}
or
\equalShares{}~\citep{pet-sko:c:welfarism-mes,pet-pie-sko:c:pb-mes},
which produce arguably more fair---or, to be precise, more
proportional---decisions~\citep{fal-fli-pet-pie-sko-sto-szu-tal:c:pabulib}. Yet,
with more advanced rules come issues about understanding the
results. Indeed, recently
\citet{BoehmerFJPPSSS2024} have argued
that proposers whose projects were rejected may find it quite
difficult to understand the reason for this outcome. To alleviate this
problem, they introduced a number of performance measures---mostly
based on the bribery family of
problems~\citep{fal-hem-hem:j:bribery,FaliszewskiR2016}---that attempt
to answer the following question: As a proposer of a project that was
not funded, what could I have done differently to have it funded?  For
example, they ask if the project would have been funded if its cost
were lower (see also the work of \citet{bau-boe-hil:c:pb-manip}), or
if its proposer convinced more people to vote for it, or if the
proposer motivated some voters to only approve his or her project.
Similar bribery-style problems were also used to evaluate the
robustness of election
results~\citep{shi-yu-elk:c:robustness,bre-fal-kac-nie-sko-tal:c:robustness,bau-hog:c:robustness,boe-bre-fal-nie:c:counting-bribery,boe-fal-jan-kac:c:robustness-pb},
or the margin of victory for the
winners~\citep{mag-riv-she-wag:c:stv-bribery,xia:margin-of-victory}.

In this paper, we follow-up on these ideas, but using candidate
control.  The main difference is that instead of focusing on
circumstances that depend on a project's proposer (indeed, the
project's cost is his or her choice, and it is his or her choice what
support campaign he or she runs), we focus on external ones,
independent of his or her actions (such as some other projects being
submitted or not\footnote{We disregard the possibility that a proposer
  might try to discourage other people from proposing projects, albeit
  we acknowledge that this may happen. In fact, this might even be
  quite benign: A group of activists focused on making their city more
  green may discuss among themselves which projects to submit and
  which to withhold.}). We believe that looking at both types of
reasons for a project's rejection gives a more complete view of its
performance.

\paragraph{Candidate Control.}
The idea of the control-in-elections family of problems is that we are
given a description of an election, a designated candidate, and we ask
if it is possible to ensure that this candidate is a winner (in
constructive control) or ceases to be a winner (in destructive
control) by modifying the structure of the
election~\citep{BartoldiTT1992,FaliszewskiR2016}. Specifically,
researchers consider control by adding or deleting either candidates
or voters (some papers---including the one that introduced election
control~\citep{BartoldiTT1992}---also consider various forms of
arranging run-off elections as a type of control). So far, election
control was mostly studied theoretically, with a focus on the
complexity
analysis~\citep{BartoldiTT1992,hem-hem-rot:j:destructive-control,MeirPRZ2008,fen-liu-lua-zhu:j:parameterized-control,erd-fel-rot-sch:j:bucklin-fallback,Yang2019},
but some empirical results exist as
well~\citep{erd-fel-rot-sch:j:experimental-control}.  We study
candidate control in participatory budgeting, that is, we ask if it is
possible to ensure funding of a given project (or, preclude its
funding) by either adding new projects---from some a'priori known set
of projects---or by deleting them.  Our results are theoretical
and focus on the complexity of our problems, but we motivate them by
project performance analysis.
As our performance analysis is
based on control by deleting projects, we pay most attention to
results regarding this variant of control, and we include the addition
case for the sake of completeness and to be in sync with the preceding
literature.

\paragraph{Performance Analysis.}
Let us now discuss how one could use control by deleting candidates to
analyze the performance of projects in participatory budgeting (we
will use the terms projects and candidates interchangeably, e.g.,
using ``candidates'' in the names of control problems). Consider a
participatory budgeting election and some not-funded project~$p$. One
basic measure of its performance is the smallest number of other
projects that have to be removed from the election for~$p$ to be
funded. The lower this number, the closer was the project to winning:
Indeed, perhaps some proposers only managed to submit their projects
in the last minute and it was possible that they would have missed the
deadline, or some projects were close to be removed from the election
due to formal reasons, but the city officials were not strict in this
regard. However, it is more likely that such issues would affect
cheaper projects than the more expensive ones---which, likely, had
more careful proposers---so instead of asking for a smallest set of
projects to delete, we may ask for a set with the lowest total cost.

Another way of using control by deleting projects to assess a project's
performance is to use a probabilistic approach, along the lines of the
one taken by
\citet{boe-bre-fal-nie:c:counting-bribery}, \citet{boe-fal-jan-kac:c:robustness-pb},
and \citet{bau-hog:c:robustness} for bribery: We ask for the
probability that project~$p$ is funded assuming that a random subset
of projects (of a given cardinality) is removed. The higher it is, and
the lower is the number of removed project, the closer was project~$p$
to winning.

A different interpretation of the above measures is that instead of
thinking that some projects ``barely made it'' to participate in the
election, we learn how many projects performed better than~$p$. The
more projects we need to delete to have~$p$ funded (or, to have~$p$
funded with sufficiently high probability) the more projects can be
seen as critically stronger than~$p$.

Finally, we can use candidate control as a way of assessing rivalry
between projects. For example, if project~$p$ has a much higher
probability of being funded after deleting a random set of projects
under the condition that some other project~$q$ was included in this
set, then we can view~$q$ as a strong rival of~$p$.

\paragraph{Contributions.}

Our results are theoretical and regard the complexity of
candidate control for four well-known voting rules, depending on how
the costs of the projects are encoded (either in binary, or in unary,
or as unit costs, which means that each project costs the same
amount). We show the overview of our results in
\Cref{fig:complexityOverview}. We mention that all our~$\NP$-hardness
results also imply~$\#\P$-hardness for respective problems where we
ask for a number of solutions (e.g., number of ways in which we can
ensure a victory of a given project by deleting a given number of
others). This is interesting because solving such problems is
necessary for estimating the probability that a project wins if a
given randomly-selected set of projects is deleted.

Additionally, we provide an experimental analysis of real-world
participatory budgeting instances, showing what one can learn about
them via candidate control. E.g., we pursue rivalry analysis
by adopting the notion of the Banzhaf index.

\begin{table*}
    \centering
    \def\arraystretch{1.3}
    \Crefname{theorem}{Thm.}{Thm.}
    \Crefname{observation}{Obs.}{Obs.}
    \Crefname{proposition}{Prop.}{Prop.}
    \renewcommand{\crefpairconjunction}{,~}
    \renewcommand{\NPh}{{\NP-c} }
    \begin{tabular}{c|M{17mm}M{17mm}|M{17mm}M{17mm}|M{17mm}M{17mm}}\toprule
            & \multicolumn{2}{c|}{\texttt{unit}} & \multicolumn{2}{c|}{\texttt{unary}} & \multicolumn{2}{c}{\texttt{binary}} \\
            & Del & Add & Del & Add & Del & Add \\\midrule
         \greedyAV
            & \phantom{xx}\P{}$^\ddagger$\phantom{xx} {\small[\Cref{thm:AV:CCDC:P:ifUnitPricesAndBinaryWeights}]}
            & \phantom{xx}\P{}$^\ddagger$\phantom{xx} {\small[\Cref{thm:AV:CCAC:P:ifUnaryPrices}]}
            & \phantom{xx}\P{}$^\ddagger$\phantom{xx} {\small[\Cref{thm:AV:CCDC:P:ifUnaryPrices}]}
            & \phantom{xx}\P{}$^\ddagger$\phantom{xx} {\small[\Cref{thm:AV:CCAC:P:ifUnaryPrices}]}
            & \NPh{} {\small[\Cref{thm:AV:CCDC:NPh}]}
            & \NPh{} {\small[\Cref{thm:AV:CCAC:NPh}]}\\\hline
         \greedyCost
            & \phantom{xx}\P{}$^\ddagger$\phantom{xx} {\small[\Cref{thm:AV:CCDC:P:ifUnitPricesAndBinaryWeights}]}
            & \phantom{xx}\P{}$^\ddagger$\phantom{xx} {\small[\Cref{thm:AV:CCAC:P:ifUnaryPrices}]}
            & \phantom{xx}\P{}$^\ddagger$\phantom{xx} {\small[\Cref{thm:AV:CCDC:P:ifUnaryPrices}]}
            & \phantom{xx}\P{}$^\ddagger$\phantom{xx} {\small[\Cref{thm:AV:CCAC:P:ifUnaryPrices}]}
            & \NPh{} {\small[\Cref{thm:AV:CCDC:NPh}]}
            & \NPh{} {\small[\Cref{thm:AV:CCAC:NPh}]} \\\hline
         \phragmen
            & \NPh {\small[\Cref{thm:Phragmen:CCDC:NPh,thm:Phragmen:DCDC:NPh}]}
            & \NPh {\small[\Cref{thm:Phragmen:CCAC:NPh,thm:Phragmen:DCAC:NPh}]}
            & \NPh {\small[\Cref{thm:Phragmen:CCDC:NPh,thm:Phragmen:DCDC:NPh}]}
            & \NPh {\small[\Cref{thm:Phragmen:CCAC:NPh,thm:Phragmen:DCAC:NPh}]}
            & \NPh{} {\small[\Cref{thm:Phragmen:CCDC:NPh,thm:Phragmen:DCDC:NPh}]}
            & \NPh{} {\small[\Cref{thm:Phragmen:CCAC:NPh,thm:Phragmen:DCAC:NPh}]} \\\hline
         \equalShares
            & \NPh {\small[\Cref{thm:MES:CCDC:NPc}]}
            & \NPh {\small[\Cref{thm:MES:CCAC:NPc}]}
            & \NPh {\small[\Cref{thm:MES:CCDC:NPc}]}
            & \NPh {\small[\Cref{thm:MES:CCAC:NPc}]}
            & \NPh {\small[\Cref{thm:MES:CCDC:NPc}]}
            & \NPh {\small[\Cref{thm:MES:CCAC:NPc}]} \\\bottomrule
    \end{tabular}
    \caption{A basic overview of our complexity results. In the first column, we list the rules we are interested in. All the remaining columns contain the complexity classification of our problem in one of the three variants: by \texttt{unit}, we mean that all input projects are of the same price, \texttt{unary} stands for cases where the costs are of size polynomial in~$\numVoters+\numProjects$, and \texttt{binary} applies for the variant where costs need to be encoded in binary (and hence can be exponential in~$\numVoters$ and~$\numProjects$). By {\normalfont Del} ({\normalfont Add}, respectively), we mean that the control operation is project deletion (addition). The complexity classification is the same for both constructive and destructive objectives. Results marked with~$\ddagger$ hold even if the control is weighted and projects' weights are encoded in binary.}
    \label{fig:complexityOverview}
\end{table*}

\section{Preliminaries}

An instance of \emph{participatory budgeting} (PB) is a triple~$E = (\projects,\voters,\budget)$, where~$\projects=\{p_1,\ldots,p_m\}$ is a set of \emph{projects},~$\voters=\{v_1,\ldots,v_n\}$ is a set of voters, and~$\budget\in\N$ is an available budget. Each voter~$v\in\voters$ is associated with an \emph{approval set}~$A(v)\subseteq P$, which is the set of projects they approve. For each project~$p\in\projects$, we know its price~$\cost(p) \in \N_{\geq 1}, \cost(p) \leq B$, for which this project can be implemented.  We extend this notation from a single project~$p$ to a set of projects~$S\subseteq P$ and set~$\cost(S) = \sum_{p\in S} \cost(p)$. We say that~$S\subseteq P$ is a set of~$\budget$-\emph{feasible} projects if~$\cost(S) \leq \budget$.

A PB rule is a function~$f\colon E\to 2^\projects$ that, for a given PB instance, outputs a~$\budget$-feasible subset of projects. Note that we assume the rules to be resolute, which can be easily ensured by incorporating some tie-breaking order into them. Let~$W=f(E)$ for some rule~$f$ and some PB instance~$E$. We say that projects from~$W$ are \emph{selected} or, equivalently, \emph{funded}. The projects not in~$W$ are called \emph{losing}.%

In this work, we consider four different participatory budgeting rules. Each of these rules starts with an empty set~$W$ and sequentially, in rounds, extends~$W$ with additional projects unless the budget~$\budget$ is exhausted or all the projects were processed. More formally, given a PB instance~$E=(\projects,\voters,\budget)$, the rules of our interest works as follows:

\begin{description}[leftmargin=0.33cm]
    \item[\greedyAV.] We define the score of a project~$p\in\projects$ as the number of voters approving~$p$; formally~$\score{\AV}(p) = |\{v\in\voters\mid p\in A(v)\}|$. The \greedyAV rule then processes the projects in non-increasing order according to their scores (with ties resolved according to a given tie-breaking order). If the rule can afford the currently processed project~$p$, i.e.,~$\cost(W) + \cost(p) \leq \budget$, then it includes~$p$ in~$W$. Otherwise, the rule continues with the next project. The rule terminates when all projects were processed.

    \item[\greedyCost.] This rule is very similar to the \greedyAV rule. The only difference is in the order in which the projects are processed. Specifically, the score of a project~$p\in\projects$ under \greedyCost rule is~$\score{\AVcost}(p) = \score{\AV}(p)/\cost(p)$. The process is then identical to the \greedyAV rule.

    \item[\phragmen.] This rule is conceptually different from the two above. Here, each voter starts with an empty virtual bank account, which is, in a continuous manner, increased by one unit of money per unit of time. Once there is a project~$p\in\projects \setminus W$ such that the sum of balances of voters approving~$p$ is exactly~$\cost(p)$, the current round ends, and the rule performs several steps. First, it includes the project~$p$ into the set of funded projects. Next, it sets the balance of all voters approving~$p$ to zero. Finally, the rule removes all projects~$p'\in\projects\setminus W$ such that~$\cost(W)+\cost(p') > \budget$. Then, the rule continues with the next round. The rule terminates when there is no remaining project in~$P$.

    \item[\equalShares.] Our last rule is also based on the idea of virtual bank accounts. However, this time, the budget is proportionally spread between the voters, meaning that each voter starts with the initial balance of~$\budget/\numVoters$, and this initial value is never increased. Again, the rule works in rounds. In each round, the rule funds a project such that the set of supporters have enough cumulative budget to fund this project, and each of them covers as small a fraction of its cost as possible. Formally, let~$b_i$ be the current balance of a voter~$v_i$. We say that a project~$p\in \projects\setminus W$ is~$q$-affordable, where~$q\in[0,1]$, if
    \[
        \sum\nolimits_{v_i\in A(p)} \min\left( b_i, q\cdot \cost(p)\right) = \cost(p).
    \]
    In each round, the rule funds a project that is~$q$-affordable for the smallest~$q$ over all projects and appropriately adjusts the balances of voters supporting~$p$. Specifically, the balance~$b_i$ of each agent is decreased by~$q_i\cdot \cost(p)$, where~$q_i = q$ if~$b_i > q\cdot \cost(p)$ and~$b_i/\cost(p)$ otherwise. The rule terminates when no affordable project exists.
\end{description}

\paragraph{Control Problems.} We focus on the control by adding or deleting projects and follow the standard notation from single- and multi-winner voting~\citep{FaliszewskiR2016}. Let~$f$ be a PB rule. In the \textsc{$f$-Constructive Control by Deleting Candidates} projects (\CCDCs{$f$}, for short), we are given an instance~$E=(\projects,\voters,\budget)$ of PB, an integer~$\numDeleted$, and a project~$p\not\in f(E)$, and our goal is to decide whether it is possible to delete at most~$\numDeleted$ projects~$D$ such that~$p\in f((\projects\setminus D,\voters,\budget))$. In the \textsc{$f$-Destructive Control by Deleting Candidates} (\DCDCs{$f$}), the project~$p$ is initially not funded, and the goal is to decide whether we can delete at most~$\numDeleted$ project so that the project~$p$ gets funded.

In control by adding projects, there are two disjoint sets of projects:~$\projects$ is a set of standard projects, and~$Q$ is a set of spoiler projects. The rule does not initially assume the spoiler projects. The question here is whether we can find at most~$\numDeleted$ spoiler projects such that once we add them to the instance, the project~$p$ is (in the case of \CCACs{$f$}) or is not (in the case of \DCACs{$f$}) funded by the rule~$f$.

In our algorithmic results, we are sometimes interested in the weighted variant of the above-defined problems. Under this consideration, each project~$p$ is additionally associated with its weight~$\wFn(p)$, and the goal is to decide whether a set~$D$ of projects securing our goal exists such that~$\sum_{p\in D} \wFn(p) \leq \numDeleted$. We indicate the weighted variant by adding a dollar sign in front of the operation type. For example, the weighted variant of \CCDCs{$f$} is referred to as \CCDCs[\$]{$f$}.
Even though the weighted variant might seem unnatural at first glance, the motivation for it is two-fold. First, it is studied in similar literature on control in elections~\citep{FaliszewskiR2016}. Second, and more importantly, if we set the weight of each project equal to its cost, we can use a hypothetical algorithm for the weighted variant to find a set of the lowest total costs that secures our goal, which is one of the proposed performance measures.

\paragraph{Computational Problems.}

Most hardness proofs in this work rely on different reductions from the same problem---\RXthreeC. This is an \NPc~\citep{Gonzalez1985} variant of the well-known \probName{Exact Cover By 3-Sets} problem and is defined as follows.

\begin{definition}\label{def:RX3C}
    In the \RXthreeC problem (\RXthreeCs), we are given a universe~$U$ of~$3N$ elements~$u_1,\ldots,u_{3N}$ together with a family~$\mathcal{S}$ of~$3N$ size-$3$ subsets~${S_1,\ldots,S_{3N}\subset U}$ such that every element~$u_i\in U$ appears in exactly~$3$ subsets of~$\mathcal{S}$. The goal is to decide whether~$N$ subsets of~$\mathcal{S}$ exist so that they form an exact cover of~$U$.
\end{definition}

The reduction of \citet{Gonzalez1985} additionally shows that the problem is \NPc even if each pair of subsets of~$\mathcal{S}$ intersects in at most one element.

\section{Deleting Projects}

We start with the complexity picture of both constructive and destructive control by deleting projects. This operation is arguably more natural for real-life instances and is also assumed in our experimental results. Before we present our results, let us illustrate the concept of control by deleting projects using a toy example.

\begin{example}
    Assume an instance with three projects~$c_1$,~$c_2$, and~$p$. The project~$c_1$ is approved by three voters, the project~$c_2$ by two voters, and the project~$p$ by a single voter. The costs of the projects are~$\cost(c_1) = 1$,~$\cost(c_2) = 2$, and~$\cost(p) = 1$. The total budget is~$\budget=2$. Under the \greedyAV rule, the project~$c_1$ is assumed first, and since its cost is smaller than the budget, this project gets funded. Next, the rule assumes the project~$c_2$, but this project costs more than the remaining budget. Lastly, project~$p$ is considered and eventually also funded. If we remove the project~$c_1$, the project~$c_2$ gets funded and exhausts the budget, and therefore, the project~$p$ cannot be funded. Hence, removing~$c_1$ from the instance is a successful destructive control that prevents~$p$ from winning and, simultaneously, a successful control that makes~$c_2$ a winning project.
\end{example}

\subsection{\greedyAV and \greedyCost}

We start with the \greedyAV rule. In our first result, we show that both constructive and destructive control are computationally intractable, even if the instance is unweighted. Maybe surprisingly, this hardness result holds even if the instance contains only two agents.

\begin{theorem}
    \label{thm:AV:DCDC:NPh}\label{thm:AV:CCDC:NPh}
    \label{thm:AV:DCDC:pNPh:voters}\label{thm:AV:CCDC:pNPh:voters}
    Both \CCDCs{\greedyAV} and \DCDCs{\greedyAV} are \NPc, even if~$|\voters|=2$.
\end{theorem}
\begin{proof}
    It is easy to see that both problems are in \NP: we guess the projects to delete, and then, in polynomial time, we simulate the rule to get the set of funded projects. From this set, we can directly verify the correctness of the guess by checking whether our distinguished candidate~$p$ is (in the constructive case) or is not (in the destructive case) between the funded projects.

    To show \NPhness, we give a reduction from the \RXthreeC problem (see \Cref{def:RX3C} for a formal definition). In the following, we focus on the constructive variant, and at the end of the proof, we show how to tweak the construction to also work for the destructive variant.

    The general idea of the construction is that the projects are in one-to-one correspondence with the sets, and using project costs, we encode which elements are covered by each set. This we achieve by having project costs as a~$3N$ digit-length number in base~$4$. A project~$p_j$ corresponding to a subset~$S_j = \{u_{i_1},u_{i_2},u_{i_3}\}\in \mathcal{S}$,~$i_1 < i_2 < i_3$, then has cost of the form
    \[
        \underset{3N}{0}00\cdots00\underset{i_3}{1}00\cdots00\underset{i_2}100\cdots00\underset{i_1}100\cdots\underset{2}0\underset{1}0,
    \]
    where~$i$-th digit of this cost has value one if and only if the element~$u_i$ belongs to~$S_j$. All other digits are always zero. Next, we set the budget and the cost of our distinguished project~$p$ so that it gets funded only if the deleted projects correspond to subsets that form an exact cover in~$\mathcal{I}$.

    Formally, given an instance~$\mathcal{I}=(U,\mathcal{S})$, we construct an instance~$\mathcal{J}$ of \CCDCs{\greedyAV} as follows. For each set~$S_j\in\mathcal{S}$,~$S_j = \{u_{i_1},u_{i_2},u_{i_3}\}$, we create a \emph{set-project}~$p_j$ with cost~$1\cdot4^{i_1} + 1\cdot4^{i_2} + 1\cdot4^{i_3}$. Next, we add our \emph{distinguished project}~$p$ and~$N+1$ \emph{guard-projects}~$g_1,\ldots,g_{N+1}$. The cost of the distinguished project~$p$ is~$\sum_{i=1}^{3N} 1\cdot 4^i$ and for every~$i\in[N+1]$, we set~$\cost(g_i) = \cost(p) + 1$. That is, the guard projects are only one unit more expensive compared to our distinguished project. This is important to ensure that the budget left after we delete some set-projects is exactly the cost of~$p$.

    The set of voters consists of just two voters,~$ v_1$ and~$ v_2$. The first voter,~$v_1$, approves all projects except for~$p$. The second voter,~$v_2$, approves only the set-projects. Such a preference profile secures that, regardless of the tie-breaking order, the method first processes all set-projects, then all guard-projects, and only as the last possibility, the method processes the distinguished project~$p$.

    To complete the construction, we set~$\budget = \sum_{i=1}^{3N} 3\cdot 4^i$ and~$\numDeleted = N$. Observe that the number of guard-projects is by one greater than the number of projects we are allowed to delete; hence, no solution may delete all guard-projects. The budget is selected so that if we do not delete any project, all the set-projects are funded, and the budget is exhausted after the last set-project is taken into consideration by the rule. On the other hand, if we remove~$N$ projects corresponding to subsets forming an exact cover in~$\mathcal{I}$, the remaining budget after the rule processes all the set-projects will be exactly the cost of~$p$.

    First, let us show that~$p$ is indeed initially not funded. The~$\score{AV}$ of every set-project~$p_j$,~$j\in[3N]$, is exactly two, the~$\score{AV}$ of the guard-projects is exactly one, while the~$\score{AV}$ of our distinguished project~$p$ is zero. Therefore, all other projects are processed before project~$p$. Moreover, their total cost is~$\sum_{i=1}^{3N} 3\cdot 4^i$ due to the definition of the costs and the fact that every element~$u_i$ appears in exactly three subsets. Therefore, when the distinguished project~$p$ is processed by the rule, the budget is~$\budget - \sum_{i=1}^{3N} = \sum_{i=1}^{3N} - \sum_{i=1}^{3N} = 0$. Hence, the project~$p$ is clearly not affordable when it is assumed.%

    For left-to-right implication, let~$\mathcal{I}$ be a \Yes-instance and let~$C \subset \mathcal{S}$ be an exact cover of~$U$. We delete every set-project~$p_j$ such that~$S_j\in C$, and we claim that~$p$ is now founded, that is, the control is successful. Since~$C$ is an exact cover, we spend exactly~$\sum_{i=1}^{3N} 2\cdot 4^i$ on the set-projects. Consequently, after the last set-project is processed by the rule, the remaining budget is~$\budget - \sum_{i=1}^{3N} 2\cdot 4^i = \sum_{i=1}^{3N} 3\cdot 4^i - \sum_{i=1}^{3N} 2\cdot 4^i = \sum_{i=1}^{3N} (3-2)\cdot 4^i = \sum_{i=1}^{3N} 1\cdot 4^i$. This is one unit of money less than the cost of any guard-project. Therefore, no guard-project is funded, and once the rule processes~$p$, the remaining budget is still~$\sum_{i=1}^{3N} 1\cdot 4^i$
    which is the cost of~$p$, so~$p$ is funded.
    Consequently,~$\mathcal{J}$ is also a \Yes-instance and~$C$ is a solution.

    In the opposite direction, let~$\mathcal{J}$ be a \Yes-instance and~$D$ be a set of deleted projects such that when projects in~$D$ are removed from the instance,~$p$ becomes funded. First, we prove an auxiliary claim that shows that the cost of no~$N-1$ set-projects sums up to exactly~$\cost(p)$.

    \begin{claim}\label{clm:noSmallSetSumUpToCostOfp}
        There is no set~$X\subseteq[3N]$ of size at most~$N-1$ such that~$\sum_{j\in X} \cost(p_j) = \sum_{i=1}^{3N} 1\cdot 4^{i}$.
    \end{claim}
    \begin{claimproof}
        Let~$X \subseteq[3N]$ be an arbitrary set of size at most~$N-1$. The sum~$\sum_{i=1}^{3N} 1\cdot 4^{i}$ can be imagined also as a base-$4$ number with digit~$1$ on positions~$1$ up to~$3N$. By the definition of projects' costs, the cost of each set-project~$p_j$ is a base-$4$ number with digit~$1$ on exactly three positions and with digit~$0$ on all the remaining positions. Moreover, there are exactly three set-projects with digit~$1$ on position~$i$ for every~$i\in[3N]$. Therefore, by the Pigeonhole principle, we obtain that if we represent~$\sum_{j\in X} \cost(p_j)$ as a base-$4$ number, there are at least~$3$ digits with value zero, as there are~$3N$ possible positions, and the set-projects can cover at most~$3(N-1) = 3N - 3$ different positions. However, the target sum requires a value of~$1$ for every digit in base-$4$ representation.
    \end{claimproof}

    First, we show that~$D$ contains no guard-project. For the sake of contradiction, assume that it is the case, that is,~$D$ contains at least one guard-project. Consequently, we removed at most~$N-1$ set-projects. As the budget is big enough to fund all original set-projects, the rule will also fund all the remaining~$2N+1$ set-projects. Let~$\budget'$ be the budget after the last set-project is funded. If~$\budget' < \cost(p)$, we have a clear contradiction with~$D$ being a solution. Hence,~$\budget' \geq \cost(p)$. Assume first that~$\budget' = \cost(p)$. Recall that the original budget~$\budget$ is~$\sum_{i=1}^{3N} 3\cdot4^i$, which is also the sum of the costs of all set-projects. We require that the budget~$\budget'$ is equal to~$\cost(p)$; that is, it must hold for the removed projects that the sum of their costs is~$\cost(p) = \sum_{i=1}^{3N} 1\cdot 4^{i}$, which is not possible by \Cref{clm:noSmallSetSumUpToCostOfp}. Hence,~$\budget' > \cost(p)$, but in this case, the first guard-project not in~$D$ is funded before~$p$ (such a guard-project always exists since~$\numDeleted = N$ and there are~$N+1$ of them), and the distinguished project~$p$ is never funded. This contradicts that~$D$ is a solution; thus, no guard-project is part of~$D$.

    Now, we know that~$D$ consists solely of set-projects. Let~$\budget'$ be the budget after the last non-deleted set-project is funded and let~$\cost(D) = \sum_{p_j\in D} \cost(p_j)$. In the instance with removed projects in~$D$, the budget~$\budget'$ is exactly~$\sum_{i=1}^{3N} 3^\cdot 4^i - (\sum_{i=1}^{3N} 3\cdot 4^i - \cost(D)) = \cost(D)$. Now, if~$\cost(D) < \cost(p)$, we clearly cannot fund~$p$. Hence,~$\cost(D) \geq \cost(p)$. If~$\cost(D) > \cost(p)$, the rule funds a guard-project before~$p$, which exhausts the budget, and~$p$ is not funded. Consequently,~$\cost(D) = \cost(p)$. We set~$C = \{S_j \mid p_j \in D\}$ and claim that it is a solution for~$\mathcal{I}$. It is clearly of the correct size, so suppose that it is not an exact cover; that is, there is an element~$u_i$ such that there are at least two sets~$S_j,S_{j'}\in C$ containing~$u_i$. But then the sum of the corresponding set-projects would be strictly greater than~$\cost(p)$, which is not possible. Therefore,~$C$ is a solution for~$\mathcal{I}$ and~$\mathcal{I}$ is a \Yes-instance.

    For the \emph{destructive} variant, we tweak the construction as follows. Instead of having~$N+1$ guard-projects, we have only one guard-project~$g$ with~$\cost(g) = (\sum_{i=1}^{3N} 1\cdot4^i) + 1$. Next, we set~$\cost(p) = 1$ and we increase the budget by one to~$\budget= (\sum_{i=1}^{3N} 3\cdot 4^i) + 1$. Now, observe that the distinguished project~$p$ is indeed funded, as the \greedyAV rule first processes all set-project whose total cost is~$\sum_{i=1}^{3N} 3\cdot 4^i$. Hence, when the guard-project~$g$ is processed by the rule, the remaining budget is exactly one. Therefore,~$g$ cannot be funded, and there is enough budget to fund~$p$.

    For correctness, assume that~$\mathcal{I}$ is a \Yes-instance and let~$C\subset \mathcal{S}$ be an exact cover of~$U$. We remove every set-project corresponding to a set of~$C$. In the reduced instance, since~$C$ is an exact cover, the remaining budget when~$g$ is processed by the \greedyAV rule is exactly~$(\sum_{i=1}^{3N} 1\cdot4^i) + 1 = \cost(g)$. Therefore,~$g$ is funded, and the budget is exhausted. In the opposite direction, first, observe that the guard-project~$g$ is never part of the solution. For the sake of contradiction, assume that~$g\in D$, for some solution~$D$. Then we have~$\cost(\projects\setminus D) \leq \budget= (\sum_{i=1}^{3N} 3\cdot 4^i) + 1$, and hence, we can fund all projects including~$p$. This contradicts that~$D$ is a solution. Next, let~$B'$ be the remaining budget in the round when the \greedyAV rule processes the guard-project~$g$. If~$B' > \cost(g)$, then~$B' \geq \cost(g) + 1$ and also~$B' \geq \cost(g) + \cost(p)$ and hence, the \greedyAV rule funds both~$g$ and~$p$. Consequently,~$B' \leq \cost(g)$. If~$B' < \cost(g)$, then the guard-project~$g$ is not funded. However, the project~$p$ is clearly funded as~$B' > 1$ since the initial budget is greater than the sum of the costs of all set-projects. Therefore,~$B'$ must be equal to~$\cost(g)$, which can be achieved if and only if the cost of all funded set-projects is exactly~$\sum_{i=i}^{3N} 2\cdot 4^i$. By \Cref{clm:noSmallSetSumUpToCostOfp}, the size of every solution~$D$ is exactly~$N$, and, by the definition of set-projects' costs, the sets corresponding to set-projects in~$D$ form an exact cover of~$U$, finishing the proof.
\end{proof}

The hardness construction from \Cref{thm:AV:CCDC:NPh} requires prices of exponential size. That is, our problems are, from the computational complexity perspective, weakly \NPh. It is natural to ask whether \Cref{thm:AV:CCDC:NPh} can be strengthened to polynomial-size prices or if a pseudopolynomial time algorithm exists for the problem.

Before we resolve this question, we provide a simple property that facilitates argumentation about the input instances. Basically, we can ignore all projects that are, in the initial instance, processed by the \greedyAV rule after the distinguished project~$p$.

\begin{observation}
    \label{thm:AV:removeUnnecessaryProjects}
    Let~$\mathcal{I}$ be an instance of \CCDCs[\$]{\greedyAV} or \DCDCs[\$]{\greedyAV} and~$q\in\projects$ be a project which the \greedyAV rule processes after the distinguished project~$p$. Then, we can remove~$q$ from~$\mathcal{I}$ and continue with the reduced instances.
\end{observation}
\begin{proof}
    The \greedyAV rule always processes the projects in non-increasing order by the number of approvals (and with ties resolved according to a given tie-breaking rule). If we remove a project, we do not change the relative order in which the rule processes the projects. Moreover, let~$D$ be a minimum size solution and, for the sake of contradiction, let~$q\in D$. We set~$D' = D\setminus\{q\}$ and claim that~$D'$ is also a solution. Let~$\budget'$ be the budget available just before the \greedyAV rule processes the distinguished project~$p$ if we remove projects from~$D$. Since~$q$ is processed after~$p$ by the \greedyAV rule, returning~$q$ to the instance keeps~$\budget'$ the same. Therefore,~$D'$ is also a solution but is smaller than~$D$. This contradicts the fact that~$D$ is of minimum size, so~$q$ can be safely deleted from the instance.
\end{proof}

In the following result, we show that if all projects are of the same price, then even the weighted variants of our problems can be solved by a simple greedy polynomial-time algorithm.

\begin{proposition}
    \label{thm:AV:CCDC:P:ifUnitPricesAndBinaryWeights}
    \label{thm:AV:DCDC:P:ifUnitPricesAndBinaryWeights}
    If all projects are of the same price, both the \CCDCs[\$]{\greedyAV} problem and the \DCDCs[\$]{\greedyAV} problem can be solved in polynomial time, even if the projects' weights are encoded in binary.
\end{proposition}
\begin{proof}
    Let~$E$ be a PB instance reduced with respect to \Cref{thm:AV:removeUnnecessaryProjects} and~$p\in P$ be the distinguished project. If~$\wFn(\projects\setminus\{p\}) \leq \numDeleted$, we can remove all projects processed before the distinguished project~$p$. Therefore, we can directly decide on such an instance. Hence, for the rest of the proof, we assume that~$\wFn(\projects\setminus\{p\}) > \numDeleted$.
    The idea for the rest of the algorithm is very simple: we iteratively try to remove the projects in the order from the smallest to the highest weight as long as the sum of weights of the removed projects is below~$\numDeleted$ or we find a solution. Nevertheless, this approach works only due to the fact that all projects are of the same price.

    More formally, as the first step of the algorithm, we find the order~$p_1,\ldots,p_{m-1},p$ in which the \greedyAV rule processes the projects. Next, we find a permutation~$\sigma\colon [m-1]\to\projects\setminus\{p\}$ such that for every pair of~$i,j\in[m-1]$,~$i < j$, it holds that~$\wFn(\sigma(i)) \leq \wFn(\sigma(j))$ and whenever~$\wFn(\sigma(i) = p_{\ell}) = \wFn(\sigma(j) = p_{\ell'})$, it holds that~$\ell < \ell'$. That is, the permutation is a non-decreasing ordering of the projects according to their weights. Also, observe that the permutation~$\sigma$ is uniquely determined.

    The core of the algorithm is then a simple iteration. Specifically, for every~$\ell$ from one to~$m-1$, the algorithm constructs~$D_{\ell} = \bigcup_{i=1}^{\ell} \sigma(i)$, remove~$D_{\ell}$ from the instance and simulates the rule. If~$D_\ell$ is a solution for the instance and~$\wFn(D_\ell) \leq \numDeleted$, the algorithm returns \Yes. Otherwise, the algorithm continues with another~$\ell$. If no such~$\ell$ leads to a \Yes response, the algorithm returns \No. We claim that~$\mathcal{I}$ is a \Yes-instance if and only if our algorithm returns~\Yes.

    If our algorithm finds a solution, then~$\mathcal{I}$ is clearly a \Yes-instance. Therefore, let us focus on the opposite direction. So, assume that~$\mathcal{I}$ is a \Yes-instance, let~$D\subseteq \projects\setminus\{p\}$ be a solution of size~$\ell=|D|$ and of the minimum weight. If~$D = D_\ell$, then our algorithm necessarily returned \Yes by its definition. Hence, let~$D\not=D_\ell$. Observe that~$\wFn(D) \geq \wFn(D_\ell)$, as~$D_\ell$ is a minimum-weight~$\ell$-sized subset of~$\projects\setminus\{p\}$. We now show an auxiliary claim that states that if all projects are of the same price, we can replace projects in a solution with projects of smaller weight.

    \begin{claim}\label{clm:AV:CCDC:P:ifUnitPricesAndBinaryWeights:replace}
        Let~$X\subseteq \projects\setminus\{p\}$ be a solution,~$x\in X$ be a project in this solution, and~$y\in (\projects\setminus\{p\})\setminus X$ be a projects outside of the solution such that~$\wFn(y) \leq \wFn(x)$. Then, also a set~$Y = (X\setminus\{x\})\cup\{y\}$ is a solution.
    \end{claim}
    \begin{claimproof}
        The size of~$Y$ is~$\wFn(X) - \wFn(x) + \wFn(y)$ and since~$\wFn(y) \leq \wFn(x)$, it holds that~$\wFn(Y) \leq \wFn(X) \leq \numDeleted$. Hence,~$Y$ is of the correct weight. It remains to show that if we delete~$Y$, we will indeed achieve our control goal. Let~$y = p_j$ and~$x = p_{j'}$. First, assume that~$j < j'$, that is, the project~$y$ is processed by the \greedyAV rule before the project~$x$. Let~$\beta$ be the sum of prices of all projects between~$y$ and~$x$ and~$\budget'$ be the budget with~$X\setminus\{x\}$ removed at the beginning of the round in which the \greedyAV rule processes~$x$.
    \end{claimproof}

    First, let~$\wFn(D) > \wFn(D_\ell)$. In this, we get an immediate contradiction with~$D$ being a minimum weights solution, as by \Cref{clm:AV:CCDC:P:ifUnitPricesAndBinaryWeights:replace}, we can replace projects in~$D$ with projects in~$D_\ell$ to obtain a smaller weight solution. Therefore, it must hold that~$\wFn(D) = \wFn(D_\ell)$. But now, we can again use \Cref{clm:AV:CCDC:P:ifUnitPricesAndBinaryWeights:replace} to replace each project that is in~$D$ and not in~$D_\ell$ with a project of the same or smaller weight that is in~$D_\ell$ but not in~$D$. It follows that~$D_\ell$ is also a solution. But our algorithm checked~$D_\ell$; that is, by definition of the algorithm, it returned \Yes. This concludes its correctness.

    For the running time, the \greedyAV rule can be simulated in polynomial time, a permutation~$\sigma$ can be found in~$\Oh{\numProjects\cdot\log\numProjects\cdot\log\numDeleted}$ time, and the set~$D_\ell$ for every~$\ell\in[m-1]$ can be found in~$\Oh{\log r}$ time. As there are~$\Oh{m}$ iterations, obviously, the algorithm can be performed in polynomial time, even if the project weights are encoded in binary. This finishes the proof.
\end{proof}

Indeed, \Cref{thm:AV:CCDC:P:ifUnitPricesAndBinaryWeights} does not provide a complete answer to the raised question about the computational complexity status of our problems, as it requires all projects to be of the same size. But what happens if the prices are not all the same but still only polynomial in the number of projects and the number of voters? Our following result shows that, even for such instances, there is a polynomial-time algorithm. This time, the algorithm is more involved and is based on a dynamic programming technique.

\begin{theorem}
    \label{thm:AV:CCDC:P:ifUnaryPrices}
    \label{thm:AV:DCDC:P:ifUnaryPrices}
    If the costs of the projects are encoded in unary, both \CCDCs[\$]{\greedyAV} and \DCDCs[\$]{\greedyAV} can be solved in polynomial time, even if the projects' weights are encoded in binary.
\end{theorem}
\begin{proof}
    \newcommand{\DP}{\ensuremath{\operatorname{DP}}}
    Our algorithm is based on the dynamic programming approach. We first present an algorithm for constructive control and later show what needs to be changed also to solve the destructive variant of control by deleting projects.

    We suppose that~$p_1,\ldots,p_{\numProjects-1},p$ is the order in which the rule processes the projects in the original instance; that is,~$p_1$ is processed the first and~$p_\ell$ is processed just before the distinguished project~$p$. Note that the assumption that~$p$ is processed the last is not in contradiction with the fact that our algorithm works for any tie-breaking order, as by \Cref{thm:AV:removeUnnecessaryProjects}, we can remove all projects that the \greedyAV rule processes after the distinguished project~$p$. Moreover, we can remove all projects with~$\cost(p_j) > \budget$, as such projects cannot be afforded under any condition.

    The central part of the algorithm is to compute a dynamic programming table~$\DP[ j, \beta ]$, where
    \begin{itemize}
        \item~$j\in[m-1]$ is an index of the last processed project, and
        \item~$\beta\in[\budget]$ is a desired remaining budget just before the rule processes a project~$p_{j+1}$.
    \end{itemize}
    We call the pair~$(j,\beta)$ a \emph{signature}. For every signature, the dynamic programming table stores the weight of a minimum-weight \emph{partial solution}~$D_{j,\beta} \subseteq \{p_1,\ldots,p_j\}$ such that if the projects from~$D_{j,\beta}$ are removed, the remaining budget just before the \greedyAV rule processed project~$p_{j+1}$ is exactly~$\beta$. If no such partial solution exists, we store some large value~$\infty$ (in fact, for the algorithm, it is enough to store value greater than~$\numDeleted$).

    The computation is then defined as follows. The basic step is when~$j=1$. Here, we just decide whether~$p_1$ needs to be deleted or not, based on the required remaining budget~$\beta$. Formally, we set the dynamic programming table as follows:
    \[
        \DP[ 1, \beta ] = \begin{cases}
            0 & \text{if } \beta = \budget - \cost(p_1),\\
            \wFn(p_1) & \text{if } \beta = \budget\text{, and}\\
            \infty & \text{otherwise.}
        \end{cases}
    \]

    For every~$j \in [2,m-1]$, the computation of the algorithm is defined as follows.
    \[
        \DP[ j, \beta ] = \begin{cases}
            \min\{ \DP[j-1,\beta] + \wFn(p_j),
            \DP[ j-1, \beta + \cost( p_j ) ] \} &\\
                \phantom{xxxxxxxxxx}
                \text{if }  \DP[j-1,\beta] + \wFn(p_j)\leq \numDeleted \text{ and }%
                \beta + \cost( p_j ) \leq \budget, &
                \\
            \DP[j-1,\beta] &\\
                \phantom{xxxxxxxxxx}
                \text{if } \beta + \cost( p_j ) > \budget,\text{ and}&
                \\
            \DP[ j-1, \beta + \cost( p_j ) ] &\\
                \phantom{xxxxxxxxxx}
                \text{if }  \DP[j-1,\beta] + \wFn(p_j) > \numDeleted . %
        \end{cases}
    \]
    The first case corresponds to a situation where we need to decide whether to include~$p_j$ in a solution or not. In the second case, we cannot fund~$p_j$ anyway, so we need not delete it, and we are only interested in whether the same budget can be achieved just before~$p_j$ is processed. In the last case, we cannot delete~$p_j$, as it would exceed the budget. Hence, the rule must fund~$p_j$.

    In the following claims, we show that the computation is indeed correct. First, we prove that a corresponding partial solution exists whenever the stored value is not infinite.

    \begin{claim}
        Whenever~$\DP[ j, \beta ] = w < \infty$, there exists a set~$D\subseteq \{p_1,\ldots,p_j\}$ such that~$\wFn(D) = w$ and~$\cost(W) = \budget - \beta$, where~$W = \greedyAV(\{p_1,\ldots,p_j\}\setminus D,\voters,\budget)$.
    \end{claim}
    \begin{claimproof}
        We show the property by induction over~$j$. First, let~$j=1$. In this case, by the definition of the computation,~$w$ is either zero or~$\wFn(p_1)$. In the former case,~$\beta = \budget - \cost(p_1)$ and the only possible solution is~$D = \emptyset$. If we do not remove any project, we obtain~$\greedyAV(\{p_1\},\voters,\budget) = \{p_1\} = S$, since~$\cost(p_1) \leq \budget$. Moreover,~$\cost(W) = \budget - \cost(p_1) = \beta$. In the latter case,~$\beta = \budget$. The only possible way to secure this is to delete all projects, that is, setting~$D = \{p_1\}$. Clearly,~$\wFn(D) = \wFn(p_1) = w$, and hence, the computation is correct also in this case.

        Now, let~$j > 1$ and assume that the claim holds for~$j-1$. Let~$\DP[ j, \beta ] = w < \infty$. By the definition of the computation, the value~$w$ is copied to~$\DP[j,\beta]$ as the result of one of the following computation:~$\DP[j-1,\beta]$,~$\DP[j-1,\beta] + \omega(p_j)$, or~$\DP[j-1,\beta+\cost(p_j)] = w$.

        First, let it be the case that~$\DP[j,\beta] = w$ because~$\DP[j-1,\beta] = w$. Then, it must hold that~$\beta + \cost(p_j) > B$. By the induction hypothesis, there exists a set~$D' \subseteq \{p_1,\ldots,p_{j-1}\}$ such that~$\wFn(D') = w$ and~$\cost(\greedyAV(\{p_1,\ldots,p_{j-1}\}\setminus D',V,\budget)) = B - \beta$. Since~$\beta+\cost(p_j) > B$, there is not enough budget to fund~$p_j$ when the \greedyAV rule processes this project. Therefore, also~$\cost(\greedyAV(\{p_1,\ldots,p_{j}\}\setminus D',V,\budget)) = B - \beta$ and~$D'$ is a sought solution.

        Next, assume that~$\DP[j,\beta] = w$ because~$\DP[j-1,\beta] + \wFn(p_j) = w$. Then necessarily~$\DP[j-1,\beta] = w - \wFn(p_j) < \infty$ and therefore, by the induction hypothesis, there exists a set~$D'\subseteq \{p_1,\ldots,p_{j-1}\}$ of weight~$w - \wFn(p_j)$ such that~$\cost(\greedyAV(\{p_1,\ldots,p_{j-1}\}\setminus D',\voters,\budget)) = B - \beta$. If we take~$D = D'\cup\{p_j\}$, we have that~$\wFn(D) = \wFn(D') + \wFn(p_j) = w - \wFn(p_j) + \wFn(p_j) = w$ and~$\cost(\greedyAV(\{p_1,\ldots,p_{j-1},p_j\}\setminus D,\voters,\budget)) = \cost(\greedyAV(\{p_1,\ldots,p_{j-1},p_j\}\setminus (D'\cup\{p_j\}),\voters,\budget)) = \cost(\greedyAV(\{p_1,\ldots,p_{j-1}\}\setminus D',\voters,\budget)) = B - \beta$. Consequently,~$D$ is a solution.

        Finally, assume that we set~$\DP[j,\beta] = w$ because~$\DP[j-1,\beta+\cost(p_j)] = w$. Clearly,~$\beta+\cost(p_j) \leq B$, and therefore, once the project~$p_j$ is processed, the rule funds it. By the induction hypothesis, there is a set~$D'$ such that~$\wFn(D') = w$ and~$\cost(\greedyAV(\{p_1,\ldots,p_{j-1}\}\setminus D'),\voters,\budget) = B - (\beta + \cost(p_j))$. If we take~$D = D'$, we obtain that~$\wFn(D) = w$. Moreover, the project~$p_j$ is not deleted and, hence, it is processed by the rule. As there is enough budget, the rule funds~$p_j$ and we get that~$\cost(\greedyAV(\{p_1,\ldots,p_{j-1},p_j\}\setminus D,\voters,\budget)) = \cost(\greedyAV(\{p_1,\ldots,p_{j-1}\}\setminus D'),\voters,\budget) + \cost(p_j) = B - (\beta + \cost(p_j)) + \cost(p_j) = B - \beta - \cost(p_j) + \cost(p_j) = B - \beta$.
    \end{claimproof}

    To complete the correctness of the algorithm, in the following claim, we show that the stored weight is indeed the minimum possible.

    \begin{claim}
        For every~$j\in[m-1]$ and every~$D \subseteq \{p_1,\ldots,p_j\}$ such that~$\wFn(D) \leq \numDeleted$ and~$\cost(\greedyAV(\{p_1,\ldots,p_j\}\setminus D,\voters,\budget)) \leq B$, we have that~$\DP[ j, \beta ] \leq \wFn(D)$.
    \end{claim}
    \begin{claimproof}
        Again, we show correctness by induction over~$j$. First, let~$j=1$. There are only two possible subsets of~$\{p_1\}$, namely~$\emptyset$ and~$\{p_1\}$. In the former situation, the project~$p_1$ is clearly funded by the \greedyAV rule. Hence, the remaining budget~$\beta$ after~$p_1$ is processed by the rule is~$B - \cost(p_1)$. By the definition, in such case, we set~$\DP[1,\beta] = 0$, which is clearly at most~$\wFn(\emptyset) = 0$. In the latter situation, the project~$p_1$ is deleted; hence, the rule does not fund any projects, and we have the remaining budget~$\beta$ after the rule processes~$p_1$ is~$B$. For such~$\beta$, we set~$\DP[1, B] = \wFn(\{p_1\})$, which is again clearly at most~$\wFn(D) = \wFn(\{p_1\})$. This finishes the basic step of the induction.

        Next, let~$j> 1$ and let the claim hold for~$j-1$,~$D$ be a subset of~$\{p_1,\ldots,p_j\}$ satisfying the conditions in the statement of the claim, and let~$\beta = \budget - \cost(\greedyAV(\{p_1,\ldots,p_j\}\setminus D,\voters,\budget))$. We distinguish two cases: either~$p_j\in D$ or~$p_j\not\in D$. Let us start with the former situation and let~$D' = D \setminus \{p_j\}$. Since~$p_j$ is removed from the instance, it holds that~$\budget - \cost(\greedyAV(\{p_1,\ldots,p_{j-1}\}\setminus D',\voters,\budget)) = \beta$. Now, by the induction hypothesis, the cell~$\DP[j-1,\beta]$ is at most~$\wFn(D) - \wFn(p_j)$. By the definition of the computation, we store in~$\DP[j,\beta]$ value at most~$\DP[j-1,\beta] + \wFn(p_j) \leq \wFn(D) - \wFn(p_j) + \wFn(p_j) = \wFn(D)$. Therefore, in this situation, the claim is true.
        Now, assume that~$p_j\not\in D$. First, suppose that~$p_j$ is funded by the rule. Then the remaining budget before~$p_j$ is processed with~$D$ removed is~$\beta + \cost(p_j)$. By the induction hypothesis, in~$\DP[j-1,\beta + \cost(p_j)]$, we store at most~$\wFn(D)$. Moreover, the cell~$\DP[j,\beta]$ is at most~$\DP[j-1,\beta+\cost(p_j)]$ and henceforth at most~$\wFn(D)$. If the rule does not fund~$p_j$, then the remaining budget when~$p_j$ is processed by the rule is lower than~$\cost(p_j)$. In this case, we set~$\DP[j,\beta] = \DP[j-1,\beta]$, where the right side is at most~$\wFn(D)$ by the induction hypothesis. This finishes the proof.
    \end{claimproof}

    Once all the cells of the dynamic programming table~$\DP$ are correctly computed, we can decide the instances. Specifically, we return \Yes whenever there exists a cell~$\DP[m-1,\beta]$, where~$\beta \geq \cost(p)$, such that~$\DP[m-1,\beta] \leq \numDeleted$. Assume that such a cell exists. Then, by definition, there is a partial solution~$D_{m-1,\beta}$ such that the sum of weights of all projects in~$D_{m-1,\beta}$ is at most~$\numDeleted$ and, after projects in~$D_{m-1,\beta}$ are deleted from the instance, the budget remaining just before~$p$ is processed is~$\beta > \cost(p)$. Hence, the \greedyAV rule necessarily funds~$p$. If no such cell of the table~$\DP$ exists, we return~\No.

    The dynamic programming table has~$\Oh{m\cdot \budget}$ cells, and each cell can be computed in time~$\Oh{\log(\numDeleted)}$. The final check can be done in~$\Oh{\budget}$ time; therefore, the overall running time of the algorithm is~$\Oh{(m\cdot\budget)\cdot\log(\numDeleted) + \budget}$, which, assuming the budget is encoded in unary, is clearly a polynomial-time algorithm, even if the projects' weights are encoded in binary. %

    For the destructive control, the computation of the dynamic programming table is identical. The only difference is in the final step. Now, we are interested whether there exists~$\beta < \cost(p)$ such that~$\DP[m-1,\beta] \leq \numDeleted$. If this is the case, there exists a partial solution~$D_{m-1,\beta}$ such that the sum of weights of the projects in~$D_{m-1,\beta}$ is at most~$\numDeleted$ and, just before the \greedyAV rule processes the project~$p$, the remaining budget is less than~$\cost(p)$. Therefore, the \greedyAV rule does not fund~$p$. If such a cell exists, we return \Yes. Otherwise, we return \No. The running time of the algorithm remains the same as in the constructive case.%
\end{proof}

In the following remark, we discuss the relation between results we developed for the \greedyAV and their counterparts for the \greedyCost rule.

\begin{remark}
    The hardness construction provided in \Cref{thm:AV:CCDC:NPh} and the algorithms from \Cref{thm:AV:CCDC:P:ifUnitPricesAndBinaryWeights} and \Cref{thm:AV:CCAC:P:ifUnaryPrices} also work for the \greedyCost rule. For the former result, one can observe that even under the \greedyCost rule, the property that the set-projects are processed before the guard-projects, and that the guard-projects are processed before the distinguished project~$p$ is preserved. This comes from the fact that set-projects are approved by exactly two voters, guard-projects by exactly one voter, and~$p$ by no voter. Moreover, no set-project is more expensive than any guard-project. The algorithms require that the relative ordering of the projects is not affected by deletions. This is clearly preserved also in the \greedyCost rule.%
\end{remark}

The outcome of this section is that, even though in the full generality control under the \greedyAV and \greedyCost rules is intractable, for practical instance, we have an efficient algorithm. This implies that for real-life elections, performance measures based on project control can be computed efficiently. In fact, we even implemented the algorithm from \Cref{thm:AV:CCDC:P:ifUnaryPrices} and used it in our experimental analysis.

\subsection{\phragmen}

Now, we turn our attention to the \phragmen rule. Here, the situation is significantly less positive. Specifically, in the following theorem, we show that for this rule, it is \NPh to decide whether successful control is possible, even if the input instance is unweighted and all projects are of the same cost.

The idea behind the construction is that we have one project for every set of \RXthreeCs instance and a lot of direct competitors of the distinguished project~$p$. The set-projects have significantly higher support than~$p$ or its competitors, and, moreover, the competitors of~$p$ share their voters with the set-projects. Hence, all the set-projects are always funded before the first project of a different type may be funded. These set-projects exhausts most of the budget, and, unless every voter approving a competitor of~$p$ approves a funded set-project, the remainder of the budget is spend on the competitor of~$p$, which is, consequently, not funded.

\begin{theorem}
    \label{thm:Phragmen:CCDC:NPh}
    \CCDCs{\phragmen} is \NPc, even if the projects are of unit cost.
\end{theorem}
\begin{proof}
    As \phragmen is a polynomial-time computable rule, both problems are trivially in \NP; given a solution~$S$, we can simulate the rule with the projects of~$S$ removed and check whether our control goal is achieved. Therefore, let us focus on the hardness part now. We again reduce from the \RXthreeC problem.

    Let~$\mathcal{I}=(U,\mathcal{S})$ be an instance of the \RXthreeCs problem. We construct an instance of the \CCDCs{\phragmen} problem as follows (see \Cref{fig:Phragmen:CCAC:NPh:construction} for an illustration of the construction). The set of projects contains one \emph{set-project}~$p_j$ for every set~$S_j\in\mathcal{S}$,~$3N^2$ \emph{guard-projects}~$g_1^1,\ldots,g_1^N,g_2^1,\ldots,g_{3N}^N$, and our distinguished project~$p$. The tie-breaking rule is so that any project is preferred before the distinguished project~$p$. The set of voters is as follows. There is a single voter~$v$ approving solely the distinguished project~$p$. Then, each voter~$u_i^\ell$ approves the guard-project~$g_i^\ell$ and all (three) set-projects~$p_j$ such that~$u_i\in S_j$. To finalize the construction, we set~$\budget = N+1$,~$\numDeleted = 2N$, and the cost of each project is one.

     \begin{figure*}
        \centering\renewcommand{\arraystretch}{1.2}
        \begin{tabular}{c|cccc|cccccccccc|c}
                &~$p_1$ &~$p_2$ &~$\cdots$ &~$p_{3N}$
                &~$g_1^1$ &~$\cdots$ &~$g_1^N$ &~$g_2^1$ &~$\cdots$ &~$g_2^N$ &~$\cdots$ &~$g_{3N}^1$ &~$\cdots$ &~$g_{3N}^N$
                &~$p$ \\\hline
           ~$v$ & & & & & & & & & & & & & & & \checkmark\\\hline
           ~$u_1^1$  &
                & \checkmark & & & \checkmark & & & & & & & & & \\
           ~$\vdots$ &
                &~$\vdots$ & & & &~$\ddots$ & & & & & & & & & \\
           ~$u_1^N$  &
                & \checkmark & & & & & \checkmark & & & & & & & & \\\hline
           ~$u_2^1$  &
                \checkmark & & & \checkmark & & & & \checkmark & & & & & & & \\
           ~$\vdots$  &
               ~$\vdots$ & & &~$\vdots$ & & & & &~$\ddots$ & & & & & \\
           ~$u_2^N$ &
                \checkmark & & & \checkmark & & & & & & \checkmark & & & & & \\\hline
           ~$\vdots$  &
                 & & & & & & & & & &~$\ddots$ & & & \\\hline
           ~$u_{3N}^1$  &
                \checkmark & \checkmark & & & & & & & & & & \checkmark & & & \\
           ~$\vdots$  &
               ~$\vdots$ &~$\vdots$ & & & & & & & & & & &~$\ddots$ & & \\
           ~$u_{3N}^N$  &
                \checkmark & \checkmark & & & & & & & & & & & & \checkmark & \\
        \end{tabular}
        \caption{An illustration of the election instance created in the proof of \Cref{thm:Phragmen:CCDC:NPh}.}
        \label{fig:Phragmen:CCAC:NPh:construction}
    \end{figure*}

    First, we show that the distinguished project~$p$ is not funded initially. Observe that each set-project~$p_j$ can be funded in time~$1/3N$. Consequently, at the latest in time~$(N+1)/3N$,~$N+1$ set-projects are funded by the rule. Currently, the support of project~$p$ is only~$(N+1)/3N < 1$. Therefore, the distinguished project~$p$ was not funded until now, and the budget is now exhausted. Consequently, no other project will be funded by the rule.

    Next, assume that~$\mathcal{I}$ is a \Yes-instance, and~$C\subseteq \mathcal{S}$ is an exact cover of~$U$. By the definition of the \RXthreeCs problem,~$C$ is of size exactly~$N$. We create a set~$D$ containing all projects~$p_j$ such that the corresponding set~$S_j\in\mathcal{S}$ is \emph{not} element of~$C$. Now, we show that~$D$ is a solution for~$\mathcal{J}$. The size of~$D$ is exactly~$|\mathcal{S}| - |C| = 3N - N = 2N$, which is correct. What remains to show is that if we remove projects in~$D$, the project~$p$ gets funded. Since~$C$ is an exact cover of~$U$, after we remove projects of~$D$ from the instance, each voter~$u_i^\ell$ is approving exactly two projects: a guard-project~$g_i^\ell$ and a set-project~$p_j$ such that~$u_i\in S_j$. Moreover, each set-project is approved by exactly~$3N$ voters. By this property, we have that at time~$1/3N$, all~$N$ remaining set-projects are funded. This decreases the overall budget to one and also decreases the personal budget of all voters~$u_i^\ell$ to zero. Consequently, the first guard project may be funded in time~$1/3N + 1$. However, the budget of voter~$v$ would now be~$1/3N + 1 > 1$. Therefore, the rule must fund the project~$p$ before this time.

    In the opposite direction, let~$\mathcal{J}$ be a \Yes-instance, and~$D\subseteq \projects\setminus\{p\}$ be a set of projects such that if we remove~$D$ from the instance, the \phragmen rule funds our distinguished project~$p$. First, we show that the size of~$D$ is exactly~$2N$. For the sake of contradiction, assume that~$|D| < 2N$. Then, after removing~$D$ from the instance, the instance contains at least~$N+1$ set-projects. However, at the latest, at~$(N+1)/3N$, exactly~$N+1$ of them are funded. At this time, the personal budget of the voter~$v$ is only~$(N+1)/3N < 1$, and therefore, the project~$p$ cannot be funded now or earlier. Hence, the funded set-projects exhaust the budget, and~$p$ is never funded. This contradicts the~$D$ is a solution. Thus,~$|D| = 2N$. By the same argumentation, we get that~$D$ does not contain any guard-project; if it would be the case, then we would again have~$N+1$ remaining set-projects which exhaust the overall budget before the voter~$v$ has enough personal budget to fund~$p$. Next, we show that there is no~$u_i^\ell$ such that~$A(u_i^\ell) = \{g_i^\ell\}$. For the sake of contradiction, let such a voter exist. At the latest, at time~$N/3N = 1/3$, the \phragmen rule funds all set-project. Recall that once a set-project approved by a voter~$u_{i'}^{\ell'}$ is funded, the guard-project~$g_{i'}^{\ell'}$ loses all its support. However, we considered that there exists a voter~$u_i^\ell$ not approving any of the set-projects. Consequently, the guard-project~$g_i^j$ never loses its support and, at time~$1$, it can be funded. At this time, the distinguished project~$p$ may also be funded. However, due to the tie-breaking, the rule first funds the guard-project~$g_i^\ell$. This exhausts the whole budget, and therefore, we have a contradiction with~$D$ being a solution. Now it is easy to see that if we take~$C = \{ S_j\in \mathcal{S} \mid p_j \not\in D \}$, we obtain a solution for~$\mathcal{I}$. This finishes the proof for constructive control.
\end{proof}

\begin{theorem}
    \label{thm:Phragmen:DCDC:NPh}
    \DCDCs{\phragmen} is \NPc, even if the projects are of unit cost.
\end{theorem}
\begin{proof}
    Again,~$\DCDCs{\phragmen}\in\NP$ is obvious, so we discuss only the \NPhness part of the proof. The reduction is again from the \RXthreeC problem, but the construction is very different from all the previous ones. Specifically, given an instance~$\mathcal{I}=(\mathcal{S},U)$ of \RXthreeCs, we create an equivalent instance~$\mathcal{J}$ of \DCDCs{\phragmen} as follow.

   We have three types of projects. First, we have~$3N$ \emph{set-projects}~$p_1,\ldots,p_{3N}$ which are in one-to-one correspondence to the sets of~$\mathcal{S}$. Next, we have~$3N$ \emph{guard-projects}~$g_1,\ldots,g_{3N}$. Finally, there is the distinguished project~$p$ we aim to remove from the set of winning projects. The set of voters contains
   \begin{itemize}
       \item for every~$i\in[3N]$, eight \emph{element-voters}~$u_i^1,\ldots,u_i^8$ corresponding to an element~$u_i\in U$ and approving a project~$p_j$ if and only of~$u_i\in S_j$. Moreover, the element-voters~$u_i^1,\ldots,u_i^j$ additionally approve the guard-project~$g_i$.
       \item for every~$j\in[3N]$, six voters~$x_j^1,\ldots,x_j^6$ approving only the set-project~$p_j$.
       \item five voters~$v_1,\ldots,v_5$ approving only the distinguished project~$p$.
   \end{itemize}
   To finalize the construction, we set~$\budget = 4N$,~$\numDeleted = 2N$, and the tie-breaking to be~$p_1\succ p_2 \succ \cdots p_{3N} \succ g_1 \succ \cdots \succ g_{3N} \succ p$. The idea behind the construction is that every guard-project has all supporters in common with exactly three set-projects, and the time~$T_p$ in which the distinguished project~$p$ is funded is fixed. Moreover, \emph{all} set-projects are funded before any guard project and also before the distinguished project. Because of this, whenever a set-project is funded, then all voters supporting some guard-projects lose their whole support. The numbers of voters are designed so that a guard-project gets funded if and only if its supporters lose their support at most once and in the concrete time step. Otherwise, the distinguished project is funded instead.

   We first show that initially, the distinguished project~$p$ is indeed funded. Let us analyze how the \phragmen rule proceeds. In time~$T_0 = 1/(6+3\cdot 8)$, first set-projects can be funded. At this time, the overall balance of supporters of any guard-project~$g_j$ is~$6/30 < 1$, and for the distinguished project~$p$, it is~$5/30 < 1$. We buy at most~$N$ set-projects, and the process continues. In time~$T_N = 1/6$, the voters~$x_j^1,\ldots,x_j^6$ have enough money to fund each set-project~$p_j$, and therefore, all the set-projects are funded, and the remaining budget~$\budget' = \budget - 3N = N$. The sum of balances of supporters of~$p$ is~$5/6$, and for any guard-project~$g_\ell$,~$\ell\in[3N]$, it is strictly smaller than~$6\cdot 1/6 = 1$, because each of them lost the whole support exactly three times. Finally, in time~$T_p = 1/5$, the distinguished project~$p$ may be funded. At this time, the sum of balances of voters~$u_j^1,\ldots,u_j^6$ is strictly smaller than~$6\cdot T_p - 6\cdot T_0 = 6/5 - 6/30 = 5/5 = 1$ because each of them lost the whole support exactly three times, and the last time in which they lost support the last time is strictly greater than~$T_0$. Hence,~$p$ gets funded as the budget still allows to fund~$N$ projects at this point.

   Now, assume that~$\mathcal{I}$ is a \Yes-instance, and~$C\subseteq \mathcal{S}$ is an exact cover of~$U$. We set~$D = \{ p_j \mid S_j \not\in C \}$; that is, we remove all set-projects whose corresponding sets are not in the exact cover~$C$. We claim that~$D$ is a solution for~$\mathcal{J}$. Again, we analyze the process. The first time a project can be funded is~$T_0 = 1/30$. At this time, we fund \emph{all} set-projects that were not deleted (if it would not be the case, then there is a pair of set-projects supported by the same element-voters, which contradicts that~$C$ is an exact cover), and the remaining budget is~$\budget' = \budget - N = 3N$. This resets the balance of all element-voters to zero. Next, at time~$T_p = 1/5$, we can fund the distinguished project~$p$. However, the support of each guard-project~$g_j$,~$j\in[3N]$, is~$6\cdot T_p - 6\cdot T_0 = 6/5 - 6/30 = 1$ and due to the tie-breaking order and the fact that their sets of supporters are disjoint, they are funded before~$p$, which exhausts the budget and~$p$ is therefore not funded. Hence,~$\mathcal{J}$ is also a \Yes-instance, and~$D$ is a solution.

   In the opposite direction, assume that~$\mathcal{J}$ is a \Yes-instance and~$D\subseteq \projects\setminus\{p\}$ is a solution. We show that in every solution~$D$, there is exactly one round ($T_0 = 1/30$, to be precise) in which the set-projects are funded. For the sake of contradiction, assume that it is not the case, and there are at least two time-steps~$T_0$ and~$T_1$,~$T_0 < T_1$, in which the \phragmen rule funds some set-projects. Then, some guard-projects lose their whole support at time~$T_1 > T_0$, and therefore, by the previous argumentation, the sum of balances of their supporters at time~$T_p = 1/5$ is strictly smaller than~$6\cdot T_p - 6\cdot T_0 = 1$, meaning that they cannot be funded at time~$T_p$. Now, we need to show that, after all set-projects are funded, the remaining budget is big enough to fund~$p$. Assume first that~$|(P\setminus D)\cap\{p_1,\ldots,p_{3N}\}| = N$, that is, we removed exactly~$2N$ set-projects. Then, after all set-projects are funded, the remaining budget is~$4N - N$. Additionally, at least three guard-projects cannot be funded in time~$T_p$, so before~$p$ is assumed, the remaining budget is at least~$3N - 3N + 3 = 3$, meaning that~$p$ is funded. This contradicts that~$D$ is a solution. Therefore, we must have~$|(P\setminus D)\cap\{p_1,\ldots,p_{3N}\}| \geq N$. Let~$\ell$ be the number of set-projects funded at some time~$t > T_0$ and~$k = |(P\setminus D)\cap\{p_1,\ldots,p_{3N}\}| - N$. It holds that~$\ell \geq k$. The remaining budget after all set-projects are funded is~$4N - N - k = 3N - k$, and each of~$\ell$ set-projects funded after~$T_0$ resets support of at least two unique guard-projects, meaning that at time~$T_p$, the remaining budget is~$4N - N - k - (3N - 2\cdot\ell) = -k + 2\cdot\ell > 0$. That is,~$p$ gets funded, and we have a contradiction with~$D$ being a solution. As we exhausted all possibilities, there must be exactly one round~$T_0$ in which set-projects are funded. This can happen if and only if, for each pair of set-projects, the sets of their supporters are disjoint.
\end{proof}

\subsection{\equalShares}

To finalize the complexity picture, we show that for the \equalShares rule, the control by adding or deleting projects is intractable even in the simplest setting where all projects are of the same cost and are unweighted. To prove this result, we exploit a reduction of \citet{JaneczkoF2023}, who showed that it is \NPh to decide whether, for a given PB instance, the \equalShares rule outputs the same outcome for every tie-breaking order.

\begin{theorem}
    \label{thm:MES:CCDC:NPc}\label{thm:MES:DCDC:NPc}
    Both \CCDCs{\equalShares} and \DCDCs{\equalShares} are \NPc, even if the projects are of unit cost.
\end{theorem}
\begin{proof}
    Membership of both problems to \NP is clear; it suffices to guess a set of projects to remove and check whether the distinguished project is (not) funded.
    To prove \NPhness, we reduce from \RXthreeCs. Let~$\mathcal{I}=(U,\mathcal{S})$ be an instance of the \RXthreeCs problem. Let~$|U|=3N$ be the number of elements in~$U$. We construct an election exactly as in \cite[Theorem 3.4]{JaneczkoF2023} with the only difference we create~$2N$ more copies of candidate~$d$ (denoted as~$d_1, \ldots, d_{2N}$ respectively). We set~$\numDeleted=2N$, projects' costs to one, and~$\budget$ to the same value as the committee size in the original instance. We set tie-breaking order to be~$C_b \succ C_u \succ \mathcal{S} \succ p \succ d \succ d_1 \succ \ldots \succ d_{2N} \succ c_1 \succ c_2 \succ D$. Let~$p$ be the distinguished project for the constructive case and~$d$ for the destructive case.

    The initial budgets in the instance will be the same as in the original instance, so we can see that the rule will process the groups in the same order as in the original groups. Therefore, irrespectively which up to~$\numDeleted=2N$ projects we delete, a)~$p$ will need to compete with at least one of~$d, d_1, d_2, \ldots, d_{2n}$ and b) total budget of voters~$V_{pd}$ and~$U'$ after processing nonremoved projects from~$C_B$ and~$C_U$ will always be lower than~$2$ because~$12N \cdot (\frac{1}{12N} + 2N \cdot \frac{1}{144N^3}) + 3N \cdot \frac{1}{18N+6} + 2N \cdot \frac{1}{54N^3+27N^2+3N} = 1 + \frac{1}{6N} + \frac{3N}{18N+6} + \frac{6N^2}{54N^3+27N^2+3N} < 2$. Having taken~$12N \cdot \frac{1}{12N} = 1$ into account, we know that exactly one nonremoved project from~$p, d, d_1, \ldots, d_{2N}$ will get funding. From the original reasoning, we know that~$p$ always loses with~$d, d_1, d_2, \ldots, d_{2N}$ unless the formerly selected~$\mathcal{S}$-projects form an exact cover over~$U$. Thus, if~$\mathcal{I}$ has an exact cover~$X$ (and is a \Yes-instance), we can remove all projects that correspond to sets~$\mathcal{S} \setminus X$ and make~$p$ winning. On the other hand, if our control instance is a \Yes-instance, then we know that the budget of each voter from~$U'$ must have been spent totally before considering~$p, d, d_1, \ldots, d_{2N}$. We know that voters from~$C_U$ cannot take the whole budget of voters from~$U'$, so it must have been projects from~$\mathcal{S}$ that finished their budget. But as we consider only disjoint sets before considering~$p, d, d_1, \ldots, d_{2N}$, we must have considered~$N$ pairwise disjoint~$\mathcal{S}$-projects before~$p, d, d_1, \ldots, d_{2N}$, which means that we found an exact cover over~$U$ and~$\mathcal{I}$ is a \Yes-instance.
\end{proof}

\section{Adding Projects}\label{sec:AC}

In this section, the control operation we can perform to achieve our goal is adding candidates. Formally, this operation is based on the idea that certain projects are initially inactive, meaning that the rule under consideration does not assume them. However, we have full information about which voters will vote for them if these projects are activated.

Especially for constructive control, it may not be intuitive how adding projects may help select the initially unfunded projects. In the following example, we illustrate the concept of control by adding projects.

\begin{example}
    Let the set of standard projects~$\projects$ consist of a project~$d$ and our distinguished project~$p$, and let the costs of the standard projects be~$\cost(d) = 2$ and~$\cost(p) = 1$. There is one spoiler project~$c$ with~$\cost(c) = 1$. Moreover, we have three voters~$v_1$,~$v_2$, and~$v_3$. The voter~$v_1$ approves only the project~$c$, the voter~$v_2$ approves both~$c$ and~$d$, and the voter~$v_3$ approves all the projects. Finally, the budget~$\budget$ is two, and we assume the \greedyAV rule. Initially, the project~$d$ is the one with the highest support. Therefore, the rule funds~$d$. This exhausts the whole budget, so~$p$ is not funded. If we add the spoiler project~$c$, this project is assumed by the rule the first. Since~$\cost(c) < \budget$,~$c$ gets funded and the budget is decreased to one. Next, the rule considers project~$d$, which cannot be funded because~$\cost(d) = 2$. The only remaining project is~$p$. This time, there is enough money in the budget to fund it.
\end{example}

In the rest of this section, we provide the complexity picture for all combinations of goals and PB rules we consider. Both the results and techniques used are very similar to the case of project deletion operation. Therefore, we defer most of the proofs into the supplementary material.

\subsection{\greedyAV and \greedyCost}

We start with the case of \greedyAV rule. First, we show that it is \NPh to decide whether a subset of spoiler projects of the correct size exists for both control goals. Similarly to the proof of \Cref{thm:AV:CCDC:NPh}, we reduce from the \RXthreeCs problem, create one project for every set~$S$, and encode elements of~$S$ using project cost. The difference here is that the projects corresponding to sets are initially inactive and, unless we add a set of projects corresponding to an exact cover for the \RXthreeCs instance, the distinguished project~$p$.

\begin{theorem}
    \label{thm:AV:CCAC:NPh}\label{thm:AV:DCAC:NPh}
    Both \CCACs{\greedyAV} and \DCACs{\greedyAV} are \NPc, even if~$|\voters|=2$ and there are two standard projects.
\end{theorem}
\begin{proof}
    It is very easy to see that both problems are in \NP: if a solution is given to us, we can simulate (in polynomial time) the rule and check whether~$p$ is or is not funded.

    For \NPhness, we again reduce from the \RXthreeC problem. The idea of the construction is similar to the one used in \Cref{thm:AV:CCDC:NPh}. Again, we create a project for every set of~$\mathcal{S}$ and encode its elements using the cost function. However, these projects are initially inactive. There is also a direct competitor~$g$ of the distinguished project~$p$, which is always funded unless we return to the instance exactly projects corresponding to an exact cover in the RX3C instance.

    More formally, let~$\mathcal{I} = (U,\mathcal{S})$ be an instance of the RX3C problem. We create an instance~$\mathcal{J}$ of the \CCACs{\greedyAV} problem (we discuss the destructive variant at the end of the proof) as follows. The qualified projects comprise just two: a \emph{guard-project}~$g$ and our distinguished project~$p$. The set~$Q$ of spoiler projects then consists of~$3N$ projects, one for each set~$S\in\mathcal{S}$. We set the prices for our projects as follows. The cost of the distinguished project~$p$ is~$\cost(p) = \sum_{i=1}^{3N} 1\cdot 4^i$ and the cost of the guard-project is~$\cost(g) = \cost(p) + 1$. Let~$p_j$ be a spoiler project corresponding to a set~$S_j=\{u_{i_1},u_{i_2},u_{i_3}\}\in\mathcal{S}$. We set~$\cost(p_j) = 1\cdot 4^{i_1} + 1\cdot 4^{i_2} + 1\cdot 4^{i_3}$. The set of voters contains just two:~$v_1$ and~$v_2$. The first voter~$v_1$ approves all spoiler projects and the guard-project~$g$, and the second voter~$v_2$ approves only spoiler projects. To finalize the construction, we set~$\numDeleted = N$ and~$\budget = \sum_{i=1}^{3N} 2\cdot 4^i$.

    First, we show that the distinguished project~$p$ is indeed initially not funded. The PB instance without spoiler projects consists only of~$g$ and~$p$. Moreover, the guard-project~$g$ has a higher \greedyAV-score than the distinguished project~$p$; therefore, the \greedyAV rule processes~$g$ first. The budget is big enough to fund~$g$. Thus,~$g$ is funded, and the remaining budget is
    \[
        \budget-\cost(g) = \sum_{i=1}^{3N} 2\cdot 4^i - \left[(\sum_{i=1}^{3N} 1\cdot 4^i) + 1\right] = (\sum_{i=1}^{3N} 1\cdot 4^i) - 1,
    \]
    which is clearly smaller than~$\cost(p) = \sum_{i=1}^{3N} 1\cdot 4^i$. Hence, there is not enough remaining budget to fund~$p$.

    Next, let us show the correctness of the construction. Assume that~$\mathcal{I}$ is a \Yes-instance and~$C\subseteq \mathcal{S}$ is an exact cover in~$\mathcal{I}$. We set~$S = \{ p_j \mid S_j\in C \}$ and claim that it is a solution for~$\mathcal{J}$. Since the spoiler projects are approved by both voters and the guard-project by only one, they are clearly processed before~$g$ by the rule. Moreover, as~$C$ was a set cover, the overall price of the added projects is~$\sum_{i=1}^{3N} 1\cdot 4^i$. Hence, before~$g$ is processed, the remaining budget is~$\budget-\sum_{i=1}^{3N} 1\cdot 4^i = \sum_{i=1}^{3N} 2\cdot 4^i - \sum_{i=1}^{3N} 1\cdot 4^i = \sum_{i=1}^{3N} 1\cdot 4^i$. This is smaller by one unit than the price of~$g$, so~$g$ is not funded. On the other hand, it is exactly the price of~$p$, so~$p$ gets funded. Therefore,~$S$ is a solution, and~$\mathcal{J}$ is also a \Yes-instance.

    In the opposite direction, assume that~$\mathcal{J}$ is a \Yes-instance, and~$S\subseteq Q$ is a solution. First, observe that the costs of the projects are set up identically to the set-projects in the proof of \Cref{thm:AV:CCDC:NPh}; therefore, \Cref{clm:noSmallSetSumUpToCostOfp} is true also for our instance. It is not hard to see that every solution~$S$ satisfies~$\sum_{p_j\in S} \cost(p_j) = \cost(p)$, as~$\budget = 2\cdot\cost(p)$. If this is not the case, then we either spend too much budget on~$S$, and~$p$ is not affordable, or the remaining budget after projects of~$S$ are funded is strictly greater than~$\cost(p)$, and thus, the guard-project~$g$ is funded, and~$p$ is not. By \Cref{clm:noSmallSetSumUpToCostOfp}, we obtain that~$|S| = N$. Since~$\cost(S) = \sum_{p_j\in S} \cost(p_j) = \sum_{i=1}^{3N} 1\cdot 4^i$, by the Pigeonhole principle and the definition of projects' costs, each project~$p_j\in S$ contributes to~$\cost(S)$ with exactly~$4^{i_1} + 4^{i_2} + 4^{i_3}$ for unique~$i_1,i_2,i_3\in[3N]$. That is, if we set~$C = \{ S_j\mid p_j\in S \}$, we immediately obtain that~$C$ is an exact cover for~$\mathcal{I}$. This shows that~$\mathcal{I}$ is a \Yes-instance, which finishes the correctness of the construction.

    For the \emph{destructive} variant of control, we again just tweak the construction a bit. Specifically, we set~$\cost(g) = \sum_{i=1}^{3N} 1\cdot 4^i$,~${\cost(p) = 1}$, and the budget to~$\budget = \sum_{i=1}^{3N} 2\cdot 4^i$. Obviously, both projects are initially funded as~$\cost(g) + \cost(p) < \budget$. For correctness, suppose that~$\mathcal{I}$ is a \Yes-instance, and~$C\subseteq \mathcal{S}$ is an exact cover. We add to the instance all spoiler projects corresponding to sets of~$C$. Since~$C$ is an exact cover, and by the definition of the costs, the rule spends~$\sum_{i=1}^{3N} 1\cdot 4^i$ on the added set-projects. Then, it funds the guard-project~$g$, which exhausts the overall budget, so the distinguished project~$p$ cannot be funded. In the opposite direction, let~$S$ be a solution for~$\mathcal{J}$. There are two cases in which adding projects of~$S$ can lead to~$p$ not being funded. The first possibility is that the whole budget is exhausted on projects of~$S$, that is,~$\cost(S) = \sum_{i=1}^{3N} 2\cdot 4^i$. This is, however, not possible, as~$|S| \leq N$, each project contributes to the sum by three distinct powers of four, and each distinct power of four is present in the cost of exactly three projects. Therefore, it must be the case that~$\cost(S) = \cost(g)$ since, in this case,  after the rule funds the added projects~$S$, the remaining budget before the guard-project~$g$ is processed is exactly~$\sum_{i=1}^{3N} 1\cdot 4^i$ and~$g$ completely exhausts it. Clearly, if~$\cost(S) < \sum_{i=1}^{3N} 1\cdot 4^i$, the remaining budget when~$g$ is processed is at least~$(\sum_{i=1}^{3N} 1\cdot 4^i) + 1$ and therefore, the rule can fund both~$g$ and~$p$. Thus,~$\cost(S) \geq \sum_{i=1}^{3n} 1\cdot 4^i$. If~$\cost(S) > \sum_{i=1}^{3n} 1\cdot 4^i$, then once the \greedyAV rule processes~$g$, there is not enough budget to fund it, and consequently, the project~$p$ is funded. Hence, we have that~$\cost(S) = \sum_{i=1}^{3n} 1\cdot 4^i$, and if we select~$C$ to consist of sets corresponding to set-projects in~$S$, by the definition of the costs, we obtain that~$C$ is an exact cover of~$U$. This finishes the proof.
\end{proof}

The costs used in the hardness results are exponential in the number of voters and projects. In our next result, we show that such costs are necessary for hardness. In particular, we present a pseudopolynomial time algorithm for both variants of the problem. The algorithm works even for the weighted setting and is again based on the dynamic programming over the projects, as in \Cref{thm:AV:CCDC:P:ifUnaryPrices}, albeit technical details are adapted for project addition.

\begin{theorem}
    \label{thm:AV:CCAC:P:ifUnaryPrices}\label{thm:AV:DCAC:P:ifUnaryPrices}
    If the costs of the projects are encoded in unary, \CCACs[\$]{\greedyAV} and \DCACs[\$]{\greedyAV} can be solved in polynomial time for any tie-breaking order, even if the projects' weights are encoded in binary.
\end{theorem}
\begin{proof}
    \newcommand{\DP}{\operatorname{DP}}
    We start by adding all spoiler projects to the input PB instance and computing an order~$p_1,\ldots,p_{m-1},p_m = p$ in which the \greedyAV rule processes the projects. By \Cref{thm:AV:removeUnnecessaryProjects}, we can assume that~$p$ is the last project. Moreover, we can assume that~$p_1$ is a spoiler project because we cannot affect projects before the first spoiler project. Hence, we can just simulate the rule on these projects to obtain a budget~$\budget'$ just before the first spoiler project is processed. Then, we remove these standard projects and set~$\budget=\budget'$ and obtain an equivalent instance. We can also, without loss of generality, assume that for every project~$p_j$, we have~${\cost(p_j) \leq \budget}$. Adding projects that are not affordable, even with the whole available budget, into a solution is never optimal; therefore, we can delete them without changing the solution. This step can be done clearly in polynomial time, as \greedyAV is polynomial-time computable rule.

    The core of our algorithm is a dynamic programming table~$\DP[j,\beta]$, where
    \begin{itemize}
        \item~$j\in[m-1]$ is an index of the last processed project, and
        \item~$\beta\in[\budget]_0$ is a desired remaining budget just before the \greedyAV rule processes a project~$p_{j+1}$.
    \end{itemize}
    The pair~$(j,\beta)$ is called a \emph{signature}, and for every signature, the dynamic programming table stores the weight of a minimum-weight \emph{partial solution}~$S_{j,\beta}\subseteq \{p_1,\ldots,p_j\}\cap Q$ such that~${\wFn(S_{j,\beta}) \leq \numDeleted}$ and if the projects from~$S_{j,\beta}$ are added to the input instance, the remaining budget just before \greedyAV rule processes project~$p_{j+1}$ is exactly~$\beta$. If no such partial solution exists, we store some large value~$\infty > \numDeleted$.

    The computation is now defined as follows. We start with the base case, which is when~$j=1$. As was argued before, we can assume that~$p_1$ is a spoiler project. Hence, we can decide whether to include~$p_1$ into a solution or not. This leads to the following computation.
    \[
        \DP[1,\beta] = \begin{cases}
            0 & \text{if } \beta = \budget\text{,}\\
            \wFn(p_1) & \text{if } \beta + \cost(p_1) = \budget\text{, and}\\
            \infty & \text{otherwise.}
        \end{cases}
    \]

    Next, we show the computation for every~$j > [2,m-1]$. We distinguish two cases based on the type of project. First, let~$p_j$ be a standard project. Our control operation cannot affect such a project, so we decide whether the prescribed remaining budget~$\beta$ allows for funding~$p_j$ or not. This is achieved as follows.
    \[
        \DP[j,\beta] = \begin{cases}
            \DP[j-1,\beta] & \hspace{-5pt}\text{if } \beta + \cost(p_j) \hspace{-1pt}>\hspace{-1pt} \budget,\\
            \DP[j-1,\beta+\cost(p_j)] & \hspace{-5pt}\text{otherwise.}
        \end{cases}
    \]

    The most interesting part is when~$p_j$ is a spoiler project. Then, we need to decide whether to include~$p_j$ in a solution or not. This is done via the following recurrence:
    \[
        \DP[ j, \beta ] = \begin{cases}
            \DP[j-1,\beta] \hspace{45pt}\text{if } \beta+\cost(p_j) > \budget,\\
            \min\big\{ \DP[j-1,\beta + \cost(p_j)] + \wFn(p_j),\\
            \phantom{\min\big\{}\DP[ j-1, \beta]\big\} \hspace{10pt}\text{ otherwise.}
        \end{cases}
    \]

    For the correctness of the algorithm, we first show that whenever the dynamic programming table stores some value smaller than infinity, a corresponding partial solution exists.

    \begin{claim}
        Whenever~$\DP[j,\beta] = w < \infty$, there exists a set~$S \subseteq \{p_1,\ldots,p_j\} \cap Q$ such that~$\wFn(S) = w$ and~$\cost(W) = \budget - \beta$, where~$W=\greedyAV((\{p_1,\ldots,p_j\}\cap \projects) \cup S, V, \budget)$.
    \end{claim}
    \begin{claimproof}
        We show the claim by induction over~$j$. The basic step of the induction is when~$j=1$. and, by the definition of the computation,~$w < \infty$ if and only if either~$\beta = \budget$ or~$\beta = \budget - \cost(p_1)$. In the former case, the value of~$w$ is zero. If we take~$S=\emptyset$, clearly~$\wFn(S) = 0$ and the budget before the rule processes project~$p_2$ is exactly~$\budget$ since~$p_1$ is a spoiler project and is not added to the instance. In the latter case,~$w=\wFn(p_1)$. We can set~$S = \{p_1\}$ and we obtain that~$\wFn(S) = \wFn(p_1) = w$ and the rule now funds project~$p_1$ and hence, the remaining budget before~$p_2$ is processed by the rule is exactly~$\cost(p_1) = \budget - \beta$. Therefore, the claim holds for~$j=1$.

        Let~$j>1$ and assume the claim holds for~$j-1$. We distinguish two cases based on whether~$p_j$ is a standard or a spoiler project. First, assume that~$p_j$ is a standard project. If~$\beta + \cost(p_j) > \budget$, the value stored in the cell~$\DP[j,\beta]$ is equal to the value of~$\DP[j-1,\beta]$. By the induction hypothesis, there is a set~$S'\subseteq \{p_1,\ldots,p_{j-1}\} \cap Q$ such that~$\wFn(S') = w$ and~$\cost(\greedyAV((\{p_1,\ldots,p_{j-1}\}\cap \projects) \cup S', V, \budget)) = \budget - \beta$. Since the remaining budget before~$p_j$ is processed is~$\beta$, the \greedyAV rule cannot fund~$p_j$. Consequently,~$S'$ is also a partial solution corresponding to~$\DP[j,\beta]$. Now, assume that~$\beta + \cost(p_j) \leq \budget$. Then~$\DP[j-1,\beta+\cost(p_j)] = w$ and, by the induction hypothesis, there exists a set~$S'\subseteq\{p_1,\ldots,p_{j-1}\}\cap W$ such that~$\wFn(S') = w$ and~$\cost(\greedyAV((\{p_1,\ldots,p_{j-1}\}\cap \projects) \cup S', V, \budget)) = \budget -  \beta + \cost(p_j)$. Since~$p_j$ is a standard project, we cannot affect it. Moreover, the remaining budget is clearly big enough to fund~$p_j$, the \greedyAV rule selects~$p_j$, and consequently,~$\cost(\greedyAV((\{p_1,\ldots,p_{j-1},p_j\}\cap \projects) \cup S', V, \budget)) = \cost(\greedyAV((\{p_1,\ldots,p_{j-1}\}\cap \projects) \cup S', V, \budget)) + \cost(p_j) = \budget - (\beta + \cost(p_j)) + \cost(p_j) = \budget - \beta$. Henceforth, we obtain the~$S'$ is also a partial solution corresponding to~$\DP[j,\beta]$ and that the claim holds even for standard projects.

        Finally, let~$p_j$ be a spoiler project. First, assume that~$\beta + \cost(p_j) > \budget$. Then, the value stored in~$\DP[j,\beta]$ corresponds to the value stored in~$\DP[j-1,\beta]$ by the definition of the computation. Using induction hypothesis, there must exists a set~$S'\subseteq \{p_1,\ldots,p_{j-1}\}\cap W$ such that~$\wFn(S') = w$ and~$\cost(\greedyAV((\{p_1,\ldots,p_{j-1},p_j\}\cap \projects) \cup S', V, \budget)) = \budget - \beta$. If we add spoiler projects from~$S'$ to the instance, the remaining budget before~$p_j$ is processed is exactly~$\beta$. If we decide not to include~$p_j$ into~$S'$, observe that the remaining budget before~$p_{j+1}$ is processed is still~$\beta$, since~${\beta + \cost(p_j) > \budget}$ and therefore, even if we would add the project to the instance, the rule does not fund this project. Hence,~$S'$ is a partial solution corresponding to~$\DP[j,\beta]$. What remains to show is that the computation is correct even if~$\beta + \cost(p_j) \leq \budget$. By the definition of the computation, the value of~$\DP[j,\beta]$ is computed as the minimum of~$\DP[j-1,\beta+\cost(p_j)]+\wFn(p_j)$ and~$\DP[j-1,\beta]$. First, assume that~$\DP[j-1,\beta] \leq \DP[j-1,\beta+\cost(p_j)]+\wFn(p_j)$. Then~$\DP[j-1,\beta] = w$ and, by the induction hypothesis, there exists a set~$S'$ such that~$\wFn(S') = w$ and~$\cost(\greedyAV((\{p_1,\ldots,p_{j-1}\}\cap \projects) \cup S', V, \budget)) =  \budget - \beta$. If we keep the set~$S'$ the same, the project~$p_j$ is not added into a solution and therefore, also~$\cost(\greedyAV((\{p_1,\ldots,p_{j-1},p_j\}\cap \projects) \cup S', V, \budget)) =  \budget - \beta$, proving that~$S'$ is corresponding partial solution. Second, let~$\DP[j-1,\beta+\cost(p_j)]+\wFn(p_j) < \DP[j-1,\beta]$. By the induction hypothesis, there exists a set~$S'\subseteq\{p_1,\ldots,p_{j-1}\}\cap W$ such that~$\wFn(S') = w - \wFn(p_j)$ and~$\cost(\greedyAV((\{p_1,\ldots,p_{j-1}\}\cap \projects) \cup S', V, \budget)) =  \budget - (\beta+\cost(p_j))$. We set~$S = S' \cup \{p_j\}$ and claim that~$S$ is a partial solution corresponding to~$\DP[j,\beta]$. Clearly,~$\wFn(S) = \wFn(S') + \wFn(p_j) = w-\wFn(p_j) + \wFn(p_j) = w$. Now, consider how the rule proceeds. Just before~$p_j$ is processed, the remaining budget is~$\beta + \cost(p_j)$. We included~$p_j$ in a solution, and moreover,~$p_j$ is affordable. Hence, the rule selects the project~$p_j$, and the budget before the project~$p_{j+1}$ is processed by the \greedyAV rule is exactly~$\beta + \cost(p_j) - \cost(p_j) = \beta$, is intended. This finishes the proof of the claim.
    \end{claimproof}

    To complete the proof of correctness, in the next claim, we show that for every admissible solution set~$S$, the table stores at most the value~$\wFn(S)$. That is, the stored weight corresponds to a minimum-weight solution securing the goal prescribed by the signature.

    \begin{claim}
        For every~$j\in[m-1]$ and every~$S\subseteq \{p_1,\ldots,p_j\}\cap W$ such that~$\wFn(S) \leq r$ and~$\cost(W) \leq B$, where~$S=\greedyAV((\{p_1,\ldots,p_j\}\cap \projects) \cup S, V, \budget)$, we have that~$\DP[j,\beta] \leq \wFn(S)$.
    \end{claim}
    \begin{claimproof}
        Again, we prove the claim by induction over~$j$. First, consider the basic step when~$j=1$. There are only two possible sets~$S$, namely~$\emptyset$ and~$\{p_1\}$, and both of them satisfy the conditions from the claim statement. If~$S = \emptyset$, then the \greedyAV rule before processing~$p_2$ does not fund any project; hence, the remaining budget~$\beta$ is still~$\budget$. However, by the definition of the computation,~$\DP[1,\beta]$ is zero whenever~$\beta = \budget$, which is clearly at most~$\wFn(\emptyset) = 0$. For the latter case, if~$S=\{p_1\}$, then the rule funds~$p_1$ in the first round and the remaining budget~$\beta$ just before~$p_2$ is processed is~$\budget - \cost(p_1)$. By the definition of the computation, our table stores~$\wFn(p_1)$ whenever~$\beta = \budget - \cost(p_1)$. It clearly holds that~$\wFn(p_1) \leq \wFn(S) = \wFn(p_1)$ and the basic step indeed holds.

        Next, let~$j> 1$, the claim hold for~$j-1$, and let~$p_j$ be a standard project. For any set~$S$ satisfying the conditions of the claim, we have that~$p_j\not\in S$. If~$\beta + \cost(p_j) > \budget$, then the rule cannot fund~$p_j$, as it would exceed the budget. Therefore, by the induction hypothesis,~$\DP[j-1,\beta] \leq \wFn(S)$ and, by the definition of the computation, also~$\DP[j,\beta] \leq \wFn(S)$. If~$\beta + \cost(p_j) \leq \budget$, then the rule must fund~$p_j$. Consequently, by the induction hypothesis,~$\DP[j-1,\beta+\cost(p_j)] \leq \wFn(S)$. Since, in this case, we set~$\DP[j-1,\beta]$ to the same value, it must hold that~$\DP[j-1,\beta] \leq \wFn(S)$.

        Finally, let~$j > 1$, the claim hold for~$j-1$, and let~$p_j$ be a spoiler project. Se distinguish two cases. First, consider that~$p_j\not\in S$. Then~$p_j$ is not added to the instance, and hence, by the induction hypothesis, it must hold that~$\DP[j-1,\beta] \leq \wFn(S)$. By the definition of the computation, we store either directly the value of~$\DP[j-1,\beta]$ or a minimum of~$\DP[j-1,\beta]$ and~$\DP[j-1,\beta+\cost(p_j)] + \wFn(p_j)$. In any case, the value stored in~$\DP[j,\beta]$ is at most~$\wFn(S)$. If~$p_j$ is an element of~$S$ and~$p_j$ is not funded by the \greedyAV rule, it is unnecessary to add~$p_j$ into a solution. Thus,~$\DP[j-1,\beta] \leq \wFn(S\setminus\{p_1\}) < \wFn(S)$, where the first inequality holds by the previous case. Hence, let~$p_j\in S$ and~$p_j$ be funded by the \greedyAV rule. Then, by the induction hypothesis, the value stored in~$\DP[j-1,\beta + \cost(p_j)]$ is at most~$\wFn(S) - \wFn(p_j)$. By the definition of the computation, our table stores minim of~$\DP[j-1,\beta]$ and~$\DP[j-1,\beta+\cost(p_j)] + \wFn(p_j)$, which is at most~$\wFn(S) - \wFn(p_j) + \wFn(p_j) = \wFn(S)$, finishing the proof.
    \end{claimproof}

    Once the dynamic programming table~$\DP$ is correctly computed for all~$j\in[m-1]$ and~$\beta\in[\budget]_0$, we can directly decide our problem. For the constructive variant, we check whether there exists~$\DP[\numProjects-1,\beta] \leq \numDeleted$ for some~$\beta \geq \cost(p)$. If this is the case, we return \Yes. If no such~$\beta$ exists, we return \No. Observe that, since we assume~$\DP$ to be correctly computed, there exists a partial solution~$S_{\numProjects-1,\beta}$ with~$\wFn(S_{\numProjects-1,\beta}) < \numDeleted$ such that the remaining budget when the \greedyAV rule processes the project~$p$ is exactly~$\beta\geq \cost(p)$. Hence, the rule fund~$p$ and the constructive control is indeed successful.
    For the destructive variant, we check whether there exists~$\beta < \cost(p)$ such that~$\DP[\numProjects-1,\beta] \leq \numDeleted$, and decide the instance correspondingly. The argumentation about the correctness of the outcome is similar to the constructive case.

    The first step, precomputation of the correct order of the projects, takes~$\Oh{\numProjects\log\numProjects}$ time, as we need to order the projects according to their \greedyAV score (we assume that the \greedyAV score of each project is given). The size of the dynamic programming table is~$\Oh{\numProjects \cdot \budget}$. Each cell can be computed in time~$\Oh{\log{\numDeleted}}$, and the final decision can be performed in time linear in~$\budget$. Therefore, the overall running time is~$\Oh{\numProjects\cdot\log\numProjects + \numProjects\cdot\budget\cdot\log\numDeleted + \budget}$, which is clearly polynomial in~$\numVoters+\numProjects$, even if~$\numDeleted$ is of size exponential in the number of projects~$\numProjects$.
\end{proof}

\begin{remark}
    Again, both \Cref{thm:AV:CCAC:NPh,thm:AV:CCAC:P:ifUnaryPrices} work even in the case of \greedyCost. The argumentation is still the same as in the project deletion operation; project addition does not change the ordering in which the projects are considered, and hence, all properties that were necessary for proofs of \Cref{thm:AV:CCAC:NPh,thm:AV:CCAC:P:ifUnaryPrices} are preserved even for the \greedyCost rule.
\end{remark}

\subsection{\phragmen}

For the \phragmen rule, we again show that both constructive and destructive goals are computationally intractable, even if all the projects are of the same costs. This time, the construction is almost identical to the project deletion operation. We just make all set-projects initially inactive and we adjust the number of projects we may add.

\begin{theorem}
    \label{thm:Phragmen:CCAC:NPh}
    \CCACs{\phragmen} is \NPc, even if the projects are of unit costs.
\end{theorem}
\begin{proof}
    We use the same reduction as in the case of deleting projects (cf. \Cref{thm:Phragmen:CCDC:NPh}) with just two minor changes: we set~$\numDeleted = N$ instead of~$2N$ and we make all set-projects spoiler projects. Now, we show that, with this slight tweak, the reduction indeed shows \NPhness for \CCAC{\phragmen}.

    First, we show that initially~$p$ is indeed not funded by the \phragmen rule. By the rule's definition, all voters start with their virtual bank account set to zero. Then, in time~$1$, a guard-project~$g_j$ can be funded. Moreover, since the sets of voters approving any pair of distinct guard-projects are disjoint, \emph{all} guard-projects can be funded at this point in time. Even though the project~$p$ can also be funded at this point, the rule funds~$N+1$ of guard-projects because of the tie-breaking order. This exhausts the whole budget, and therefore, the distinguished project~$p$ is never funded.

    Next, let~$\mathcal{I}$ be a \Yes-instance and~$C\subseteq \mathcal{S}$ be an exact cover of~$U$. We introduce all spoiler projects~$p_j$ such that~$S_j\in C$ and claim that the distinguished project~$p$ is now funded. Observe that the support for the set-projects is much bigger than the support for the guard-projects. Hence, the first round of funding takes place in time~$1/3N$. Now, we can fund at least one set-project~$p_j$. Since~$C$ was an exact cover of~$U$, the sets of~$A(p_j)$ and~$A(p_{j'})$ are disjoint for each pair of added set-projects~$p_j$ and~$p_{j'}$. Consequently, all added set-projects are funded simultaneously, decreasing the budget to one. After this, the sum of virtual bank accounts of all guard-projects supporters is~$0$, while the sum for supporters of~$p$ is~$1/3N$.
    The first time when a guard-project may be funded is~$1/3N + 1$; however, the sum of money of supporters of~$p$ at this point is~$1/3N + 1 > 1$. Thus, the distinguished project must have been funded before this time. Hence, the added set-projects indeed form a solution, and~$\mathcal{J}$ is a \Yes-instance.

    In the opposite direction, let~$\mathcal{J}$ be a \Yes-instance and let~$S\subseteq Q$ be a set of~$N$ spoiler projects such that if we add them to the instance, the project~$p$ gets funded. We clearly added at least one set-project since, if not, we are in the initial state where~$p$ is not funded. We create a solution~$C$ such that we add a set corresponding to a set-project in~$S$, and we claim that it is an exact cover of~$U$. For the sake of contradiction, assume that it is not the case. Then there must exist an element~$u_i\in U$ such that~$u_i\not\in C$. By the construction, there exists a set~$U_i=\{u_i^1,\ldots,u_i^N\}$ such that~$U_i\cap \bigcup_{p_j\in S} A(p_j) = \emptyset$. Let us analyze the rule in this situation. In time~$1/3N$, at least one set-project is funded. This decreases the sum of money for supporters of some guard-projects to zero. However, there are~$N$ guard-projects~$G_i = \{g_i^1,\ldots,g_i^N\}$ who keep their money since voters of~$U_i$ do not approve any added set-project. As~$\numDeleted = N$, in time~$1$, there is still some budget remaining. Now, both projects in~$G_i$ and the project~$p$ may be funded, and, due to the tie-breaking, the projects of~$G_i$ are funded before~$p$, which completely exhausts the budget. This contradicts that~$S$ is a solution. Therefore, such~$U_i$ cannot exist, and~$C$ is, therefore, an exact cover of~$U$, which finishes the proof.
\end{proof}

\begin{theorem}\label{thm:Phragmen:DCAC:NPh}
    \DCACs{\phragmen} is \NPc, even if the projects are of unit costs.
\end{theorem}
\begin{proof}
    Let~$\mathcal{I} = (U,\mathcal{S})$ be an instance of \RXthreeC. We create an instance of the \DCACs{\phragmen} as follows (see \Cref{fig:Phragmen:DCAC:NPh:construction} for an illustration of the construction). The set of standard projects consists of a \emph{guard-project}~$g$ and the distinguished project~$p$. The set of spoiler projects contains one \emph{set-project}~$p_j$ for every set~$S_j\in\mathcal{S}$. The tie-breaking is so that any other project is preferred before the distinguished project~$p$, and the cost of every project is one. The set of voters contains~$3N-1$ voters~$v_1,\ldots,v_{3N-1}$ approving only the guard-project~$g$. Next, for every element~$u_i\in U$, we create (slightly abusing the notation) a voter~$u_i$ approving the distinguished project~$p$ and a set-project~$p_j$ if and only if~$u_i\in S_j$. Additionally, every set-project~$p_j$ is associated with~$x=9N^2-3N-3$ \emph{dummy-voters}~$d_j^1,\ldots,d_j^{x}$ approving solely~$p_j$. Observe that the number of supporters for each~$p_j$ is~$x+3$, the number of voters approving~$p$ is~$3N$, while only~$3N-1$ voters approve the guard-project~$g$. To finalize the construction, we set~$\budget = N+1$ and~$\numDeleted = N$. Also, it is clear that~$p$ is initially funded, as there are only two projects, each with cost one, and the available budget is~$N+1$.

    \begin{figure}
        \centering\renewcommand{\arraystretch}{1.5}
        \begin{tabular}{c|ccc|c|c}
                  &~$p_1$ &~$\cdots$ &~$p_{3N}$ &~$g$ &~$p$ \\\hline
           ~$v_1,\ldots,v_{3N-1}$ &       &          &          & \checkmark & \\\hline
           ~$u_1$ &  & \checkmark  &          & & \checkmark  \\
           ~$u_2$ & \checkmark &          & \checkmark & & \checkmark  \\
           ~$\vdots$ &       &          &          & &~$\vdots$  \\
           ~$u_{3N}$ & \checkmark  & \checkmark &          & & \checkmark  \\\hline
           ~$d_1^1,\ldots,d_1^{x}$ & \checkmark & & & & \\
           ~$\vdots$ & &~$\ddots$ & & & \\
           ~$d_{3N}^1,\ldots,d_{3N}^{x}$ & & & \checkmark & & \\
        \end{tabular}
        \caption{An illustration of the election instance constructed in the proof of \Cref{thm:Phragmen:DCAC:NPh}.}
        \label{fig:Phragmen:DCAC:NPh:construction}
    \end{figure}

    For correctness, assume that~$\mathcal{I}$ is a \Yes-instance, and~$C\subseteq\mathcal{S}$ is an exact cover of~$U$. We add to an instance all set-projects corresponding to sets in~$C$. Since~$C$ is an exact cover, in the modified instance, each voter~$u_i$ approves exactly one set-project and the distinguished project~$p$. At the time of~$1/(x+3)$, all the set-projects and the distinguished project can be funded. Due to the tie-breaking order, the rule funds all~$N$ added set-projects. This completely exhausts the personal budget of all supporters of~$p$. The guard-project can be funded in time~$1/(3N-1)$. At this time, the sum of personal budgets of voters approving~$p$ is
    \[
        \frac{3N}{3N-1} - \frac{3N}{x+3} =
        \frac{3N}{3N-1} - \frac{3N}{9N^2 - 3N} =
        \frac{3N-1}{3N-1} = 1.
    \]
    Therefore, the distinguished project~$p$ may also be funded, but the tie-breaking order funds the guard-project~$g$ first. This exhausts the available budget, so the project~$p$ is not funded.

    In the opposite direction, let~$\mathcal{J}$ be a \Yes-instance and~$S\subseteq Q$ be a set of added projects securing that the distinguished project~$p$ is not funded. Assume that there is a voter~$u_i$ such that this voter is not approving any of the projects in~$S$, i.e., this voter approves only the distinguished project~$p$. At the latest, at the time~$t_0 = 1/x$, all the set-projects are funded by the rule because of the dummy-voters. As the supporters of the guard-project~$g$ are disjoint from the supporters of any other project, they have enough budget to fund~$g$ at time~$1/(3N-1)$. However, at this point in time, the overall budget of supporters of~$p$ is at least
    \begin{multline*}
        \frac{3N}{3N-1} - \frac{3N-1}{x}
            = \frac{3N}{3N-1} - \frac{3N-1}{9N^2-3N-3}%
            = \frac{3N(9N^2 - 3N - 3) - (3N-1)^2}{(3N-1)(9N^2-3N-3)}
            = \frac{27N^3 - 18N^2 - 3N - 1}{27N^3 - 18N^2 - 6N + 3},
    \end{multline*}
    which is clearly greater than~$1$ for any~$N \geq 2$. Therefore, the rule must fund the project before time~$1/(3N-1)$. This contradicts that~$S$ is a solution. Hence, no such~$u_i$ can exist. That is, every voter~$u_i$ approves at least one added set-project. By a simple counting argument, it must hold that each~$u_i$ approves exactly one added set-project, which directly implies that if we construct~$C$ such that we take all sets~$S_j\in\mathcal{S}$ for which the corresponding set-project~$p_j$ is in~$S$, we obtain an exact cover of~$U$.
\end{proof}

\subsection{\equalShares}

We conclude with a similar result also for the \equalShares rule. Again, control under this rule is computationally hard, even if all projects are of the same cost and the control is unweighted.

\begin{theorem}
    \label{thm:MES:CCAC:NPc}\label{thm:MES:DCAC:NPc}
    Both \CCACs{\equalShares} and \DCACs{\equalShares} are \NPc, even if the projects are of unit cost.
\end{theorem}
\begin{proof}
    Membership of both problems to \NP is clear; it suffices to guess a set of projects to remove and check whether the special project is (un)selected. To prove \NP-hardness, we reduce from \RXthreeCs. Let~$\mathcal{I}=(U,\mathcal{S})$ be an instance of the \RXthreeCs problem. We construct an instance of our problem exactly in the same way as in \Cref{thm:MES:CCDC:NPc}, but~$\mathcal{S}$-projects are removed from the constructed instance and placed in the set of projects that can be added. We also set~$\numDeleted=N$.

    One can clearly see that initially,~$d$ is winning, and~$p$ is losing. If we add a set of~$N$ projects forming an exact cover over~$U$, then~$p$ will win with~$d$ and other projects from~$d_1, \ldots, d_{2N}$. Conversely, if~$p$ wins after adding some up to~$\budget=N$ projects to the instance, then the added projects must have finished the budget of all voters from~$U'$, which is possible if and only if these~$\mathcal{S}$ projects form an exact cover over~$U$.
\end{proof}

\section{Counting Solutions}\label{sec:counting}

For performance measures based on the probability that a project wins/loses if a randomly selected set of projects is deleted or added, it is essential to efficiently determine the number of solutions for an instance. However, it turns out that all our \NPhness results, excluding \equalShares and the deletion operation, also imply~$\#\P$-hardness for respective problems. Therefore, we do not expect the existence of a significantly faster algorithm for computing such measures than a simple enumeration of all possible solutions.

\section{Experiments}

We analyze the effect of project deletions on real-world data from Pabulib~\citep{fal-fli-pet-pie-sko-sto-szu-tal:c:pabulib} and explore how different performance measures based on this operation can help with understanding and explanation of outcomes for proposers of losing projects and PB election organizers.

\paragraph{Data.} In our experiments, we analyze~$543$ approval-based PB instances from Pabulib.
We included every instance with approval ballots available as of October 2024 and contained at least one losing project.
In total, our dataset contains~$10531$ losing projects for \greedyAV,~$5771$ for \greedyCost,~$6004$ for \phragmen, and~$7412$ for \equalShares. Moreover, the largest instance consists of~$160$ projects and~$90494$ voters.

\paragraph{Experimental Setup.}
The experiments were run on computers with two AMD EPYC™ 7H12, 64-core, 2.6 GHz CPUs, and 256 GB DDR4 3200MT/s RAM. We use the same Python implementation for rules \greedyAV, \greedyCost, and \phragmen as in \citep{BoehmerFJPPSSS2024}. For \equalShares, we use our own implementation in C++ that significantly outperforms the available implementations (see \Cref{sec:MESimplementation} for details).
For each rule, every instance, and every~$\numDeleted$ projects in this instance, where~$\numDeleted \in \{1,2,3\}$, we determined the winners after deleting these~$\numDeleted$ projects. Evaluating the instances with the \phragmen rule took us 38000 computing hours (single core), and both \greedyAV and \greedyCost took 900 computing hours. Our C++ implementation of \equalShares required 5000 computing hours. The overall running time is significantly skewed by the large instances, as we need to try all~$\Oh{n^\numDeleted}$ subsets of projects for each instance (recall that significant speed-up for these rules is not possible due to \Cref{thm:MES:CCDC:NPc,thm:Phragmen:CCDC:NPh}).
For computing the optimal subset of projects to delete in the setting with either \greedyAV or \greedyCost, we use the dynamic programming approach from~\Cref{thm:AV:DCDC:P:ifUnaryPrices}, which can compute optimal control for the whole dataset in less than 51 minutes.

\subsection{How Close is a Project to Being Funded?}

\begin{figure*}[!tb]
    \centering
    \begin{minipage}[t]{.45\textwidth}
        \centering
        \includegraphics[width=0.97\linewidth, angle=0]{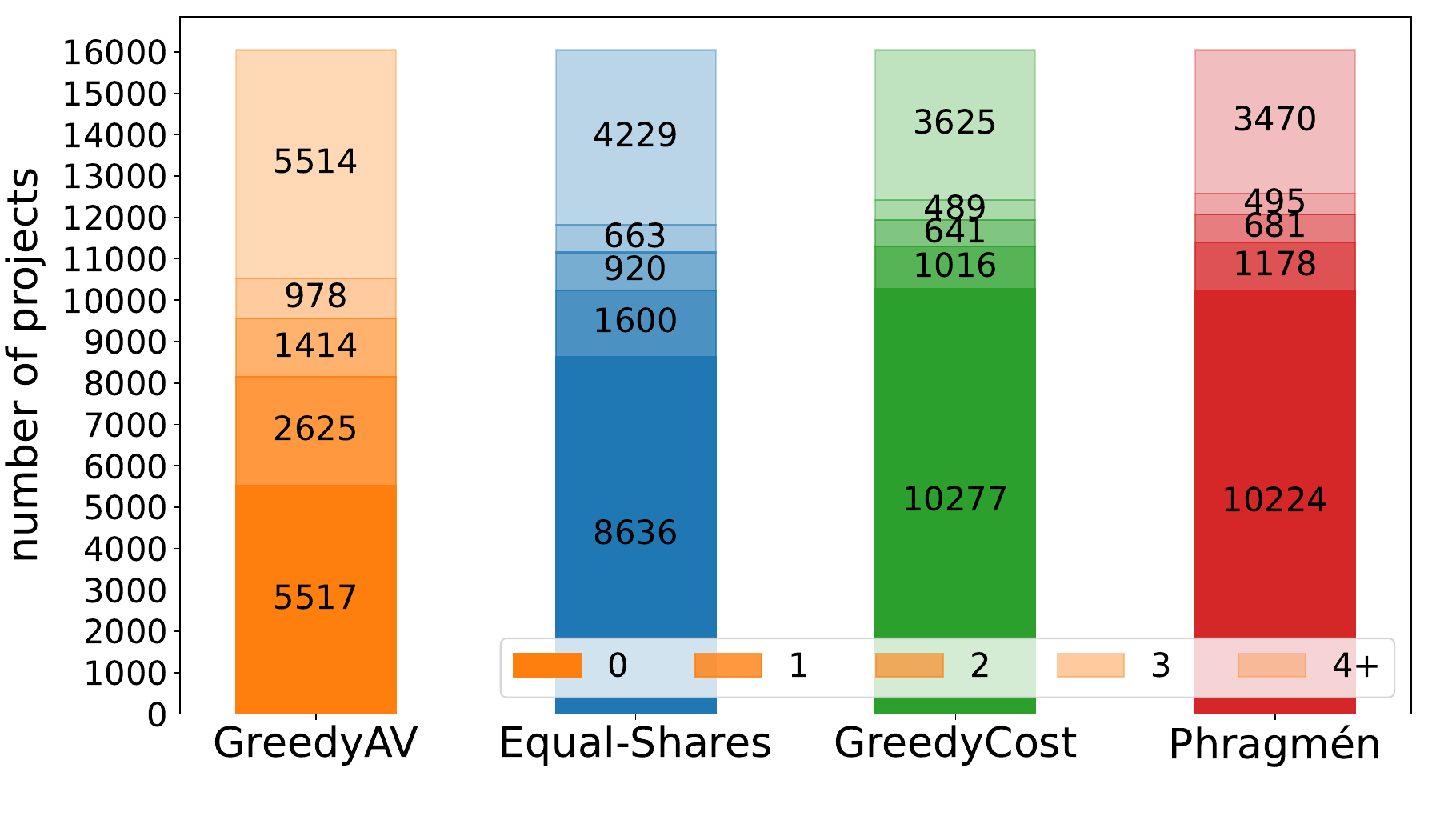}
        \caption{The distribution of projects according to their optimal control size. Each bar represents one rule and is partitioned into five parts whose sizes correspond to the number of projects that require {\small~$r\in\{0,1,2,3,4+\}$} deletions to get funded. The part with winning projects is always the darkest, and the part for projects with~$r \geq 4$ is the lightest.}
        \label{fig:dataset-overview-histogram}
    \end{minipage}%
    \quad
    \begin{minipage}[t]{0.45\textwidth}
        \centering
         \includegraphics[width=\linewidth, angle=0]{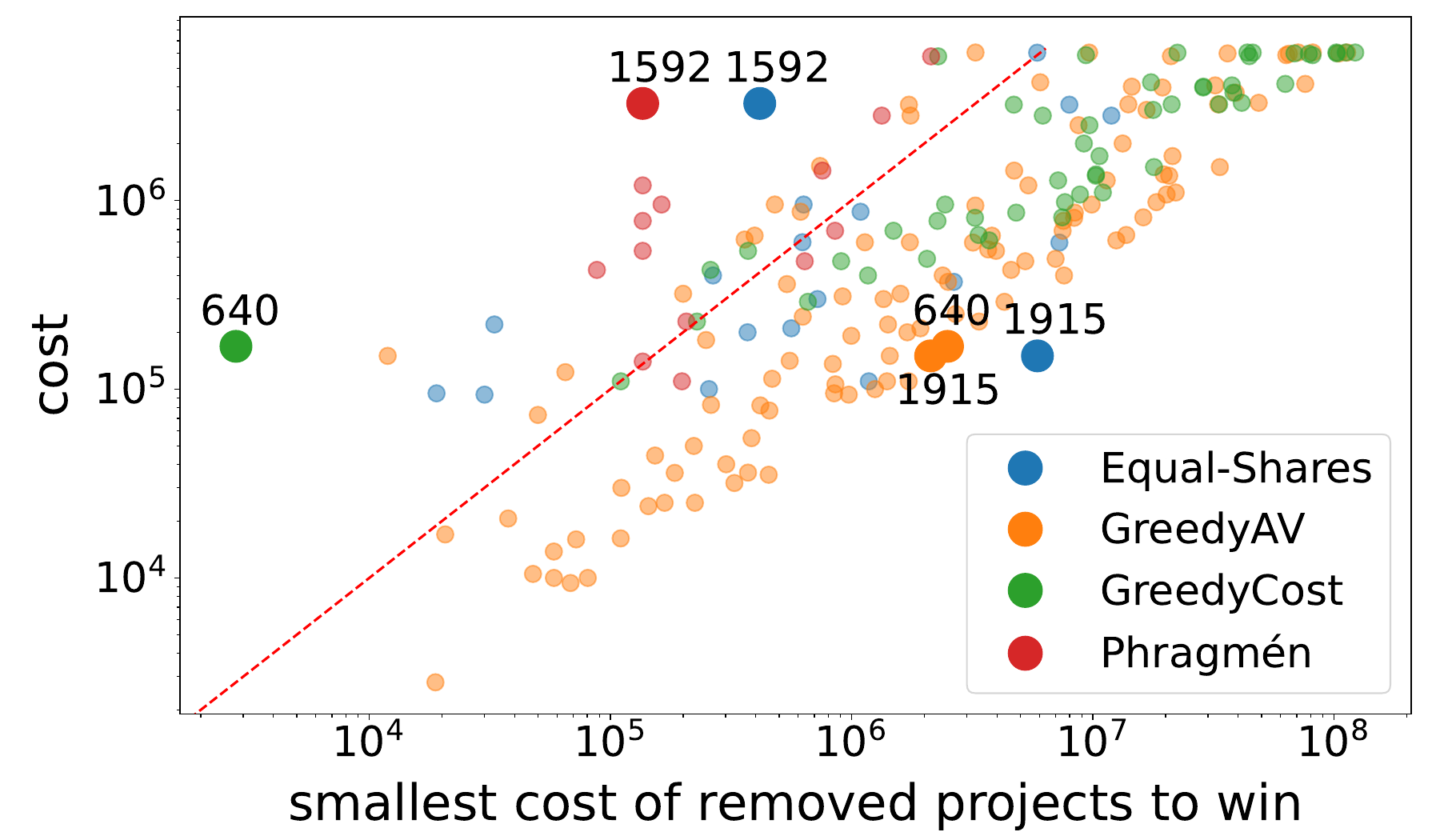}
        \caption{The ratio between the project's cost and the sum of costs of projects in the cheapest possible set of projects that makes a given project winning. The interesting projects are labeled with their name, and each project~$p$ is plotted for every rule~$f$ (unless~$p\in f(E)$). Instance: Warsaw, 2023.}
        \label{fig:cheapest-removal-set}
    \end{minipage}%
    \quad
    \begin{minipage}[t]{0.45\textwidth}
        \centering
        \includegraphics[width=\textwidth, angle=0]{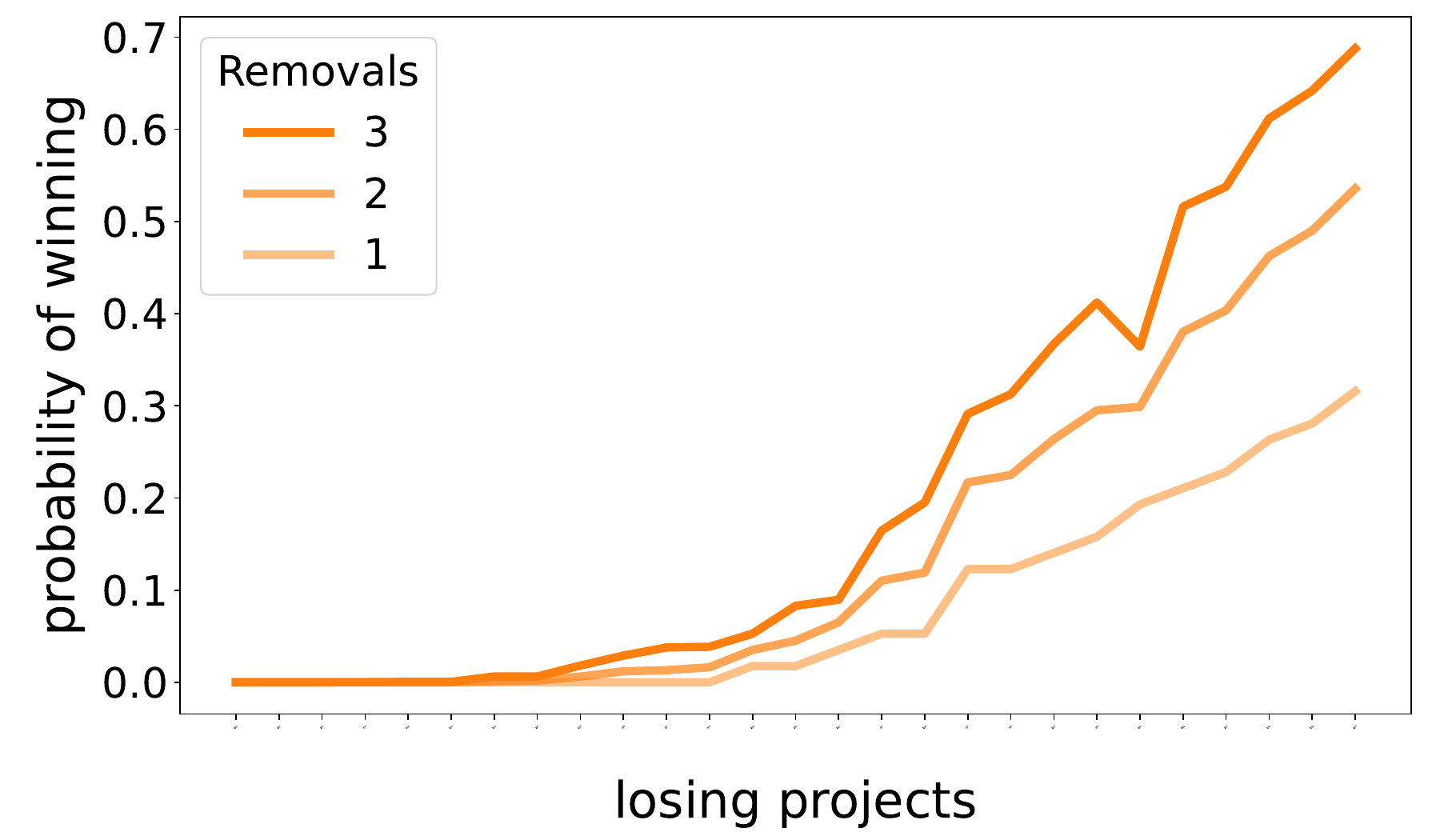}
        \caption{The probability of being funded for each initially losing project after removing~$r\in\{1,2,3\}$ projects. The projects are ordered according to their winning probability for~$r=1$. The used rule is \greedyAV, and the depicted instance is Warsaw, Ursynow, 2019.}
        \label{fig:prob-instance}
    \end{minipage}
\end{figure*}

In general, PB rules do not provide any information about the performance of proposed projects---a project is either funded or not---and there are no direct ways of measuring how close a project were to being successful.
Indeed, this is exactly what drove \citet{BoehmerFJPPSSS2024}
to initiate the study of project performance measures.
In this subsection, we propose several further such measures, based on constructive control by deleting projects.

The first, very basic, approach we suggest is to count how many other projects need to be removed from the instance to make some initially losing project~$p$ funded. In \Cref{fig:dataset-overview-histogram}, we present an overview of the results for the whole dataset. For each initially losing project~$p$, we computed the minimum number of projects we need to remove in order to make~$p$ funded. For \greedyAV, more than 47\% initially losing projects get funded after the removal of at most~$3$ projects; for the remaining rules, the value is smaller, but still significant -- 37\% for \greedyCost, 40\% for \equalShares, and 43\% for \phragmen, respectively. Another interesting information we gain from \Cref{fig:dataset-overview-histogram} is that for most of the projects for which it is enough to remove at most three projects, it is actually enough to remove only one other project. This is most evident in the case of \greedyAV.

Yet, saying that a project
is close to winning simply because it can be funded after deleting some small number of carefully selected projects is
overly simplistic. On the one hand, it is not too surprising
that a project gets selected after deleting some other,
expensive projects. On the other hand, we expect such expensive projects to be well-prepared and to not be removed for formal reasons, or due to missing some deadline.
Consequently, instead of taking the number of projects
that we need to delete to get some initially losing project
$p$ funded, we may rather seek the cheapest set of projects
(of a given cardinality) whose deletion gets~$p$ funded.
This measure indeed is much more fine-grained than our first one. To see this, we consider the results for Warszawa 2023 Citywide PB election, shown in \Cref{fig:cheapest-removal-set}: There are certain projects where the removal of expensive projects is the only way to get them funded (e.g., project~$1915$), but there are also projects that get funded after removing rather cheap projects (see, e.g., projects~$640$ and~$1592$). Note that this behavior is not very consistent among different PB rules, which is caused mostly by their underlying principles: whether the rule is more proportionality- or social-welfare-oriented.

Both above measures have the downside that they focus
on deleting exactly a particular subset of projects. However, even if we can get some initially losing project~$p$ be funded after deleting projects whose cost is~$X$, it is  possible that~$p$ would be losing after deleting some other projects, whose cost is~$2X$ (while this may seem unintuitive, it can happen due to various possible interactions among the projects and involved operation of our PB rules).
Hence, in the remainder of this section we take a more
stochastic approach and analyze the probability that a project gets funded if we remove a random subset of projects of some predefined size. We start with an overview of the whole dataset and later we analyze one specific instance, to show what information can be gained from measures based on the probability of winning.

In \Cref{fig:overall-prob}, we have a data point for every instance and plot the percentage of initially losing projects that have at least 25\% probability of being funded after the removal of three random projects. First, we see that for all the voting rules, there are instances where successful control is either unlikely for large fractions of projects or, on the contrary, where control is likely to be successful for nearly all projects. This affects more frequently instances of smaller size. Also, there is a visible difference between \greedyAV and the remaining rules: unlike in \greedyAV, for these rules, the instances are more `clustered' around certain percentage values, while for \greedyAV, the instances are more spread.

\begin{figure}[bt!]
    \centering
    \subfloat[][\greedyAV]{\includegraphics[width=0.47\linewidth, angle=0]{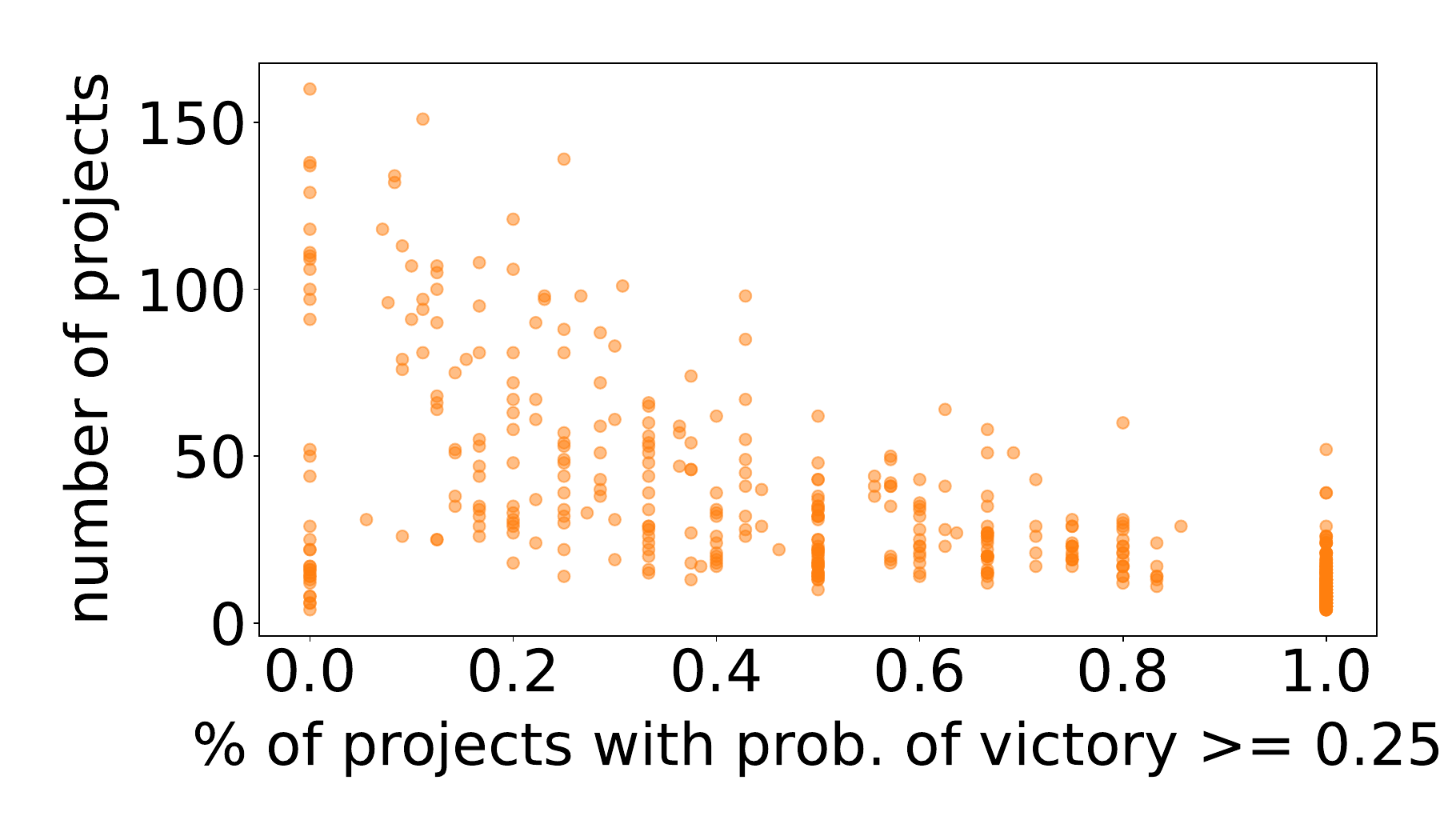}\label{fig:overall-prob:AV}}\quad
    \subfloat[][\greedyCost]{\includegraphics[width=0.47\linewidth, angle=0]{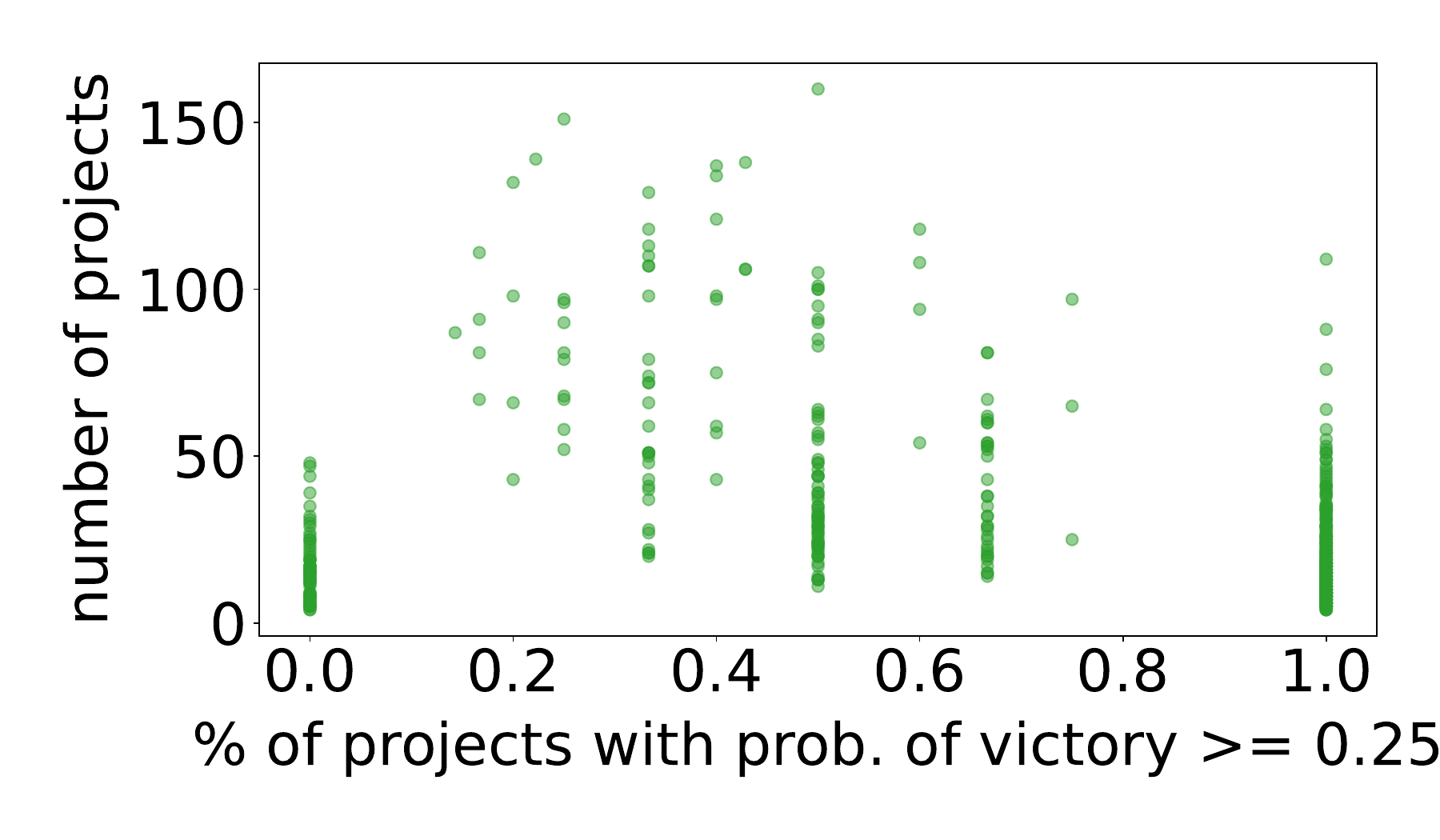}}\\
    \subfloat[][\equalShares]{\includegraphics[width=0.47\linewidth, angle=0]{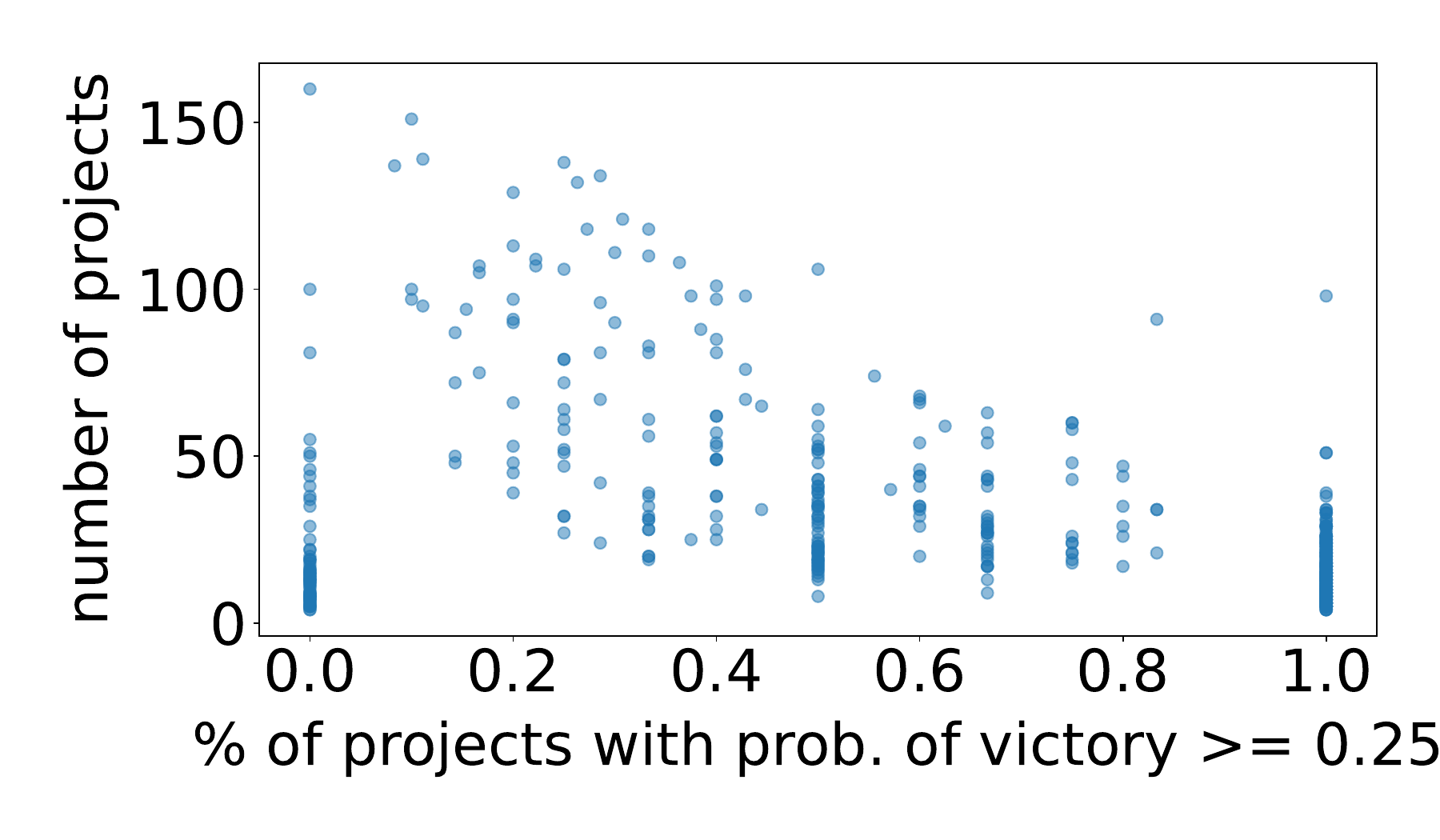}}\quad
    \subfloat[][\phragmen]{\includegraphics[width=0.47\linewidth, angle=0]{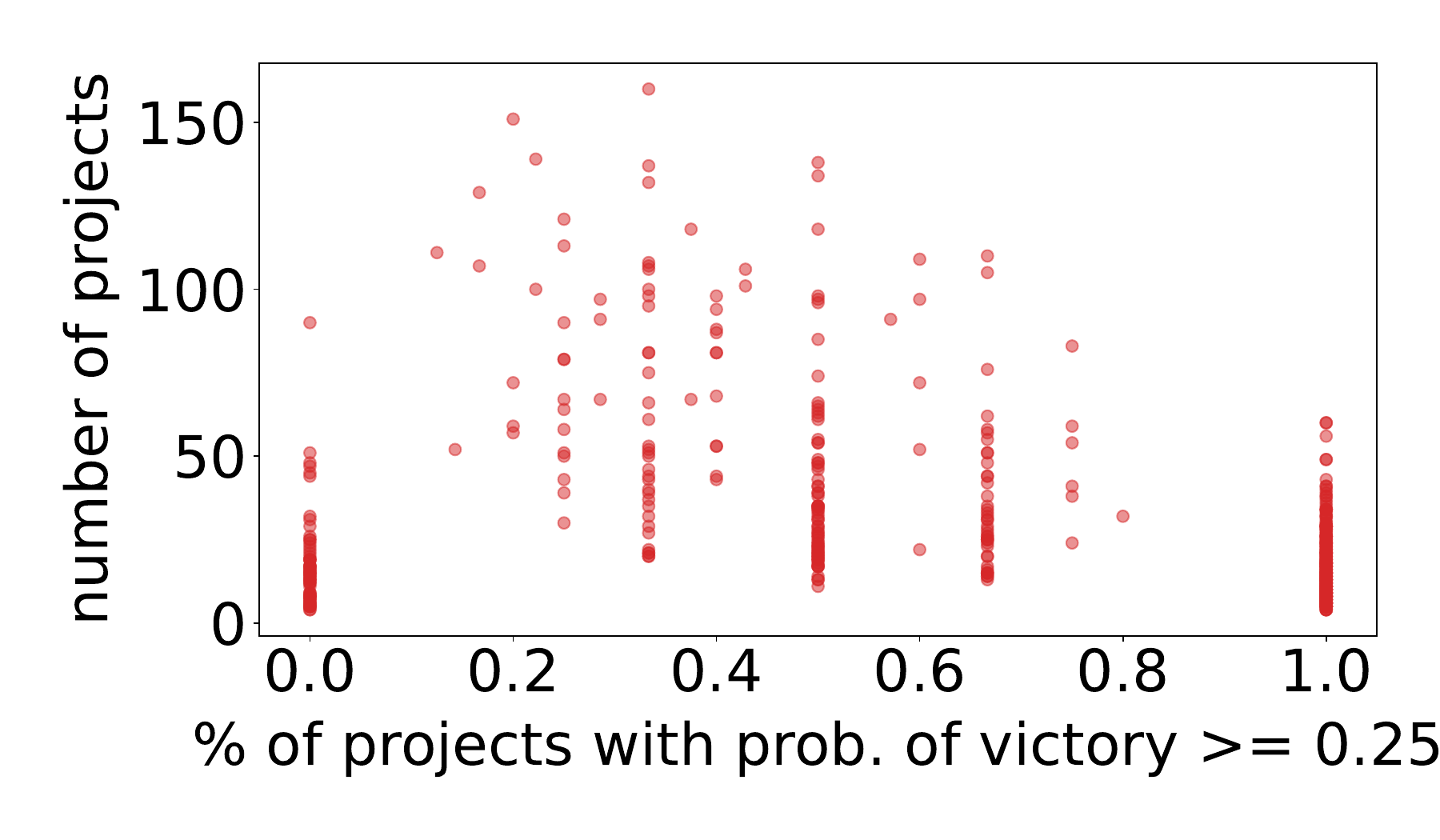}\label{fig:overall-prob:PH}}
    \caption{Each point of this plot represents a single instance. On the~$x$ axis, we have the percentage of losing projects with a probability of at least 0.25 for getting funded after the removal of~$3$ random projects. On the~$y$ axis, we have the instance size~$\numProjects$.%
    }
    \label{fig:overall-prob}
\end{figure}

Now, we focus on a specific instance. In \Cref{fig:prob-instance}, we plot for each project of the instance Warsaw, Ursynov, 2019, the probability for this project to become a winner if~$i$ projects are removed for different values of~$i$. From this plot, we can nicely distinguish between projects that are `clear losers,' meaning that their probability of winning is very close to zero regardless of the size of the removal set, and projects that are much closer to winning. This clearly demonstrates how useful such a measure of a project's strength can be.

\subsection{Who Are My Biggest Rivals?}

The performance measures introduced so far are mostly useful for election organizers. They allow us to compare projects from a `global perspective', meaning that, based on them, we can see how a losing project performed relative to other losing projects. However, for project proposers, it is very important to know why their project was not funded and what can be done better to improve their project's performance in the next election. One possibility is to identify a given project~$p$'s rivals---projects that were funded, but if they would be removed, then~$p$'s chances to get funded significantly increase. Proposers of losing projects can then analyze such rivals, learn what they did differently, and improve their projects.

We propose the following measure of rivalry. Let~$p$ be a losing project. We say that rivalry between~$p$ and some other project~$q$ is equal to the probability that~$p$ is funded after we remove~$q$ and~$r$ other random projects. In \Cref{fig:rivalry-heatmaps}, we present the results for~$r=2$, instance Warsaw, Ursynow, 2019, and different PB rules. It is not surprising that the strongest rivals are usually the projects that were initially funded. However, this is not always the case. One such example is project 1490 under the \phragmen rule (\Cref{fig:rivalry-heatmaps:PH}), for which the (initially losing) project 1432 is a much stronger rival than most of the initially winning projects. It can also be the case that for some projects, their funding relies solely on the performance of a few other projects. A very good example of this behavior is project 1806 under the \greedyAV rule (\Cref{fig:rivalry-heatmaps:AV}): unless we remove project 210, there is almost no chance that project 1806 will ever be funded. This measure also nicely complements the measures from the previous section, as, based on plots similar to \Cref{fig:rivalry-heatmaps}, we can clearly distinguish projects that are hopeless losers, projects that are somewhere in the middle, and projects that almost got in.

\begin{figure}
    \centering
    \subfloat[][\greedyAV]{\includegraphics[width=.47\linewidth]{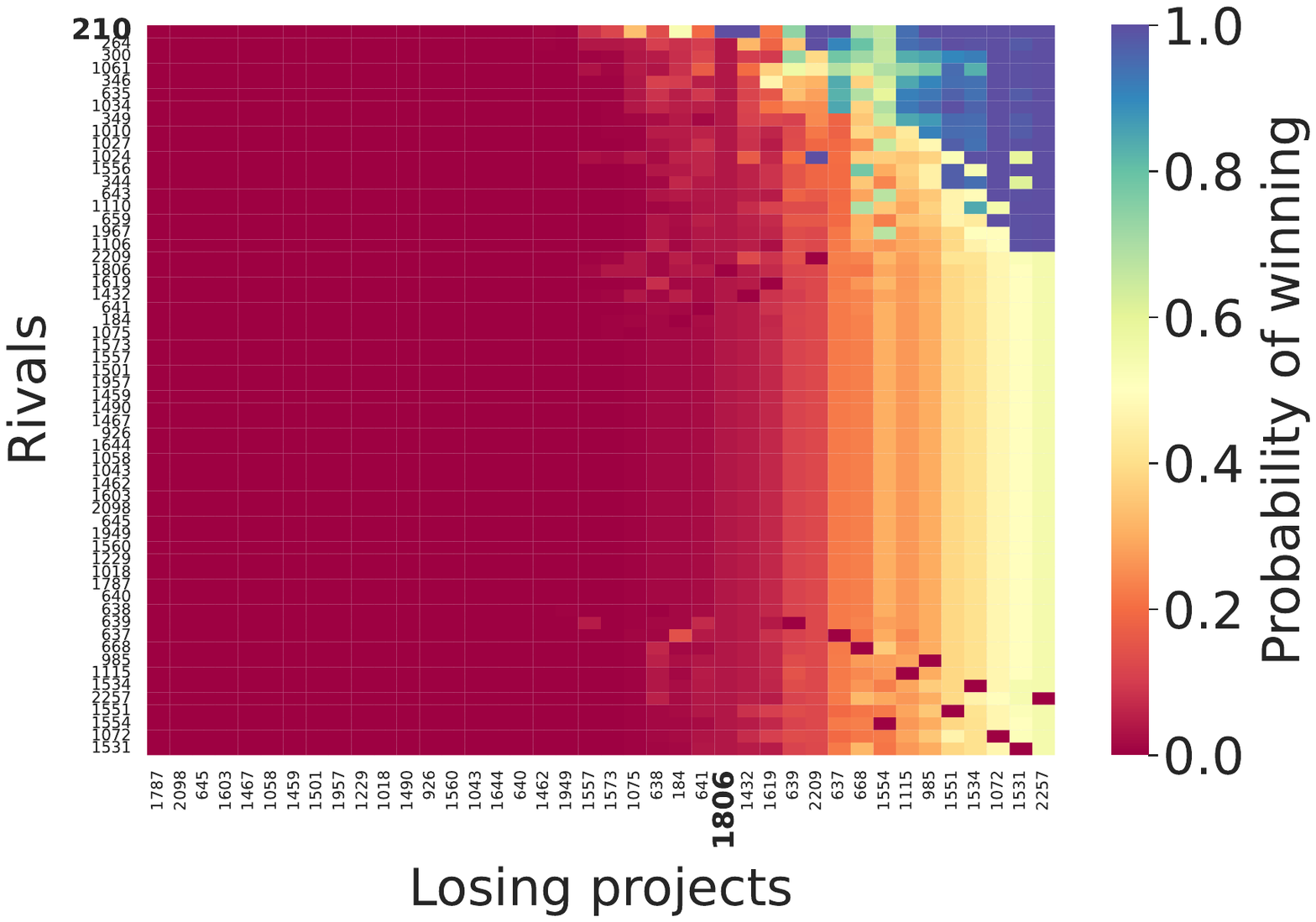}\label{fig:rivalry-heatmaps:AV}}\quad
    \subfloat[][\greedyCost]{\includegraphics[width=.47\linewidth]{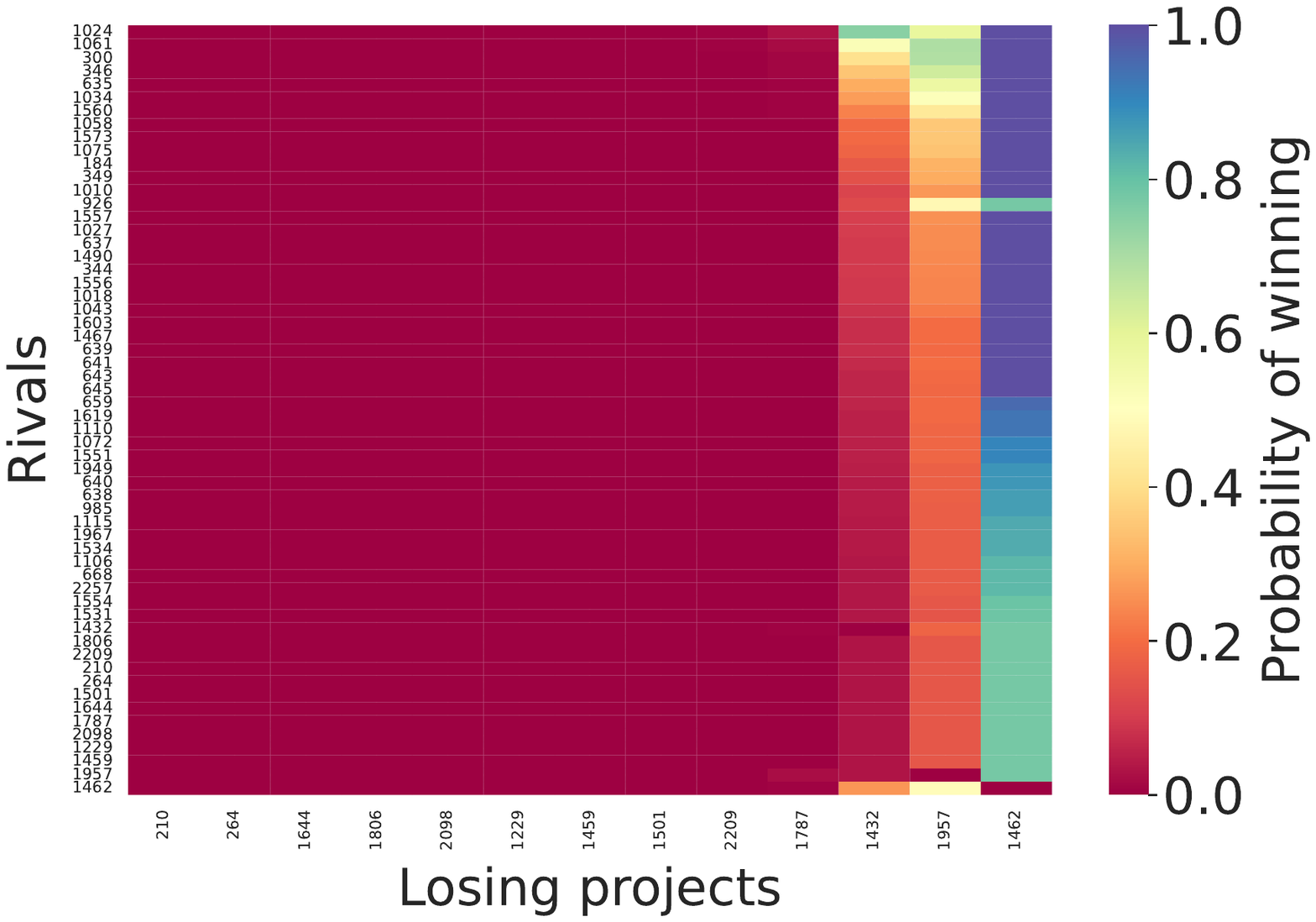}}\\
    \subfloat[][\equalShares]{\includegraphics[width=.47\linewidth]{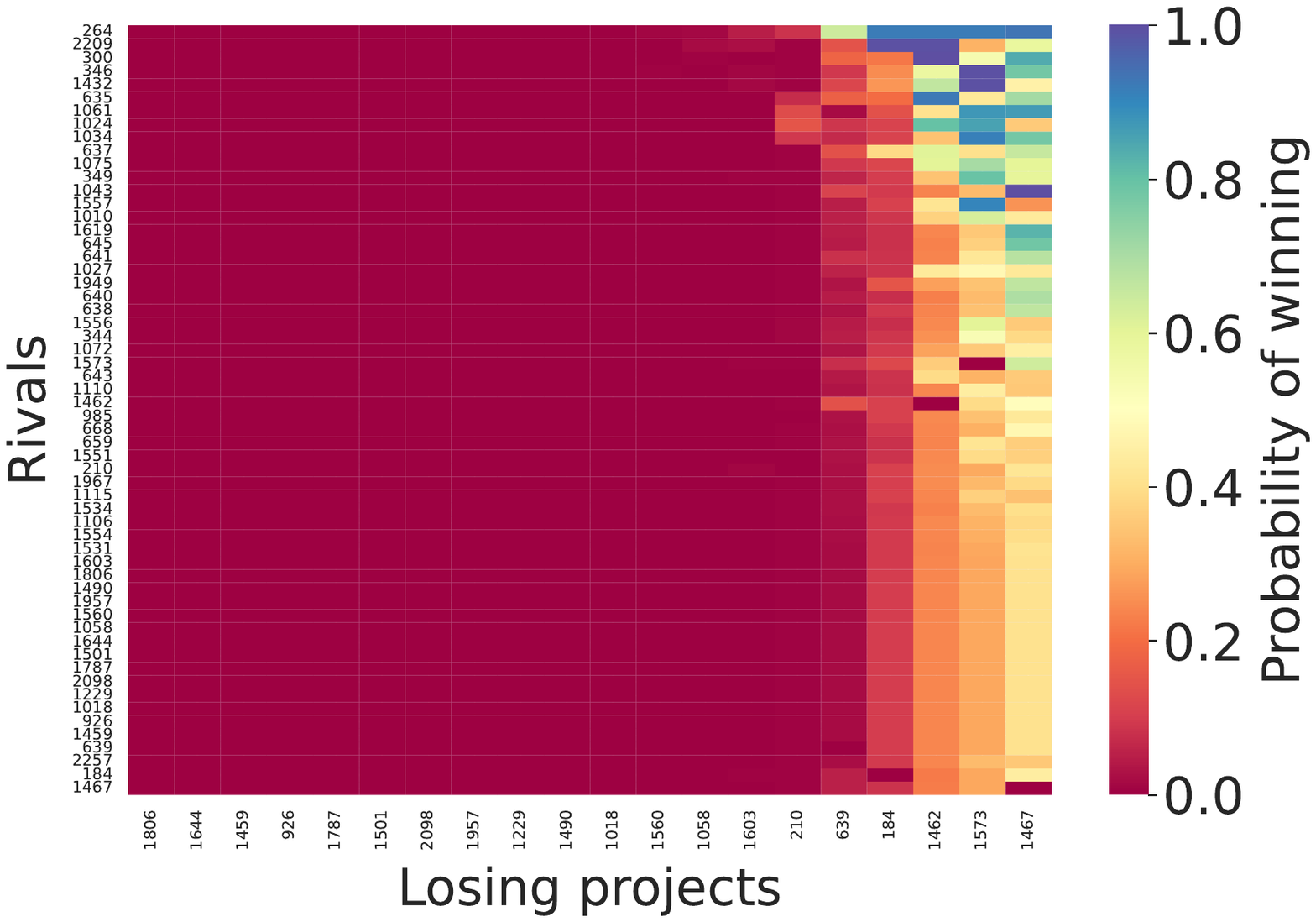}}\quad
    \subfloat[][\phragmen]{\includegraphics[width=.47\linewidth]{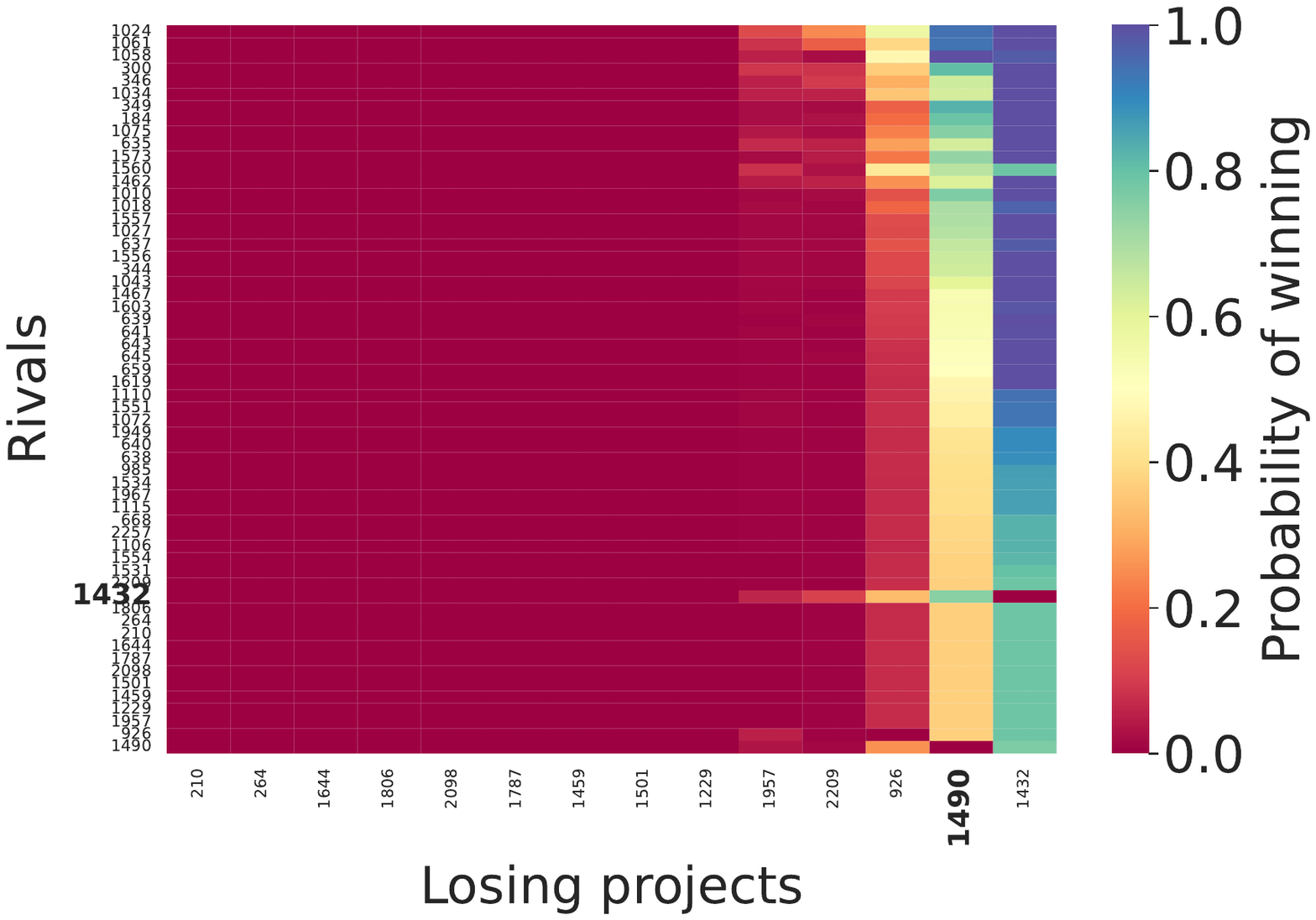}\label{fig:rivalry-heatmaps:PH}}
    \caption{Rivals for initially losing projects. For each losing project~$p$ (plotted on the~$x$-axis) and every other (not necessarily losing) project~$q$ (plotted on the~$y$-axis), we display the probability (red represents a value close to~$0$, blue represents a value close to~$1$) that~$p$ is funded when~$q$ and~$r=2$ other random projects are removed. Warsaw, Ursynow, 2019.}
    \label{fig:rivalry-heatmaps}
\end{figure}

\subsection{Comparison with Existing Measures}

In this section, we analyze the correlation between our performance measures based on the project control and the existing performance measures based on different operations. More specifically, we compare \texttt{del1}, \texttt{del2}, and \texttt{del3}---the probability of winning after the removal of at most~$1$,~$2$, and~$3$ projects, respectively---with two performance measures of \citet{BoehmerFJPPSSS2024}---\texttt{cost}, where the cost of a losing project is reduced until the project is winning, and so-called \emph{singleton} (\texttt{add}), where we keep adding voters approving only the losing project until it becomes winning. Each of these measures assigns values between $0$ (far from winning) and $1$ (close to winning). Note that \citet{BoehmerFJPPSSS2024} introduced four more performance measures, but we do not assume them, as they are anyway strongly correlated with the \texttt{add} measure, and, moreover, their calculation is very computational demanding.

The overview of the computed correlation coefficients is available in \Cref{tab:PCC}. The computed value of correlation is based on all losing projects from all our instances. We can see that for all of our rules, the correlation follows the same pattern. The metric \texttt{del1} has correlation coefficient around $0.45$ with \texttt{cost} and $0.35$ with \texttt{del}. The correlation coefficients increase slightly with increasing number of deleted projects.
The correlations coefficient are much lower than those found in \citet{BoehmerFJPPSSS2024}.
This can be explained with not many projects being able to be made winning just by removing some amount of winning projects and thus having probability of winning $0$. With more removals, more projects have non-zero probability of winning and so the metric better distinguishes the losing projects.
This behavior can be seen in \Cref{fig:correlations_add,fig:correlations_cost} as the thick line at the bottom of the plots. Additionally, mostly projects that have a good score on the other measure are seen having non-zero score with measure \texttt{del3} with the exception of \greedyAV in the comparison with \texttt{cost}.
With increasing the probability of winning the scores of the other measures also on average increase.
Together this indicates that this the measeure \texttt{delr} can be useful in distinguishing the very best losing projects.

Hence, our performance measures can clearly provide additional information that nicely extends the information package for project proposers and election organizers as introduced in \citep{BoehmerFJPPSSS2024}.

\begin{table}[bt!]
    \centering\renewcommand{\arraystretch}{1.2}
    \newcommand{\corr}[1]{\tikz{\pgfmathsetmacro{\myperc}{#1*#1*#1*#1*40}\node [transform shape, rounded corners=0pt,fill=blue!\myperc,inner sep=0, minimum width=25pt] {#1}; }}
    \begin{tabular}{lcccc}
        \toprule
         & \small\greedyAV & \small\greedyCost & \small\phragmen & \small\textsc{Eq-Shares}  \\
        \midrule
        \small corr(\texttt{cost},\texttt{add}) & \corr{0.45} & \corr{0.91} & \corr{0.80}  & \corr{0.72} \\
        \small corr(\texttt{cost},\texttt{del1}) & \corr{0.45} & \corr{0.43} & \corr{0.44}  & \corr{0.48} \\
        \small corr(\texttt{cost},\texttt{del2}) & \corr{0.48} & \corr{0.48} & \corr{0.49} & \corr{0.53}\\
        \small corr(\texttt{cost},\texttt{del3}) & \corr{0.49} & \corr{0.51} & \corr{0.52} & \corr{0.56}\\
        \small corr(\texttt{del1},\texttt{add}) & \corr{0.38} & \corr{0.38} & \corr{0.29} & \corr{0.30}\\
        \small corr(\texttt{del2},\texttt{add}) & \corr{0.42} & \corr{0.43} & \corr{0.32}& \corr{0.33} \\
        \small corr(\texttt{del3},\texttt{add}) & \corr{0.43} & \corr{0.45} & \corr{0.33} & \corr{0.34}\\
        \bottomrule
    \end{tabular}
    \caption{Pearson correlation coefficients between our performance measures \texttt{del$r$}---the probability of winning after deleting up to~$r$ projects---and two measures introduced by \protect\citet{BoehmerFJPPSSS2024}---adding singleton votes (\texttt{add}) and reducing the cost of a project (\texttt{cost}). Values near~$1$ mean a stronger correlation.}
    \label{tab:PCC}
\end{table}

\begin{figure}
    \centering
    \subfloat[][\greedyAV]{\includegraphics[width=.47\linewidth]{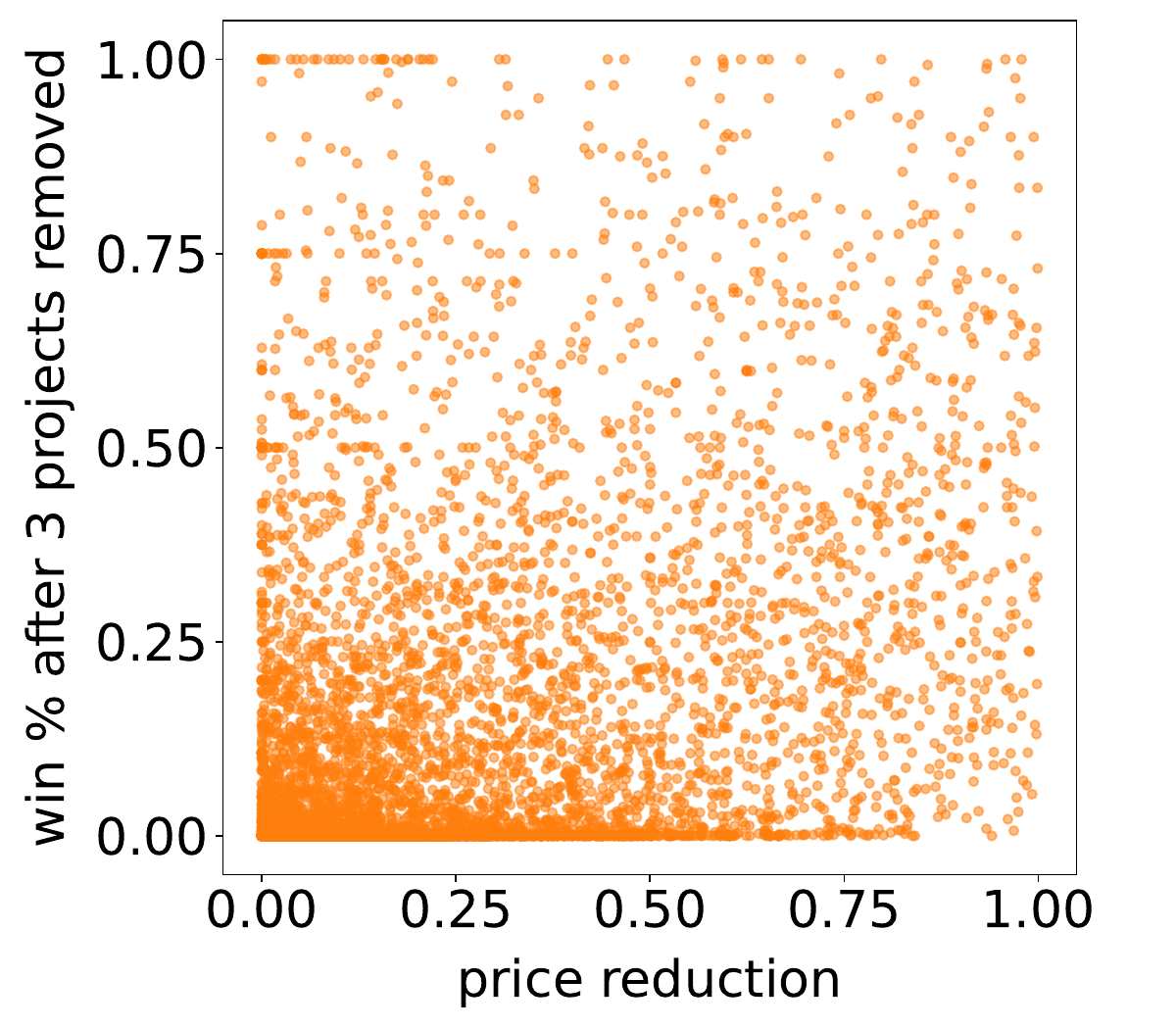}}\quad
    \subfloat[][\greedyCost]{\includegraphics[width=.47\linewidth]{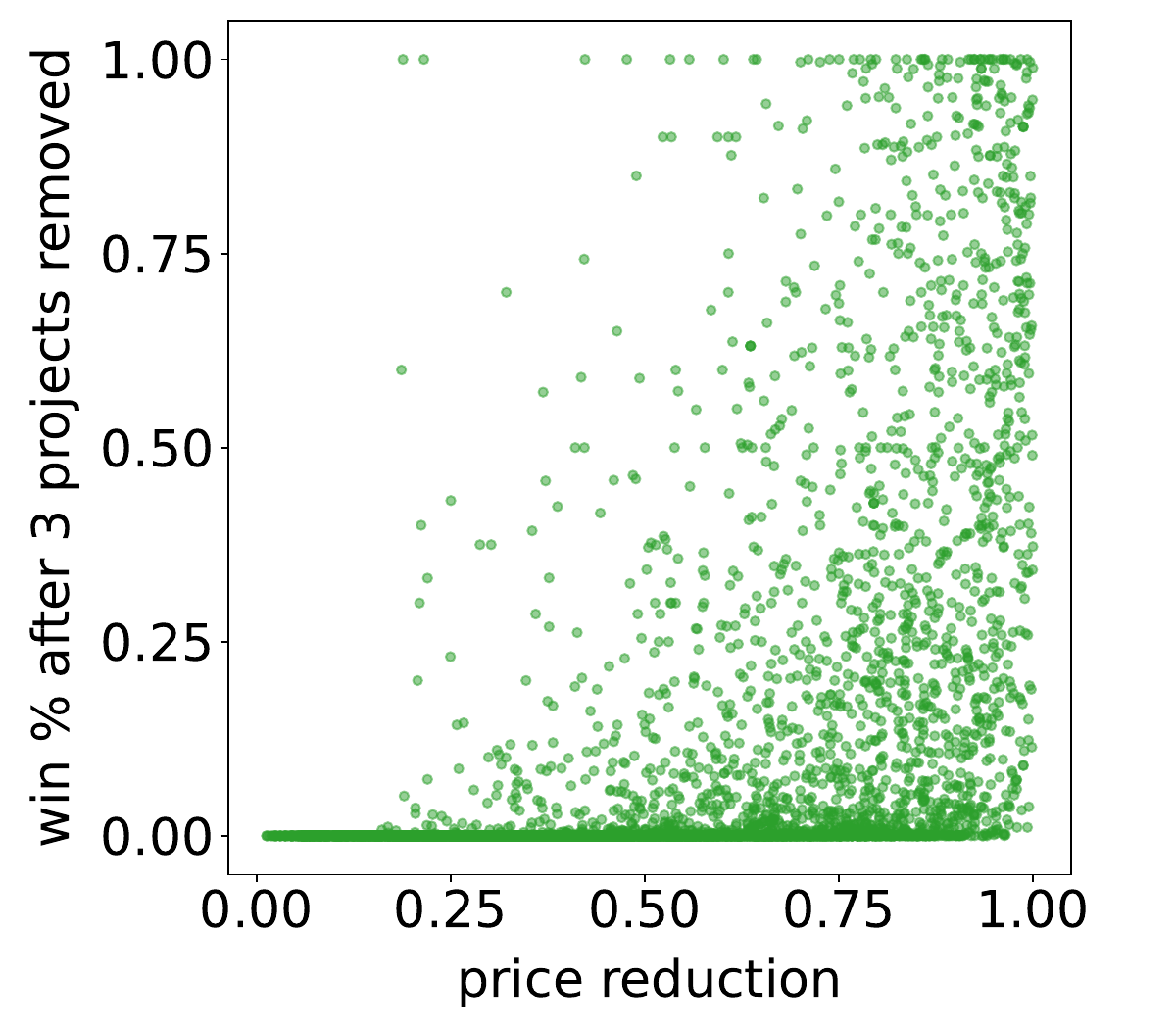}}\\
    \subfloat[][\equalShares]{\includegraphics[width=.47\linewidth]{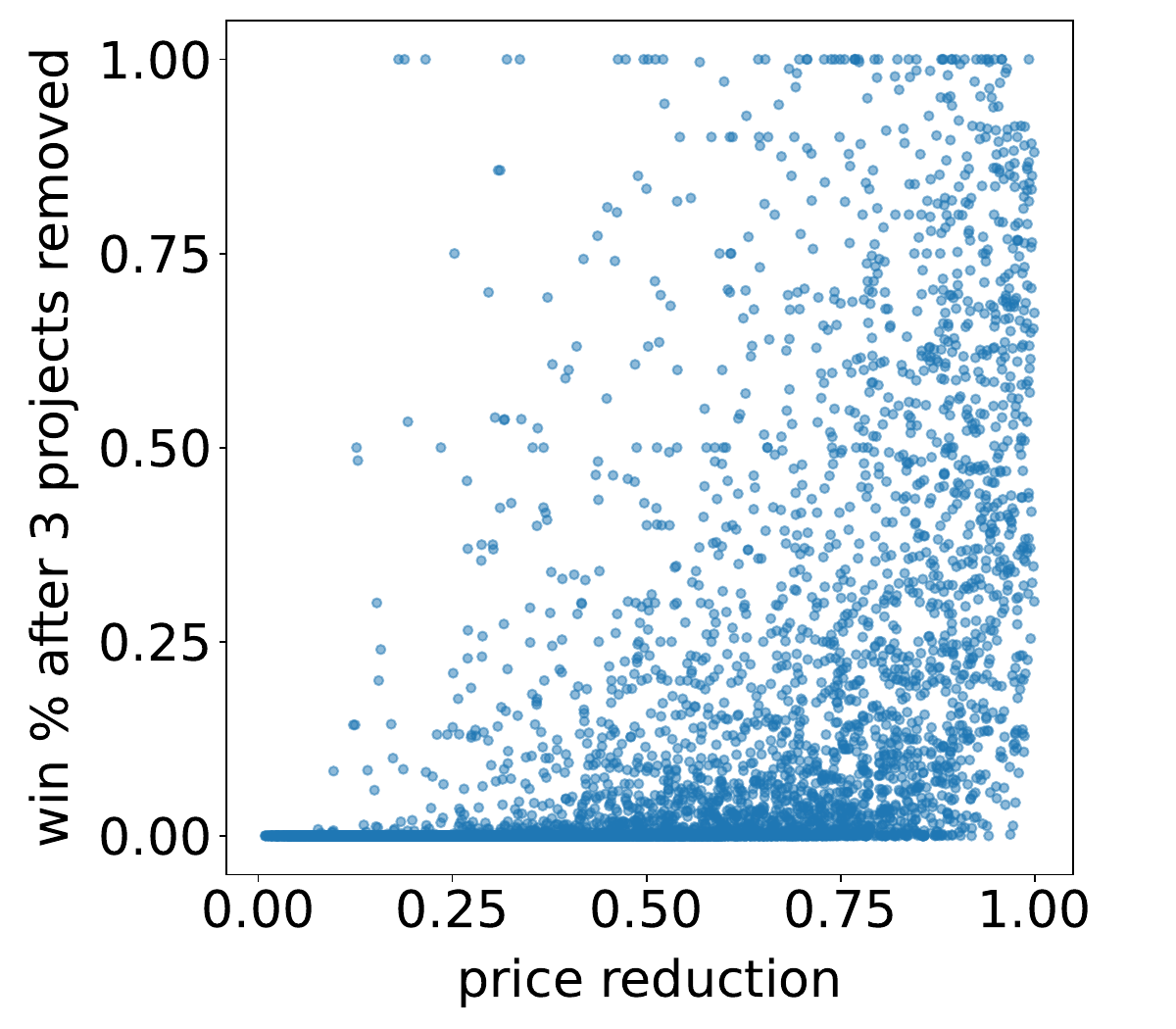}}\quad
    \subfloat[][\phragmen]{\includegraphics[width=.47\linewidth]{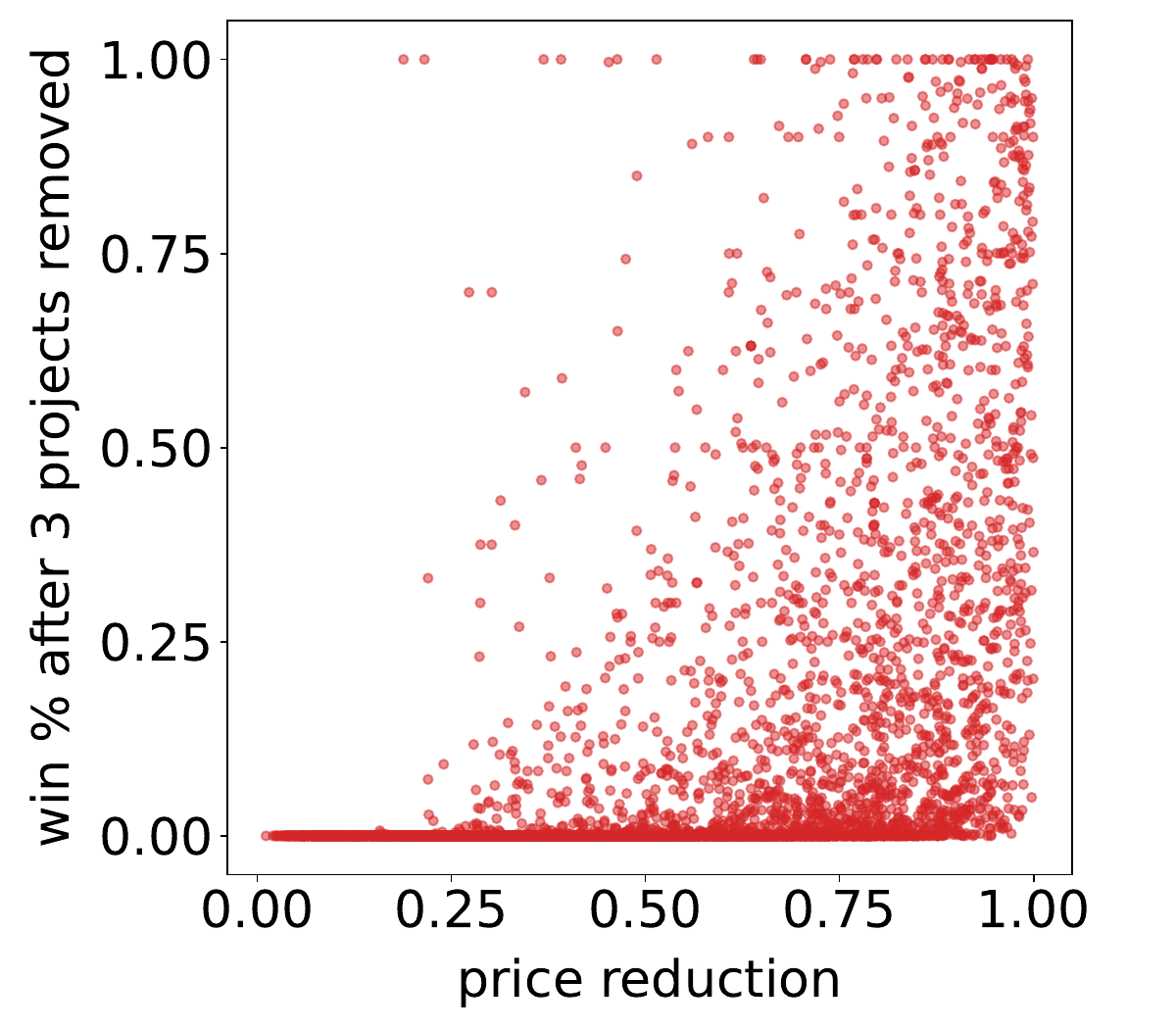}}
    \caption{Correlation plots of measures \texttt{cost} and \texttt{del3} for each rule. Each dot represents one losing project.}
    \label{fig:correlations_cost}
\end{figure}

\begin{figure}
    \centering
    \subfloat[][\greedyAV]{\includegraphics[width=.47\linewidth]{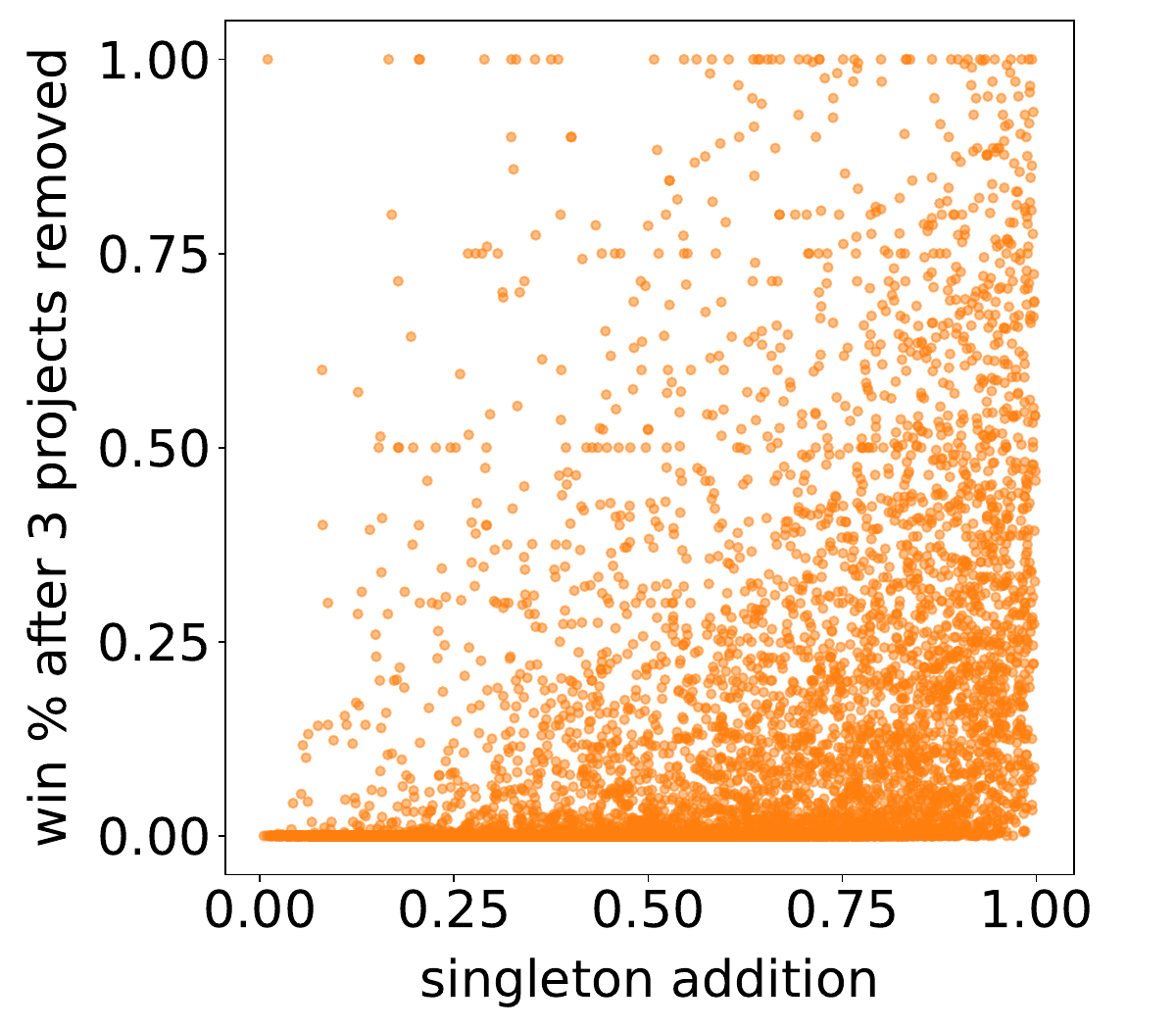}}\quad
    \subfloat[][\greedyCost]{\includegraphics[width=.47\linewidth]{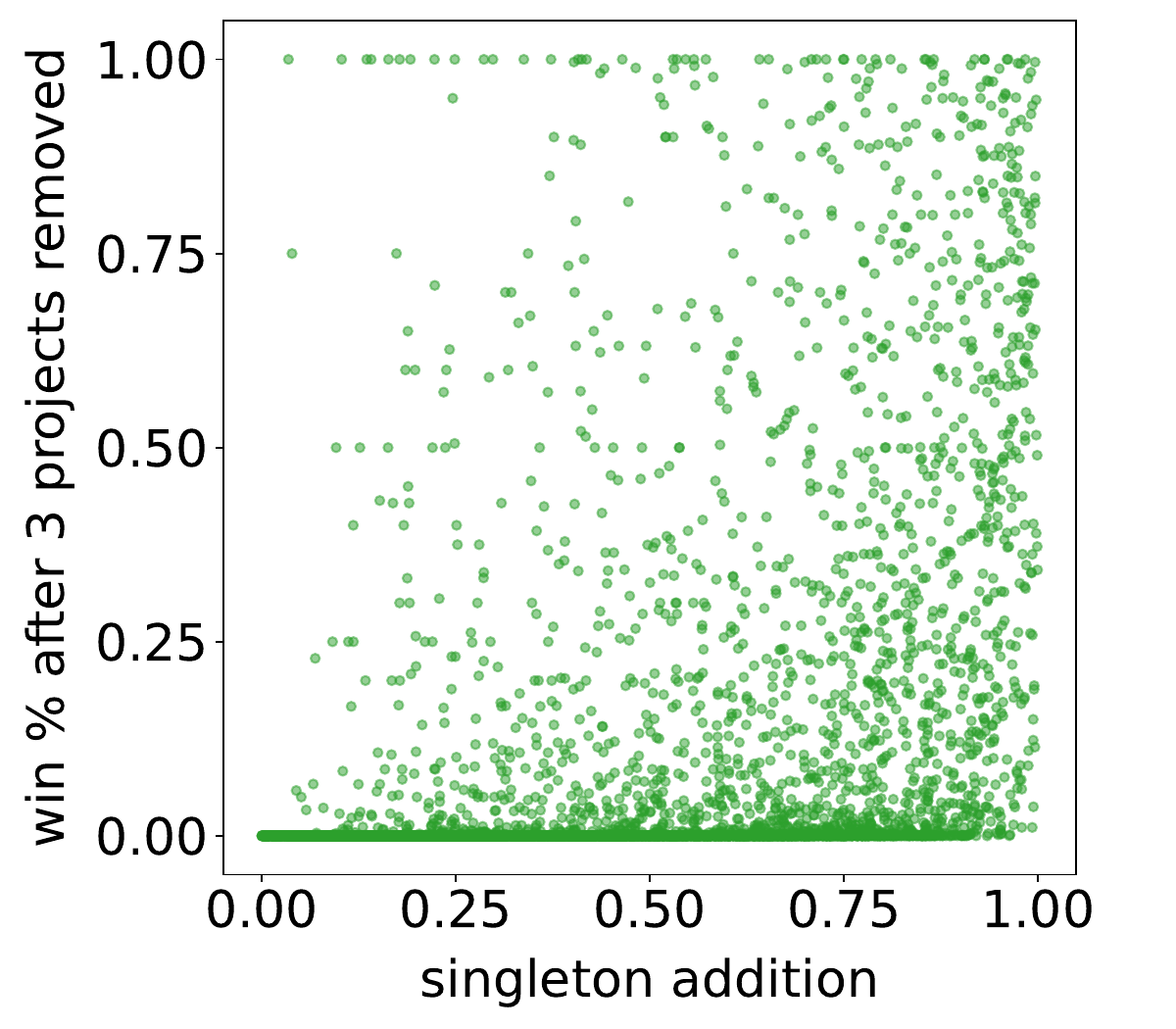}}\\
    \subfloat[][\equalShares]{\includegraphics[width=.47\linewidth]{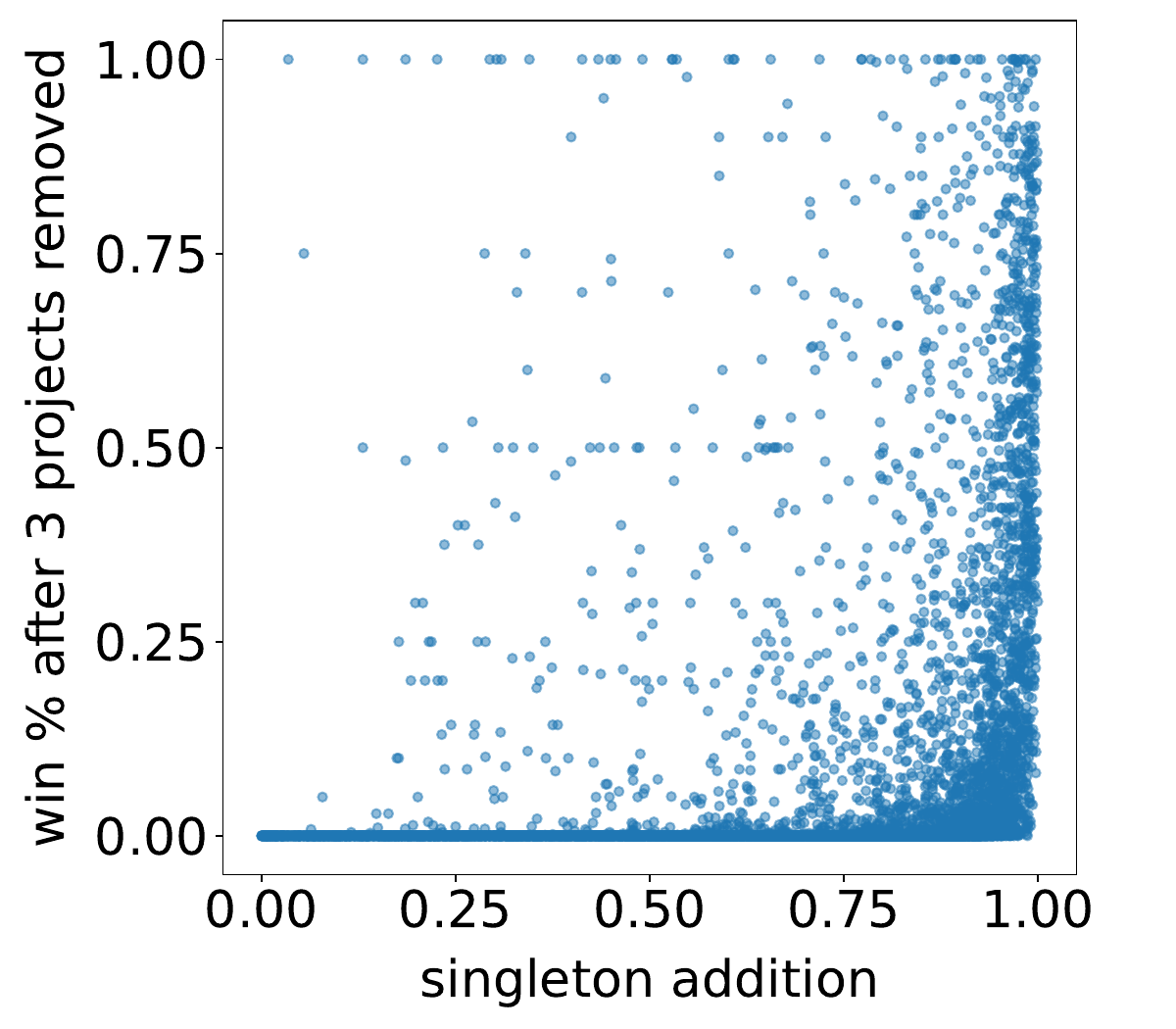}}\quad
    \subfloat[][\phragmen]{\includegraphics[width=.47\linewidth]{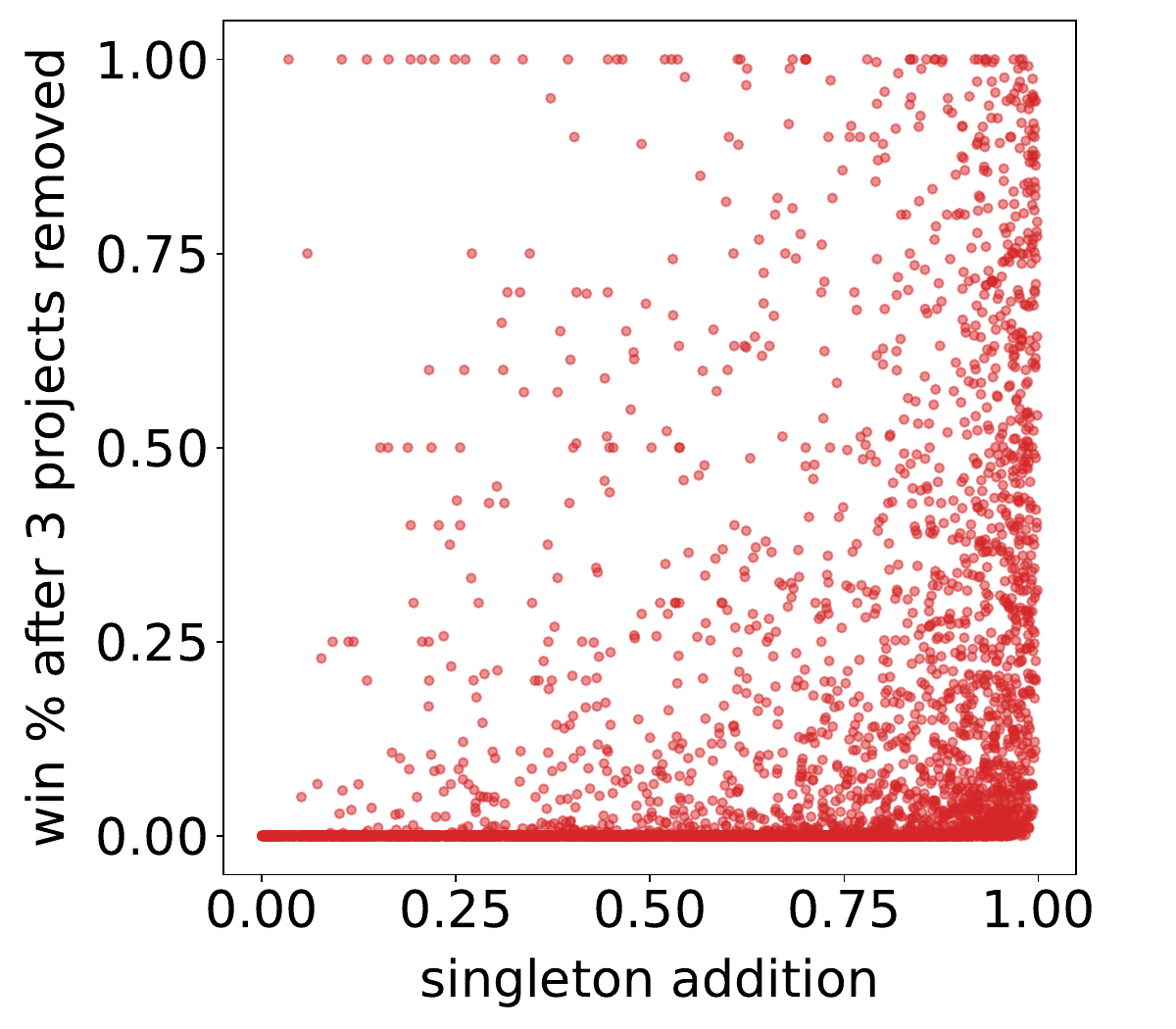}}
    \caption{Correlation plots of measures \texttt{add} and \texttt{del3} for each rule. Each dot represents one losing project.}
    \label{fig:correlations_add}
\end{figure}

\subsection{Discussion}

In our experiments, we clearly demonstrated how useful performance measures based on project control can be for all interested parties. We can use these measures to compare different projects and provide the election organizers with information on which projects were very close to being funded. This is very useful because, in practice, cities often try to discuss popular losing projects and fund them from an increased or completely separate budget (recall that PB rules are often not exhaustive). Moreover, our rivalry-based measures may help project proposers identify which other projects prevented their projects from being funded. If such strong rivals share some similarities with the losing project, the proposer can try to learn from the rival and significantly improve their losing project to make it funded in the next installation of elections. Finally, we compared our measures with the measures already existing in the literature and demonstrated that our measures nicely supplement and extend the information package for losing projects, which can increase transparency, popularity, and perceived legitimacy of participatory budgeting elections.

\FloatBarrier

\section{Conclusions}

We have studied the computational complexity of candidate control in participatory budgeting elections with two different goals---constructive and destructive---and two different control operations---project deletion and project addition. Regardless of the rule used, the associated decision problems are computationally intractable; however, for the \greedyAV and \greedyCost, which are used in a majority of real-life PB elections, a polynomial-time algorithm exists for natural and realistic restrictions of the problem. Based on our theoretical results, we suggest multiple project performance measures based on the control operations. In the series of experiments on real-life PB instances, we demonstrate how useful insights these measures can provide for losing projects.

\appendix

\begin{acks}
This project has received funding from the European Research Council (ERC) under the European Union’s Horizon 2020 research and innovation programme (grant agreement No 101002854).
ŁJ’s research presented in this paper is supported in part from the funds assigned by Polish Ministry of Science and Technology to AGH University.
This work was co-funded by the European Union under the project Robotics and advanced industrial production (reg. no. CZ.02.01.01/00/22\_008/0004590).
This work was supported by the Ministry of Education, Youth and Sports of the Czech Republic through the e-INFRA CZ (ID:90254), project OPEN-30-16 and OPEN-34-18.
JP and ŠS acknowledge the additional support of the Grant Agency of the CTU in Prague, grant No. SGS23/205/OHK3/3T/18.

\begin{center}
    \vspace{0.5cm}
    \includegraphics[width=5cm]{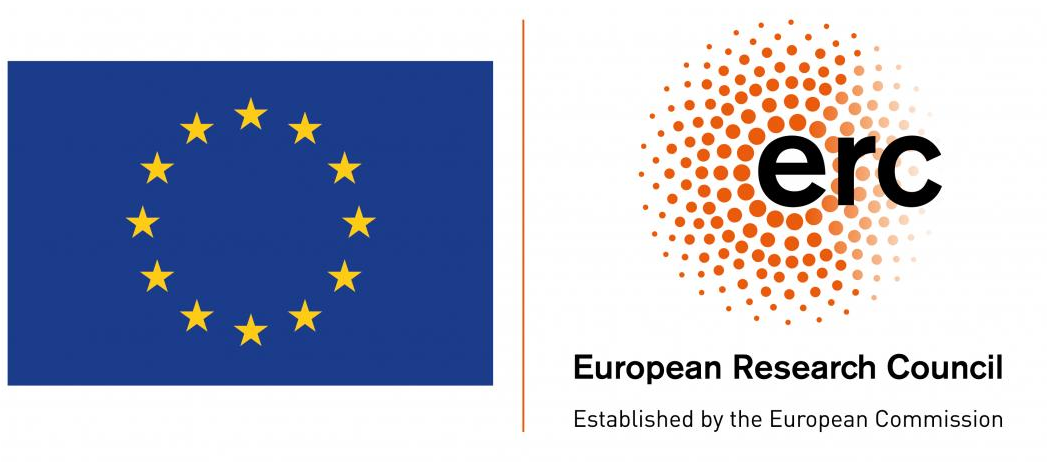}
\end{center}
\end{acks}

\bibliographystyle{ACM-Reference-Format}
\bibliography{references}

\appendix

\section{A New Implementation of \equalShares}\label{sec:MESimplementation}

The new implementation of MES has two major improvements in comparison to the previous implementations in Python \cite{BoehmerFJPPSSS2024} and JavaScript \footnote{https://github.com/equalshares/equalshares-compute-tool}. It is written in C++, utilizing STL data structures, which is on average much faster programming language than Python and JavaScript. The second improvement is using cache friendly data structures like C++ vectors instead of previously used hash tables.

We performed tests to measure the improvements introduced in our implementation of MES. The tests were made on the same instances as the instances evaluation \footnote{https://pabulib.org/?hash=66e85250d9ca8}  PB instances from Pabulib \citep{fal-fli-pet-pie-sko-sto-szu-tal:c:pabulib}. For each PB instance, we run each implementation of MES 10 times and take the average of measure. The running time is taken inside the implementation (e.g. we are not measuring the time the python interpreter needs to parse the python program).

\begin{table}[h!]
        \centering
        \def\arraystretch{1.5}
        \begin{tabular}{p{3.5cm}M{2cm}M{2cm}M{2cm}M{2cm}}
                \toprule
                 comparison & min & max & avg & median \\
                \midrule
                 js/cpp & 2.65 & 139 & 59.1 & 59.9 \\
                 python/cpp & 0.787 & 605 & 94 & 67.5 \\
                \bottomrule
        \end{tabular}
        \caption{Comparison of speedup between our C++ MES implementation and  different implementations of MES.}
        \label{tab:speedups}
\end{table}

The table \Cref{tab:speedups} shows that we achieved a great speedup. C++ implementations are always much faster than the Python ones and, in most cases, are also significantly faster than JavaScript implementations.

\begin{figure}
    \centering
    \includegraphics[width=0.5\linewidth]{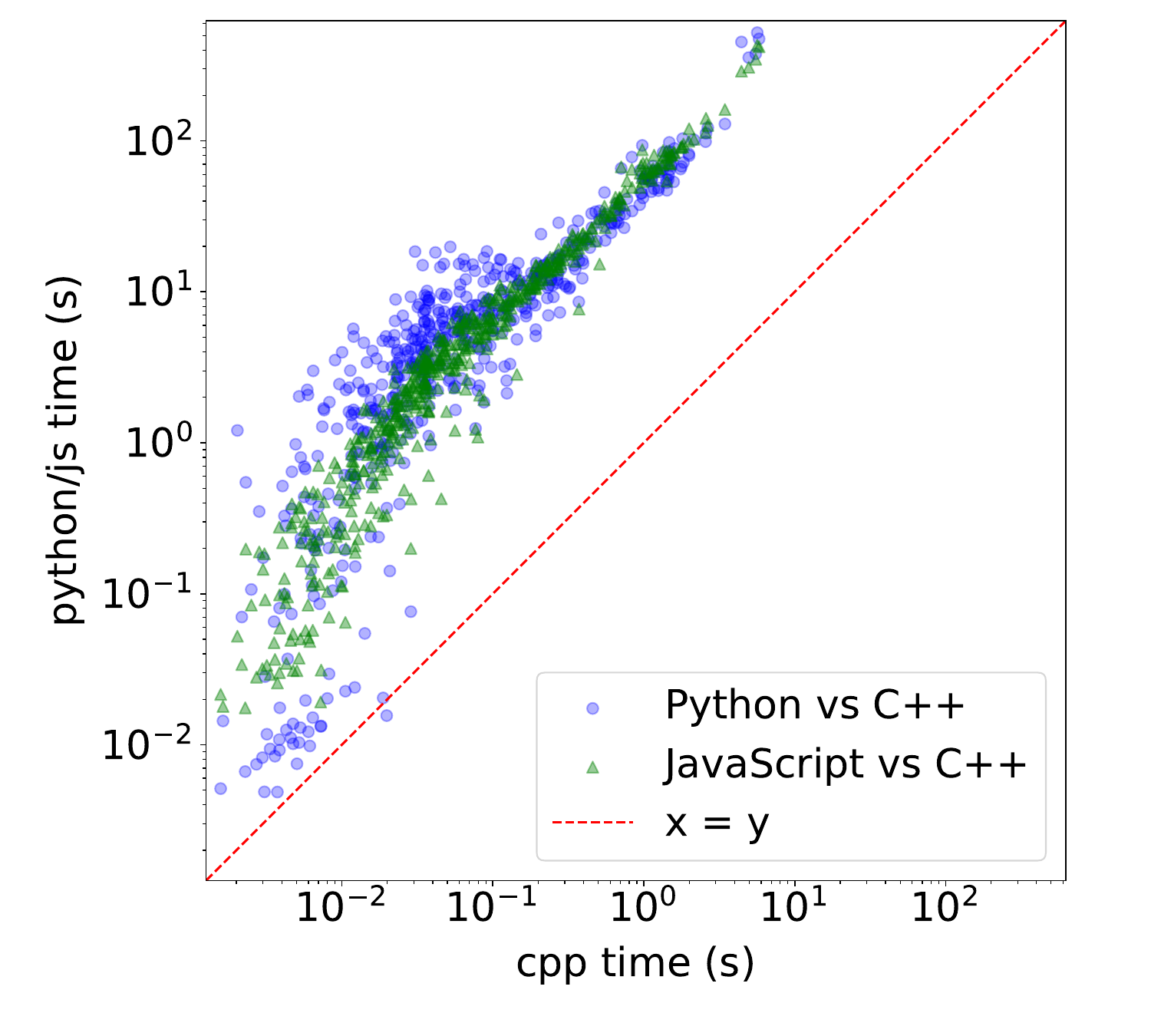}
    \caption{Caption}
    \label{fig:enter-label}
\end{figure}

\begin{figure}
    \centering
    \includegraphics[width=0.5\linewidth]{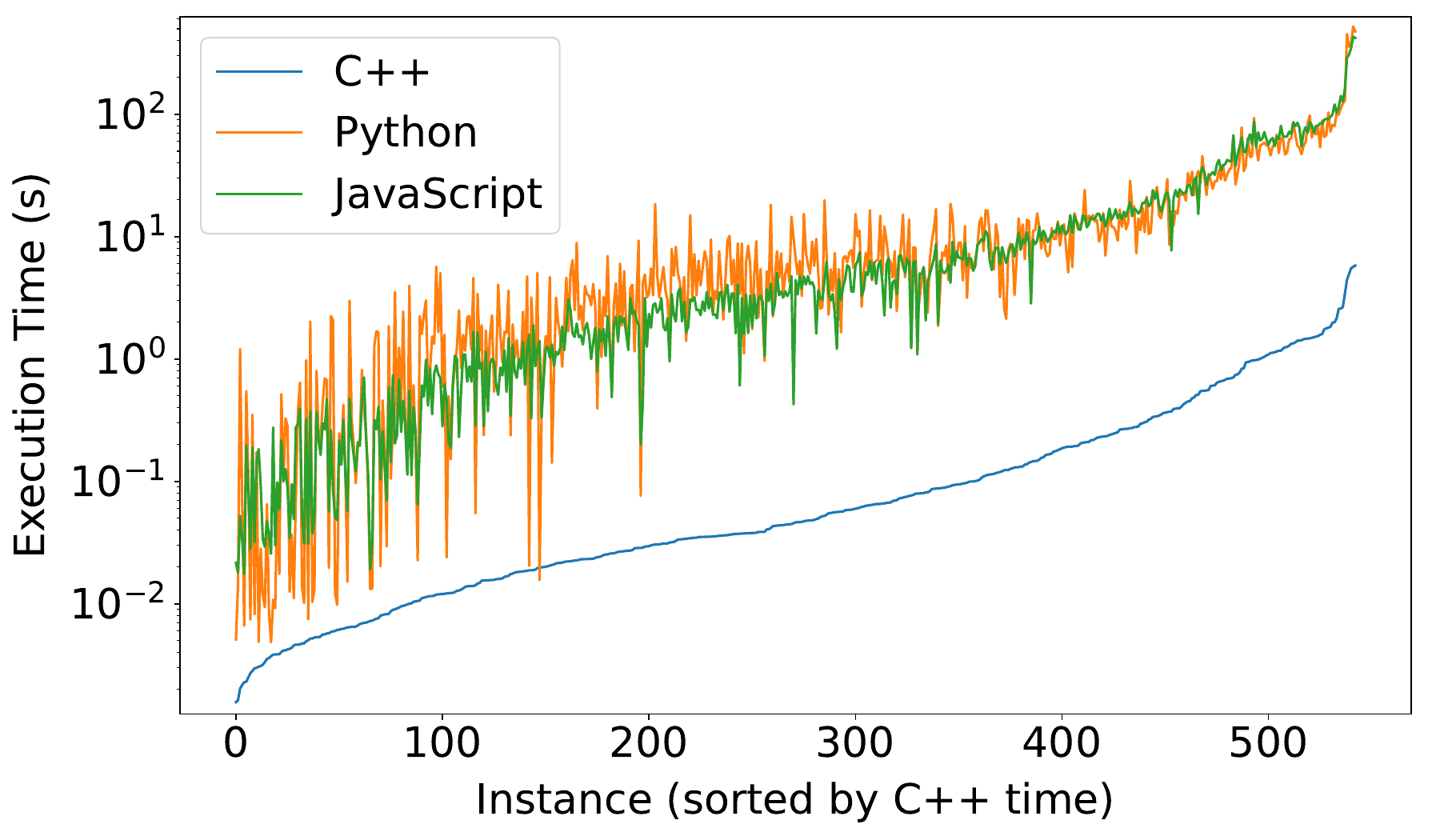}
    \caption{Caption}
    \label{fig:enter-label}
\end{figure}

\begin{figure}
    \centering
    \includegraphics[width=0.5\linewidth]{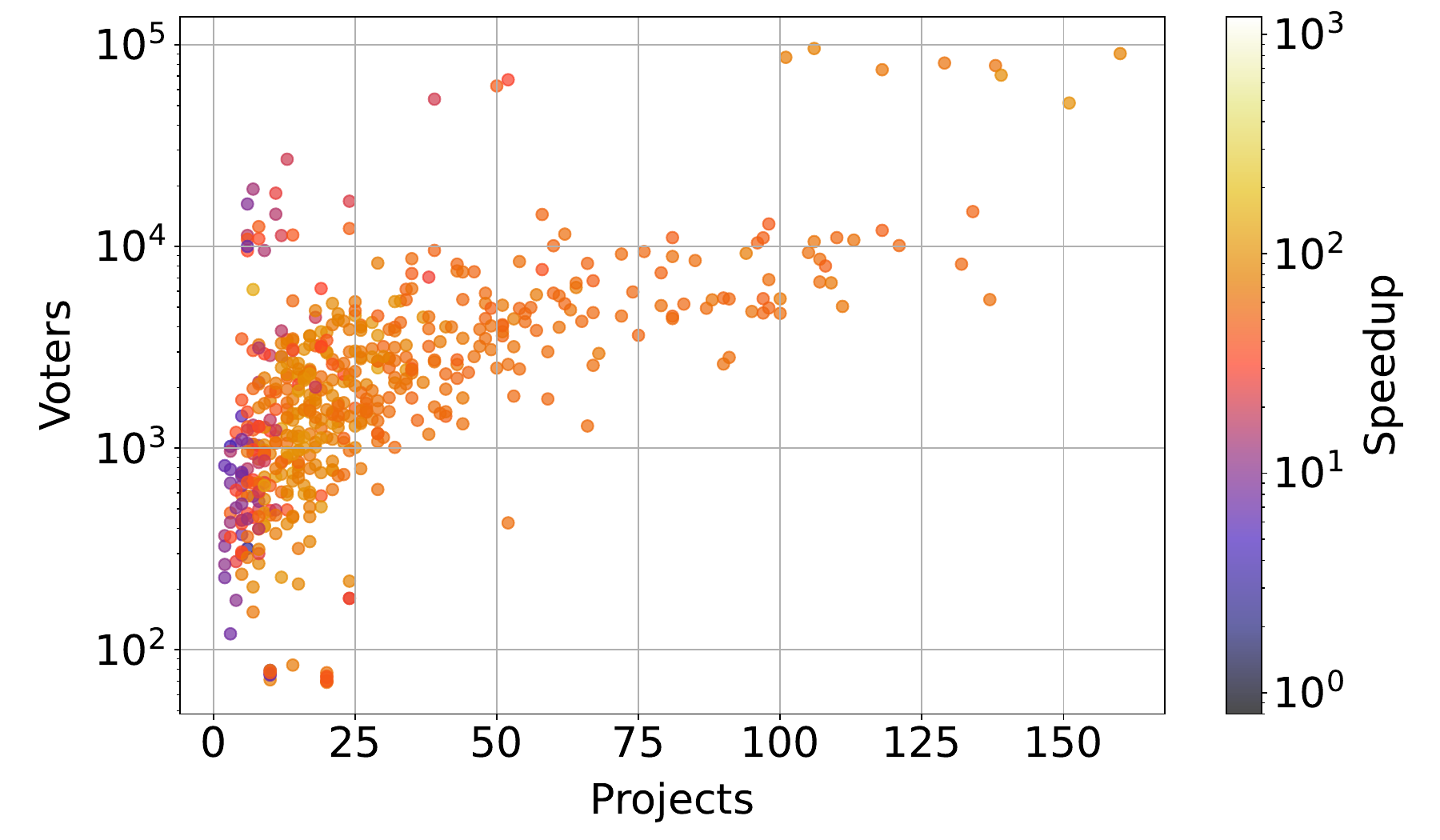}
    \caption{Caption}
    \label{fig:enter-label}
\end{figure}

\begin{figure}
    \centering
    \includegraphics[width=0.5\linewidth]{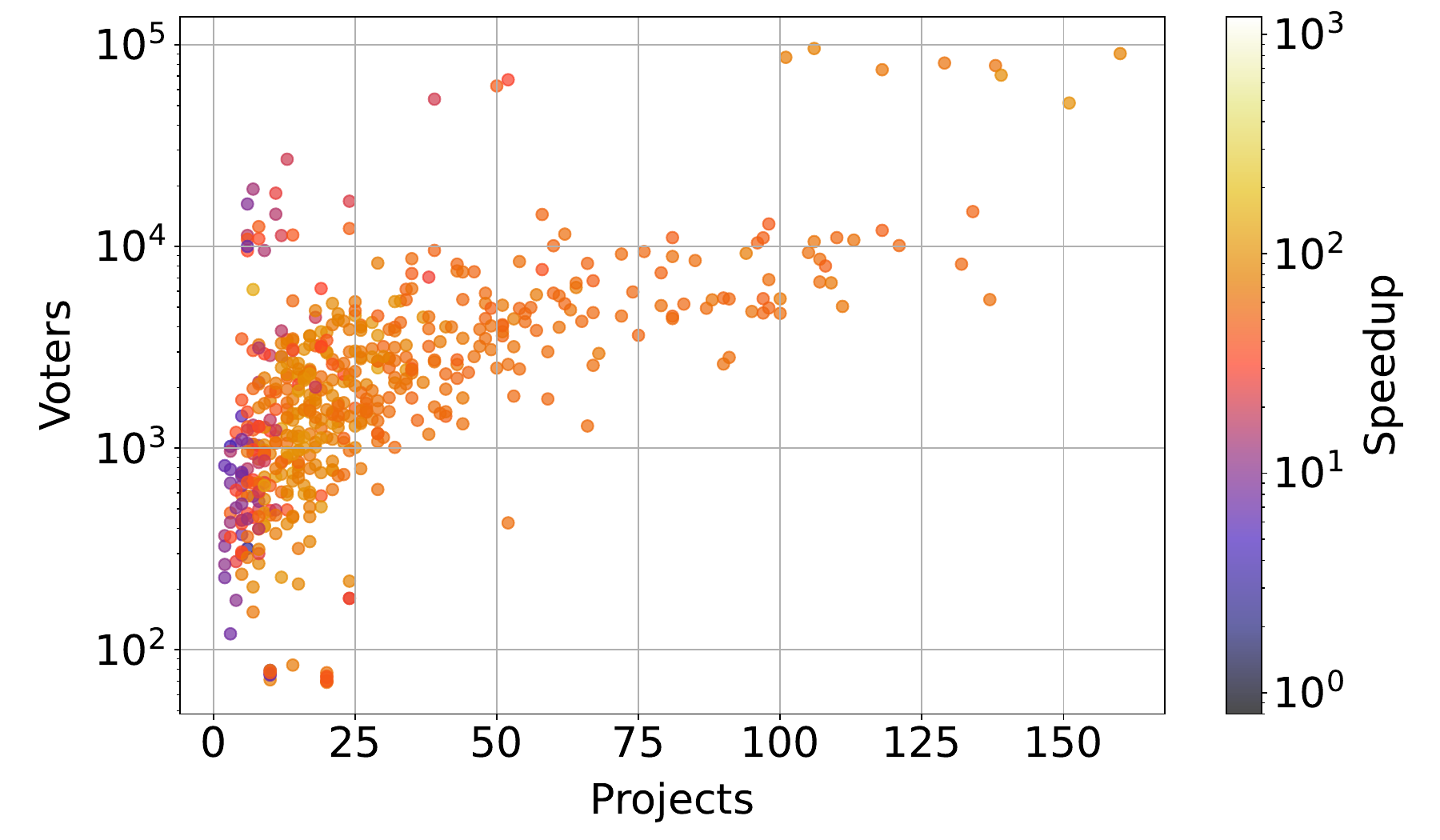}
    \caption{Caption}
    \label{fig:enter-label}
\end{figure}

\section{Ethical Statement}

Even though our paper's goal is to improve the explainability and transparency of participatory budgeting outcomes, voting control is traditionally understood as malicious and highly undesired behavior. As such, one might object that our work could increase awareness of possible manipulation in PB elections. We want to stress that manipulations based on our performance measures are implausible: Our measures can be computed only after the elections have ended and we have complete (and anonymized) information about the whole instance. That is, the potential knowledge based on our measures can be used, if at all, to manipulate the next installation of PB elections.
However, the new instance will most likely be different, since some projects have already been funded, new projects will be proposed, and, in particular, voters' preferences may change over time.

\end{document}